\documentclass{article}
\usepackage[utf8]{inputenc}
\usepackage{jheppub}
\usepackage{enumitem}
\usepackage{amsmath,amssymb,amsthm,amscd,mathrsfs,bm}
\usepackage{mathtools}
\usepackage{standalone}

\usepackage{xcolor}
\usepackage{url}
\usepackage{listings}
\usepackage{tabularx}
\usepackage{diagbox}
\usepackage{tikz-cd}
\usepackage{nicematrix}
\usepackage{pifont}
\usepackage{subfig}
\usepackage{booktabs}
\usepackage{float}

\newcommand{\ee}{\text{e}}

\usepackage{tikz,textcomp}
\usetikzlibrary{decorations.pathreplacing,calc,fadings,fit,shapes,arrows,positioning,chains,matrix}

\newtheorem{theorem}{Theorem}[section]

\newtheorem{definition}[theorem]{Definition}

\DeclareMathOperator{\disc}{disc}

\newcommand{\IZ}{\mathbb{Z}}
\newcommand{\IC}{\mathbb{C}}
\newcommand{\IP}{\mathbb{P}}
\newcommand{\IN}{\mathbb{N}}
\newcommand{\IR}{\mathbb{R}}
\newcommand{\IQ}{\mathbb{Q}}

\newcommand{\IF}{\mathbb{F}}

\newcommand{\id}{\mathrm{id}}
\newcommand{\cA}{{\cal A}}
\newcommand{\cO}{{\cal O}}
\newcommand{\cD}{{\cal D}}
\newcommand{\cT}{{\cal T}}
\newcommand{\cE}{{\cal E}}
\newcommand{\cM}{{\cal M}}
\newcommand{\cS}{{\cal S}}

\newcommand{\cF}{{\cal F}}
\newcommand{\cZ}{{\cal Z}}

\newcommand{\cL}{{\cal L}}
\newcommand{\cN}{{\cal N}}
\newcommand{\cC}{{\cal C}}

\newcommand{\cH}{{\cal H}}

\newcommand{\sB}{\mathscr{B}}
\newcommand{\sL}{\mathscr{L}}
\newcommand{\sS}{\mathscr{S}}
\newcommand{\sR}{\mathscr{R}}
\newcommand{\sA}{\mathscr{A}}

\newcommand{\n}{\mathfrak{n}}

\newcommand{\e}{\mathrm e}

\newcommand{\ii}{\mathrm{i}}
\newcommand{\dd}{\mathrm{d}}
\newcommand{\omegag}{{\omega_{\gamma}}}

\newcommand{\re}{\mathrm{Re}}

\newcommand{\LR}{\mathrm{LR}}

\newcommand{\FgredN}{\cF_g^{\mathrm{red},\,N}}
\newcommand{\Fgredn}[1]{\cF_g^{\mathrm{red,{#1}}}}
\newcommand{\Fgred}{\cF_g^{\mathrm{red}}}

\newcommand{\cht}{\mathrm{ch}_2}
\newcommand{\KPt}{K_{\IP^2}}

\newcommand{\CYX}{X}
\newcommand{\mirrorX}{\CYX^{\circ}}

\newcommand{\tC}{\widetilde{\IC}}
\newcommand{\tD}{\widetilde{D}}
\newcommand{\var}{\mathrm{var}}
\newcommand{\singphi}{\overset{\triangledown}{\varphi}}
\newcommand{\stokesAut}{\mathfrak{S}}
\newcommand{\hatphi}{\hat{\varphi}}
\newcommand{\hphi}{\hatphi}
\newcommand{\phimajor}{\check{\varphi}}

\newcommand{\phiOR}{\varphi}
\newcommand{\hatpsi}{\hat{\psi}}
\newcommand{\singpsi}{\overset{\triangledown}{\psi}}

\newcommand{\tildephi}{\tilde{\varphi}}

\newcommand{\tildepsi}{\tilde{\psi}}
\newcommand{\fana}{\check{f}}

\newcommand{\anyLog}{\cL og}

\newcommand{\totalt}{D_{u,t}}
\newcommand{\cZpert}{\cZ^{(0)}}
\newcommand{\cFpert}{\cF^{(0)}}
\newcommand{\cFmod}{\cF^{(0)}_{\mathrm{mod}}}
\newcommand{\cZmod}{\cZ^{(0)}_{\mathrm{mod}}}
\newcommand{\deltamod}{\delta_{\mathrm{mod}}}
\newcommand{\cZmodFrame}{\cZ^{(0)}_{\mathrm{mod},\omegag}}
\newcommand{\Litwo}{\mathrm{Li}_2}
\newcommand{\cont}{\mathrm{cont}}
\newcommand{\hatF}{\hat{\cF}}
\newcommand{\hatFpert}{\hatF^{(0)}}

\newcommand{\ray}{d}
\newcommand{\ana}{\mathrm{ANA}}
\newcommand{\SING}{\mathrm{SING}}
\newcommand{\simp}{\mathrm{simp}}
\newcommand{\SINGsimp}{\SING^{\simp}}
\newcommand{\SINGsimpRam}{\SING^{\sr}}
\newcommand{\holgerms}{\IC\{\zeta\}}
\newcommand{\sing}{\mathrm{sing}}
\newcommand{\sRhat}{\widehat{\sR}}
\newcommand{\sRtilde}{\widetilde{\sR}}
\newcommand{\Esing}{\overset{\triangledown}E}

\newcommand{\EsingOmegaN}[2]{\overset{\triangledown}{E_{#1}^{#2}}}
\newcommand{\dotDelta}{\dot{\Delta}}
\newcommand{\pd}{\dotDelta}

\newcommand{\Deltaplus}{\Delta^{+}}
\newcommand{\dotDeltagammaN}[1]{\dotDelta_{\omega_{\gamma_{#1}}}}
\newcommand{\dotDeltagammaKN}[2]{\dotDelta_{#2\omega_{\gamma_{#1}}}}
\newcommand{\dotDeltaPlusR}{\dotDelta^+_r}
\newcommand{\dotDeltaR}{\dotDelta_r}
\newcommand{\DeltagammaN}[1]{\Delta_{\omega_{\gamma_{#1}}}}
\newcommand{\pade}{\text{Padé}}

\newcommand{\res}{\mathrm{Res}}

\newcommand{\sr}{\mathrm{s.ram}}
\newcommand{\omegaResurgentSimplyRamified}{\bigoplus_{n\in \IN} \IC \delta^{(n)} \oplus \sRhat^{\sr}_\Omega}

\newcommand{\partk}{\mathbf{k}}

\newcommand{\chargeV}{\gamma}
\newcommand{\tchargeV}{\tilde{\gamma}}
\newcommand{\omegaG}{\omega_{\chargeV}}
\newcommand{\omegatG}{\omega_{\tchargeV}}

\newcommand{\twistedDG}{D_{\omegaG}}
\newcommand{\twistedDtG}{D_{\omegatG}}
\newcommand{\alienDiffOp}{\cD}
\newcommand{\alienDiffOpRay}{\alienDiffOp_{\theta_{\chargeV}}}

\newcommand{\SG}{S_{\chargeV}}
\newcommand{\SoG}{S_{\omegaG}}

\newcommand{\Aleph}{\alpha}

\newcommand{\tc}{\tilde{c}}
\newcommand{\td}{\tilde{d}}

\newcommand{\arrowProduct}{\mathop{\prod^{\curvearrowright}}}
\newcommand{\parrowProduct}{\mathop{\prod^{\curvearrowleft}}}

\newcommand{\mirrorQ}{X^\circ_5}
\newcommand{\quintic}{X_5}

\title{\centering The non-perturbative topological string: from resurgence to wall-crossing of DT invariants}

\abstract{We study the resurgence structure of the topological string partition function, with an emphasis on the Borel analysis of the instanton amplitudes. To this end, we introduce a differential operator that implements the pointed alien derivative when acting on the topological string partition function and its iterated alien derivatives. We show that the algebra of alien derivatives is isomorphic to the Kontsevich-Soibelman Lie algebra, thus establishing a direct link between the resurgence of the topological string and wall-crossing of generalized Donaldson-Thomas invariants. Numerically, we continue the exploration of the Borel plane of the quintic and local $\IP^2$. For the latter, we identify Borel singularities due to bound states involving D4-branes, and match the associated Stokes constants to the appropriate Donaldson-Thomas invariants. Finally, we identify the manifestation of a D2-brane decay in the Borel plane, and match to theoretical predictions.}

\author[a]{Simon Douaud}
\author[a]{Amir-Kian Kashani-Poor}

\affiliation[a]{Laboratoire de Physique de l’\'Ecole normale sup\'erieure,\\
CNRS, PSL Research University and Sorbonne Universit\'es,\\
24 rue Lhomond, 75005 Paris, France}

\emailAdd{simon.douaud@ens.fr}
\emailAdd{amir-kian.kashani-poor@ens.fr}

\DeclareFontFamily{U}{wncy}{}
\DeclareFontShape{U}{wncy}{m}{n}{<->wncyr10}{}
\DeclareSymbolFont{mcy}{U}{wncy}{m}{n}
\DeclareMathSymbol{\Sh}{\mathord}{mcy}{"58}

\begin{document}
\maketitle

\section{Introduction}
Topological string theory assigns to a Calabi-Yau threefold $X$ an infinite sequence of amplitudes $F_g$, one for each worldsheet genus $g$. Over the years, many reasons beyond the analogy with physical string theory have emerged for combining these into one object, the topological string amplitude
\begin{align} \label{eq:introF}
    F = \sum_{g=0}^\infty F_g g_s^{2g-2} \,. 
\end{align}
To name but a few, this combination arises in the graviphoton-curvature coupling of supergravity \cite{Antoniadis:1993ze}, can be related via geometric transitions \cite{Gopakumar:1998ki} to Chern-Simons theory \cite{Witten:1992fb}, leading to the topological vertex formalism \cite{Aganagic_2004} which solves the problem of the computation of $F$ when $X$ is toric, and a reorganization of the exponential of this sum leads to the generating function of Donaldson-Thomas invariants \cite{Maulik:2003rzb,Maulik:2004txy}.

The sum \eqref{eq:introF} does not converge. Clear indications of this are the factorial growth of the constant map contributions, which are proportional to the Euler characteristics of the moduli spaces $\cM_g$ of genus $g$ Riemann surfaces \cite{Faber:1998gsw, Bershadsky:1993cx}, $\chi(\cM_g) \sim (-1)^g (2g)!$, or the factorial growth of the leading behavior near a conifold point \cite{Ghoshal:1995wm}. Happily, strong evidence has been accumulating \cite{Huang:2007sb,Couso-Santamaria:2013kmu,Couso-Santamaria:2014iia,Couso-Santamaria:2016vcc,Gu:2022sqc,Gu:2023mgf,Douaud:2024khu} that the divergence of \eqref{eq:introF} is indeed factorial, and no worse. Hence, the formidable tools of resurgence can be brought to bear on the study of the topological string amplitude.

To render the rest of the introduction accessible to readers not familiar with resurgence, we present a bullet point summary of the principal tools that will enter the discussion:
\begin{itemize}
    \item {\it Borel transform:} A factorially divergent formal series $\tildephi(z)$ is mapped to a convergent series $\hatphi(\zeta)$ by dividing the $n^{\mathrm{th}}$ coefficient by $n!$. We call the complex plane parametrized by $\zeta$ the {\it Borel plane} of $\hatphi(\zeta)$.
    \item {\it Laplace transform:} $\hatphi(\zeta)$ under favorable circumstances can be mapped, via an integral transform, to a function which has $\tildephi(z)$ as its asymptotic expansion. This function is called its {\it Borel-Laplace summation}. We will denote the  operator mapping $\tildephi(z)$ to its Borel-Laplace summation by $\sS$.
    \item {\it Alien derivatives:} $\hatphi$ will generically exhibit singularities in the Borel plane. These are extracted via the alien derivative. These singularities often have standard forms up to coefficients called {\it Stokes constants}.
    \item {\it Stokes automorphism:} The Borel-Laplace summations of $\tildephi(z)$ along two rays separated by singularities in the Borel plane will differ by what we shall call {\it instanton amplitudes} (the precise definition is in the text). As both must yield the asymptotic expansion of $\tildephi(z)$, their difference must be invisible to the asymptotic expansion; it is typically of the form $\e^{-c/z}$. The map from (the trans-series expansion of) one Borel-Laplace transform to another is called a Stokes automorphism. We will denote it by the symbol $\stokesAut$.
    \item {\it Active rays:} A ray in the Borel plane emanating from the origin and passing through a singularity is called active.
\end{itemize}

In this paper, we continue our study of the non-perturbative structure of the perturbative string building on the results of \cite{Gu:2022sqc, Gu:2023mgf}. In previous work \cite{Douaud:2024khu}, we provided numerical evidence that the singularities in the Borel plane of the topological string amplitude can be associated to D-brane charges $\gamma$ on the Calabi-Yau threefold $X$, and that the Stokes constant associated to a singularity $\omega_\gamma$ coincides with the generalized Donaldson-Thomas invariant \cite{Joyce} associated to the charge $\gamma$. Here, we will in particular be interested in the question of wall-crossing: how does the Stokes automorphism change when, upon varying moduli, two active rays in the Borel plane cross? To address this question, we rewrite the result for the alien derivative derived in \cite{Gu:2022sqc, Gu:2023mgf} and arrive at a differential operator which implements the alien derivative when acting on the topological string partition function $\cZpert$ in a holomorphic limit. This permits us to study the resurgence of instanton amplitudes (and not just the original perturbative series) at arbitrary points of their Borel plane, vs. merely at multiples of the singularity to which they are associated. To study the question of wall-crossing, we compute the commutator of two alien derivatives using the presentation as a differential operator, and show that they satisfy the Kontsevich-Soibelmann Lie algebra relation \cite{Kontsevich:2008fj}. This implies that the identification of Stokes constants and Donaldson-Thomas invariants is precisely what is needed to guarantee that the Stokes automorphism across a sector in the Borel plane remains constant under variations of moduli which do not lead to singularities entering or leaving the sector. In previous work, relations between resurgence of the topological string partition function and Kontsevich-Soibelman wall-crossing were uncovered in terms of a dual partition function \cite{Iwaki_2024,Alexandrov:2023wdj}. The structure of the Kontsevich-Soibelman symplectomorphism is already visible in the context of exact WKB in the work \cite{DDP}. 

Numerically, our focus is on the study of higher order alien derivatives, which manifest themselves both in multi-instanton contributions and in the study of the Borel plane of instanton amplitudes. Our main numerical advance compare to our previous work \cite{Douaud:2024khu} is that we have dramatically sped up our computer algebra algorithm for computing these amplitudes. We can now efficiently compute them to an order comparable to the truncation order of the perturbative topological string amplitudes. This allows us to more effectively uncover subleading singularities in the Borel plane. We have thus identified bound states of D4-branes in the Borel plane of local $\IP^2$ and matched their Stokes constants with the associated Donaldson-Thomas invariants, encoded in the Poincaré polynomial of the Hilbert scheme of points. We furthermore study the Borel plane of instanton amplitudes by Borel transforming the perturbative instanton series obtained from the formalism of \cite{Gu:2022sqc,Gu:2023mgf}. Finally, we study the manifestation of a D2-brane decay in the Borel plane of local $\IP^2$.

\vspace{0.5cm}

The remainder of the paper is organized as follows: We begin in Section~\ref{s:lightning} with a rapid review of the topological string and the non-perturbative results of \cite{Gu:2022sqc,Gu:2023mgf}. We focus on multi-instanton amplitudes and their associated Stokes constants in Section~\ref{s:formalism}. We present the formalism in a slightly modified form compared to \cite{Gu:2022sqc,Gu:2023mgf} in order to lay the groundwork for the ensuing section. Section~\ref{s:resurgenceWallCrossing} contains the main theoretical advance of this paper. We separate the action of the alien derivative from the amplitude it acts on (rather than introducing a modified genus 0 amplitude as in \cite{Gu:2022sqc, Gu:2023mgf}). This allows us to consider successive alien derivatives in arbitrary directions, and leads to the result regarding wall-crossing announced above. Section~\ref{s:numerics} is devoted to our numerical work. We conclude in Section \ref{s:conclusions}. The paper closes with two appendices. Appendix~\ref{app:resurgentTools} contains a rather detailed summary of resurgent tools based on the excellent book \cite{Sauzin}. We deviate from this source only to generalize the presentation from simple singularities to simply ramified singularities, see e.g. \cite{sauzin2007}, as the resurgence of the topological string amplitude is most conveniently described in terms of these. In Appendix~\ref{app:parametric}, we consider two questions related to the behavior of the alien derivative in the context of parametric resurgence, so as to more carefully justify manipulations that have appeared previously in the string theory resurgence literature.

\section{A lightening review of the topological string and non-perturbative results} \label{s:lightning}
\paragraph{The perturbative topological string amplitudes} The topological string amplitudes $F_g$ as obtained from the holomorphic anomaly equations (HAE) \cite{Bershadsky:1993cx} are anholomorphic functions on all of moduli space $\cM$. The space $\cM$ is the quantum K\"ahler moduli space of the target when considering the topological A-model on the Calabi-Yau threefold $\CYX$, or equivalently the complex structure moduli space of the target when considering the B-model on the mirror Calabi-Yau threefold $\mirrorX$.\footnote{To streamline the presentation, we will move back and forth between the A- and B-model perspective as needed. The reader should have no difficulty mapping the statements we make to the mirror perspective.} The $F_g$ are polynomials in the three point function and twisted propagators, and rational functions of global coordinates on moduli space \cite{Yamaguchi:2004bt, Alim:2007qj}. All anholomorphicity is captured by the propagators. 

\paragraph{Choosing a frame and the holomorphic limit} Upon fixing a reference section of the K\"ahler line bundle, the $F_g$ become global functions on $\cM$. In contrast, the topological string amplitudes with characteristic behavior in a neighborhood of distinguished points in moduli space -- most notably having an enumerative interpretation around a large radius point or exhibiting a gap around a conifold point -- is obtained from $F_g$ by taking a so-called holomorphic limit. The choice of limit is equivalent to a choice of Lagrangian splitting of $H_3(\mirrorX,\IZ)$, i.e. a split into $A$- and $B$-periods, and is referred to as a choice of frame. We will denote by $\cF_g$ the holomorphic limit of $F_g$, leaving the choice of frame implicit.

\paragraph{The topological string amplitude} The topological string amplitude $F^{(0)}$ is defined by the formal Laurent series
\begin{align}
    F^{(0)} = \sum_{g=0}^{\infty} F^{(0)}_g g_s^{2g-2} \,.
\end{align}
Its holomorphic limit $\cFpert$ will feature more prominently in this paper. The unconventional superscript $(0)$ is to distinguish it from the topological string instanton amplitudes which we now introduce.

\paragraph{Borel-Laplace resummation} As we review briefly in the introduction and thoroughly in Appendix \ref{app:resurgentTools}, a divergent series under favorable circumstances can be associated to a function via Borel-Laplace summation. This association is not unique. The difference between two such summations must be invisible in an asymptotic expansion. Typically, the leading contribution to the difference is suppressed by a factor of the form $\e^{-\omega/z}$, with $\omega$ sometimes referred to as the instanton action in the physics literature. We call the asymptotic expansion of the coefficient of such terms times the exponential coefficient an instanton amplitudes. It is an example of a trans-series.

\paragraph{Topological string instanton amplitudes and HAE} In \cite{Couso-Santamaria:2013kmu,Couso-Santamaria:2014iia}, the idea was developed of studying the topological string amplitude non-perturbatively by solving the holomorphic anomaly equations for such a trans-series ansatz,
\begin{align} \label{eq:transseriesAnsatz}
    \cF = \cFpert + S_{\omega_\gamma} \cF^{(1)}_\gamma + \ldots \,,
\end{align}
and a general relation between instanton actions and periods of the Calabi-Yau threefold, previously observed in the case of the conifold \cite{Pasquetti_2010}, anticipated in \cite{Couso-Santamaria:2013kmu} and confirmed in the case of local $\IP^2$ \cite{Couso-Santamaria:2014iia}. 
A closed form of the trans-series was obtained in \cite{Gu:2022sqc,Gu:2023mgf} in terms of the perturbative amplitudes of the topological string:\footnote{In the following, we will write formulae for the case of the topological string on compact Calabi-Yau threefolds, as developed in \cite{Gu:2023mgf} The local case, studied in \cite{Gu:2022sqc} follows upon some simple modifications.} 
\begin{equation} \label{eq:firstInst}
    \cF^{(1)}_\gamma = \frac{1}{2\pi \ii} \left( 1 + g_s c^I \frac{\partial \cFmod}{\partial X^I}(X^I - g_s c^I)\right) \exp \left[ \cFmod(X^I -g_s c^I) - \cFmod(X^I) \right] \,,
\end{equation}
where the instanton action, in terms of a symplectic basis $\{P_I, X^I\}$ of periods of the Calabi-Yau threefold respecting the Lagrangian splitting of $H_3(X,\IZ)$ underlying the holomorphic limit considered, is given by
\begin{equation} \label{eq:instantonAction}
    \omegaG = c^I P_I + d_I X^I \,.
\end{equation} 
Here, $\cFmod(X^I)$ is the topological string amplitude modified at genus 0 and 1,
\begin{equation}
    \cFmod = \frac{1}{g_s^2} \tilde{\cF}_0 + \tilde{\cF}_1 + \sum_{g \ge 2} g_s^{2g-2} \cF_g \,, 
\end{equation}
where, in terms of the conventional genus 0 and genus 1 amplitudes,\footnote{In the local setting, $\cF_1$ is not modified.}
\begin{align} \label{eq:redefineF0F1}
    \tilde{\cF}_0 &= \cF_0 + m_I n_J X^I X^J \quad \text{such that} \quad c^I \partial_I \tilde{\cF}_0 = \omega_\gamma \,,\\
    \tilde{\cF}_1 &= \cF_1 - \left(\frac{\chi}{24}-1 \right) \log X^0 \,.
\end{align}
When the vector $c^I$ does not vanish identically, the constraint on $\tilde{\cF}_0$ can be solved by
\begin{align}
    \tilde{\cF}_0 = \cF_0 +\deltamod \cF_0 \,,
\end{align}
where
\begin{align}
    \deltamod \cF_0 = \begin{cases}
                            \frac{1}{2 c^I d_I} (d_I X^I)^2 & \text{if} \quad c^I d_I \neq 0\,,\\
                            \frac{1}{c^I a_I} a_I d_J X^I X^J & \text{for any} \quad a_I: c^I a_I \neq 0 \quad \text{if} \quad c^I d_I = 0 \,. 
                    \end{cases} 
\end{align}
Else, the relation $c^I \partial_I \tilde{\cF}$ is to be substituted in $\eqref{eq:firstInst}$ (also in the exponent upon Taylor expansion) before setting $c^I =0$.

Note the structure of $\cF^{(1)}_{\gamma}$. For $c^I \not \equiv 0$, the exponent upon Taylor expansion yields
\begin{align}
    \exp \left[ \cFmod(X^I -g_s c^I) - \cFmod(X^I) \right] =
    \exp \left[ -\frac{\omega_\gamma}{g_s}\right] \exp \left[\frac{1}{2} c^I \partial_I \omega_{\gamma}\right] \left( 1 + \cO(g_s) \right) \,.
\end{align}
This is the expected characteristic exponential instanton damping factor multiplied by a formal power series in $g_s$. The prefactor of the exponential yields a formal Laurent series, beginning at order $g_s^{-1}$. Note also that the dependence on $c^I$ only enters via $c^I \partial_I$ and is hence invariant under change of basis of $H_3(X, \IZ)$ which respects the Lagrangian splitting.

\paragraph{Boundary conditions and distinguished frames} The case $c^I \equiv 0$ is clearly very special: with the prescription detailed in the previous paragraph, the formal Laurent series multiplying the exponential collapses with the constant term, yielding
\begin{align} \label{eq:FoneInDistinguishedFrame}
    \cF^{(1)}_\gamma = \frac{1}{2\pi \ii} \left( 1+ \frac{\omegaG}{g_s}  \right) \e^{-\frac{\omega_\gamma}{g_s}} \,.
\end{align}
This behavior, inspired by the behavior of the topological string amplitude in a neighborhood of both the large radius and the conifold point, was indeed imposed as a boundary condition in \cite{Gu:2022sqc, Gu:2023mgf} to obtain the solution \eqref{eq:firstInst}: $\cF^{(1)}_\gamma$ was constrained to take the form \eqref{eq:FoneInDistinguishedFrame} whenever the instanton action is an $A$-period with regard to the Lagrangian splitting underlying the holomorphic limit. A frame in which this is true is called distinguished with regard to the instanton action. If we wish to indicate that the holomorphic limit of the topological string amplitude is taken in a distinguished frame with regard to the instanton amplitude $\omega_\gamma$, we write $\cFpert_{\omega_\gamma}$.

\paragraph{Instantons actions, integrality, and normalization} The reference \cite{Gu:2023mgf} was the first to emphasize that instanton actions should be proportional to  integral periods of the Calabi-Yau threefold. It should be noted that the notion of integrality of the periods is much more straightforward in the context of compact Calabi-Yau threefolds, the object of study of \cite{Gu:2023mgf}, than in the local case studied previously. An integral period, defined as a period of the holomorphic 3-form $\Omega$ over an integral cycle $H_3(X,\IZ)$, is only defined up to a choice of normalization of $\Omega$. This choice enters in the solution of the holomorphic anomaly equations in the normalization of the three point function -- see the detailed discussion in \cite{Gu:2023mgf}. The standard choice coincides with the normalization of $\Omega$ leading to $X^0 = (2\pi \ii)^{3/2} + \ldots$ for the holomorphic period in the neighborhood of a point of maximal unipotent monodromy. With this convention, analytic and numerical evidence strongly suggest that the coefficients $c^I,d_I$ in \eqref{eq:instantonAction} specifying the instanton actions always satisfy 
\begin{align} \label{eq:aleph}
    c^I,d_I \in \Aleph \, \IZ\quad \mathrm{with}\quad \Aleph = \ii\sqrt{2 \pi \ii} \,.
\end{align}
The relation between the normalization of the three point function and the normalization of periods is more tentative in the case of local Calabi-Yau threefolds. For such threefolds, a constant period $x^0$ exists. We will choose the normalization $x^0 = (2\pi \ii)^{3/2}$, which is natural when the local model is obtained as a limit of a compact Calabi-Yau threefold. The condition \eqref{eq:aleph} then also holds in the local case \cite{Couso-Santamaria:2014iia,Gu:2022sqc}.

\paragraph{From formal Laurent series to the Borel plane} Many aspects of instanton corrections become conceptually clearer when described in terms of the Borel plane, and this will be the language we will adopt henceforth.\footnote{We have already anticipated this choice by denoting the instanton action as $\omega$, a letter commonly used to denote the location of singularities in the Borel plane in the resurgence literature, rather than $\cA$.} As we review in detail in Appendix \ref{app:resurgentTools}, the ambiguities in the Borel-Laplace summation are due to singularities in the analytic continuation of the Borel transform\footnote{The genus 0 and 1 amplitudes can be included upon adjoining the symbols $\delta^{(0)}$ and $\delta^{(1)}$ to $\IC\{\zeta\}$ when introducing the Borel transform.}
\begin{align}
    \hatFpert = \sum_{g=2}^\infty \frac{\cF_g}{(2g-3)!}\zeta^{2g-3}\,.
\end{align}
The locations of the singularities coincide with the instanton actions. Up to finitely many polar terms, the instanton amplitudes with the exponential pre-factor stripped off coincides with the asymptotic expansion of the coefficient of a logarithmic singularity of the Borel transform. Alien operators capture these coefficients. Careful specifying the path of analytic continuation leads to a particular such operator, the pointed alien derivative $\dotDelta_{\omega_\gamma}$, which in fact yields the logarithm of the the Stokes automorphism: when the only singularities along a ray $\ray$ are multiples of a singularity $\omega_\gamma$, the Stokes automorphism across $\ray$ takes the form
\begin{equation} \label{eq:SandDelta}
    \stokesAut_\ray = \exp \sum_{\ell=1}^\infty  \dotDelta_{\ell \omega_{\chargeV}} \,.
\end{equation}
In the language of pointed alien derivatives, we can make the statement \eqref{eq:transseriesAnsatz} more precise:\footnote{We pass between formal Laurent series and the Borel plane via the Borel transform operator $\sB$, using the same notation in both instances.}
\begin{equation} \label{eq:alienF}
    \dotDelta_{\omegaG} \cF^{(0)} = S_{\omega_{\chargeV}} \cF^{(1)}_\gamma \,.
\end{equation}
\paragraph{Beyond simple singularities} Note that by the discussion in \ref{app:backToz}, the leading power of $\tfrac{1}{g_s}$ visible in the expression \eqref{eq:firstInst} for $\cF^{(1)}_\gamma$ indicates that this contribution arises from a double pole $\delta^{(1)}$ of the Borel transform, i.e. if we introduce the coefficients
\begin{align}
    \cF^{(1)}_\gamma = \e^{-\omega_\gamma/g_s} \sum_{k=-2}^\infty \cF^{\omega_\gamma}_k g_s^{k+1} \,, 
\end{align}
then, in a neighborhood of $\omega_\gamma$,
\begin{align} \label{eq:FhatNearSing}
    \frac{1}{S_{\omega_\gamma}}\hatFpert \sim \frac{-\cF^{\omega_\gamma}_{-2}}{2 \pi \ii \,(\zeta-\omega_{\gamma})^2} +  \frac{\cF^{\omega_\gamma}_{-1}}{2 \pi \ii \,(\zeta-\omega_{\gamma})} + \hatF^{\omega_\gamma}(\zeta-\omega_\gamma) \frac{\anyLog(\zeta-\omega_{\gamma})}{2\pi \ii} + R(\zeta-\omega) \,, \quad R(\zeta) \in \IC\{\zeta\} \,,
\end{align}
where $\anyLog$ denotes any branch of the logarithm, and 
\begin{align}  \label{eq:FhatomegaPert}
    \hatF^{\omega_\gamma}(\zeta) = \sum_{k=0}^\infty \frac{\cF^{\omega_\gamma}_k}{k!} \zeta^k \,.
\end{align}
Higher alien derivatives, which we shall discuss below, lead to ever higher poles. If we wish to maintain the standard indexing conventions of $\cFpert$, the class of singularities most commonly discussed in the resurgence literature, that of simple singularities comprising simple poles and logarithms, does not suffice to discuss the resurgence properties of the topological string. The inclusion in the resurgence formalism of simply ramified singularities (poles of arbitrary high order and logarithms) is thankfully not difficult, see e.g. \cite{sauzin2007}. The necessary adjustments to standard definitions and theorems can be found in Appendix \ref{app:resurgentTools}.

\paragraph{Stokes constants and BPS invariants} The coefficients $\{c^I,d_I\}$ parametrizing the location of the singularity \eqref{eq:instantonAction} map to D-brane charges $\gamma$, see e.g. \cite[Section]{Douaud:2024khu}. The charge lattice $\Gamma$ can be identified with $H_3(X,\IZ)$ in the B-model setting. The identification of the Stokes constant at a singularity $\omega_\gamma$ and the BPS invariant of charge $\gamma$ was anticipated in \cite{mm-s2019} -- see \cite[page 30/31]{Marino:2024tbx} for the history of this conjecture. In \cite{Douaud:2024khu}, we conjectured in the context of the topological string on compact Calabi-Yau threefolds that the Stokes constant associated to a singularity $\omega_\gamma$ coincides with the generalized Donaldson-Thomas invariant \cite{Joyce} associated to the charge $\gamma \in \Gamma$ (this explains the $\gamma$-index on the quantities introduced above). We investigated this claim for many points of the Borel planes of the 13 hypergeometric one-parameter Calabi-Yau threefolds. Whenever the generalized DT invariant was known in the correct stability chamber, we found coincidence. In all cases, the Stokes constants determined numerically were integers.

\section{Revisiting multi-instantons} \label{s:formalism}

\subsection{Non-primitive charges} \label{s:nonprimitive}
The identification of the Stokes constants with DT invariants suggests that we distinguish between two types of contributions from singularities to the Stokes automorphism: we shall call these primitive and multi-cover contributions respectively.\footnote{The terminology is chosen in analogy to BPS counts based on D-branes: for a given charge, states can arise from a D-brane wrapping an irreducible curve once (primitive contribution) or multiple times (multi-covering contribution).} A primitive contribution at a singular point $\omegaG$ gives rise to multi-cover contributions at all singularities $\omega_{\ell \chargeV}$, $\ell \in \IN$ associated to the charge $\ell \chargeV$. In resurgence terminology, these are multi-instanton contributions, and both analytical evidence (at the mirror of MUM and conifold points) and numerical evidence (at all singular points) \cite{Pasquetti_2010,Couso-Santamaria:2013kmu, Couso-Santamaria:2014iia,Gu:2022sqc,Gu:2023mgf} suggest that the associated Stokes constant is $\frac{1}{\ell^2}\SoG$. Note that the singularity associated to a charge $\chargeV$ which is non-primitive\footnote{A primitive {\it charge} is one that cannot be written as non-trivial integer multiples of charges $\gamma_1,\ldots, \gamma_n$; not to be confused with the terminology primitive {\it contribution} introduced above.} of the form $\gamma = m \gamma'$, $m>1$, can receive both a primitive and multi-cover contributions. In the following, we will write primitive contributions to a Stokes constant $\SoG$ as $\SG$. Thus, the Stokes constant $\SoG$ can be decomposed as
\begin{align} \label{eq:nonPrimitiveCurves}
    \SoG = \sum_{\chargeV' | \chargeV} \frac{1}{(\chargeV/\chargeV')^2}S_{\chargeV'} =\sum_{n\vert \gamma}\frac{1}{n^2}S_{\frac{\gamma}{n}}\,.
\end{align}
This nicely parallels the discussion in \cite[Section 6.2]{Joyce}, where generalized DT invariants ${\bar{\Omega}(\gamma)\in \IQ}$ (denoted as $\bar{DT}^\gamma$ in \cite{Joyce}) are decomposed in terms of invariants $\Omega(\gamma)$ (denoted as $\hat{DT}^\gamma$ in \cite{Joyce}) conjectured to be integer and defined via the relation \eqref{eq:nonPrimitiveCurves}. We thus conjecture the identification
\begin{equation}
    \SoG=\bar{\Omega}(\gamma)\,,\quad S_{\gamma}=\Omega(\gamma) \,.
\end{equation} 
Note that we treat instanton amplitudes associated to primitive charges $\gamma$ and non-primitive charges $\ell \gamma$, $\ell>1$, on the same footing, in the following sense. For simplicity, assume first that $\ell$ is prime and that no primitive contribution at the charge $\ell \gamma$ arises. Then, rather than imposing separate boundary conditions
\begin{align} \label{eq:FellInDistinguishedFrame}
    \cF^{(\ell)}_\gamma = \frac{1}{2\pi \ii} \left( \frac{1}{\ell^2}+ \frac{\omegaG}{\ell \,g_s}  \right) \e^{-\frac{\ell \omega_\gamma}{g_s}} \,,
\end{align}
we interpret this expression as
\begin{align}  \label{eq:multi}
    \dotDelta_{\omega_{\ell \gamma}} \cFpert = \frac{S_\gamma}{\ell^2}\frac{1}{2\pi \ii} \left( 1 + \frac{\omega_{\ell \gamma}}{g_s}  \right) \e^{-\frac{\omega_{\ell\gamma}}{g_s}} \,,
\end{align}
with the $\tfrac{1}{\ell^2}$ prefactor absorbed into the Stokes constant. More generally, this is the first contribution on the right-hand side of \begin{align}  \label{eq:multi}
    \dotDelta_{\omega_{\ell \gamma}} \cFpert &= \sum_{n\vert \ell}\frac{S_{\frac{\ell}{n}\gamma}}{n^2}\frac{1}{2\pi \ii} \left( 1 + \frac{\omega_{\ell \gamma}}{g_s}  \right) \e^{-\frac{\omega_{\ell\gamma}}{g_s}} \\
    &=\frac{S_{\gamma}}{\ell^2}\frac{1}{2\pi \ii} \left( 1 + \frac{\omega_{\ell \gamma}}{g_s}  \right) \e^{-\frac{\omega_{\ell\gamma}}{g_s}}+ \ldots+S_{\ell\gamma}\frac{1}{2\pi \ii} \left( 1 + \frac{\omega_{\ell \gamma}}{g_s}  \right) \e^{-\frac{\omega_{\ell\gamma}}{g_s}}\,.
\end{align}

\subsection{Higher powers of alien derivatives} \label{ss:higherAlienPowers}
To work out the action of higher powers of alien derivatives, it is convenient to express the relation \eqref{eq:alienF} in terms of the partition function $\cZmod = \exp \cFmod$. As a logarithm of an automorphism is a derivation, we have\footnote{Note that $\cFpert$ and $\cFmod$ differ by an additive term which is the sum of a constant and a monomial in $\frac{1}{g_s}$, hence \eqref{eq:alienF} is valid for either. $\cZpert$ and $\cZmod$ by contrast differ by a rescaling.}
\begin{align} \label{eq:alienZmod}
    \dotDelta_{\omegaG}\cZmod(X^I) &= \cZmod(X^I) \dotDelta_{\omegaG} \cFmod(X^I) \\
    &= \frac{\SoG }{2\pi \ii} \left( 1 + g_s c^I \frac{\partial \cFmod}{\partial X^I}(X^I - g_s c^I)\right) \cZmod(X^I -g_s c^I) \nonumber\\
    &=  \frac{\SoG}{2\pi \ii} \left( 1 + g_s c^I \frac{\partial }{\partial X^I}\right) \cZmod(X^I -g_s c^I)\,. \nonumber
\end{align}
Note how the exponential of the trans-series $\cF$ needs to be interpreted: keeping track of powers of a characteristic instanton factor $\e^{-\omega_\gamma/g_s}$ via a constant $\cC$, 
\begin{equation}
    \cF = \cF^{(0)} + \cC \, \cF^{(1)} + \cO(\cC^2) \,,
\end{equation}
we have
\begin{equation}
    \cZ = \cZ^{(0)}\e^{ \cC \, \cF^{(1)} + \ldots} = \cZ^{(0)}\left(1 + \cC \,  \cF^{(1)} + \cO(\cC^2) \right) \,.
\end{equation}
We obtain a trans-series multiplied by an overall factor of $\cZpert$ (which happens to have a leading $\e^{\cF_0/g_s^2}$ behavior). 

To iterate \eqref{eq:alienZmod} and obtain an expression for higher powers of an alien derivative acting on $\cZpert$, we will assume the following two statements to be true:
\begin{enumerate}
    \item The alien derivative, and hence also the Stokes automorphism, commutes with derivatives $\partial_I$. 
    \item The alien derivative commutes with $g_s$-dependent shifts of parameters.
\end{enumerate}
In Appendices \ref{app:parametricDerivatives} and \ref{app:pointedInfinitesimalCommutation} respectively, we reduce these two claims to certain regularity assumptions on $\cFpert$.\footnote{Note that the first claim would also imply that $\cZ/\cZpert$ must satisfy the holomorphic anomaly equations, as these can be formulated as a linear differential operator annihilating this quantity, see e.g. \cite[equation (5.126)]{Gu:2023mgf}.}

Assuming the validity of these claims, we obtain
\begin{align} \label{eq:deltan}
    \left(\dotDelta_{\omegaG} \right)^n {\cZmod}(X^I) &=  \frac{\SoG}{2\pi \ii} \left( 1 + g_s c^I \frac{\partial }{\partial X^I}\right) \left(\dotDelta_{\omegaG}\right)^{n-1} \cZmod(X^I -g_s c^I) \\ 
    &= \left[ \frac{\SoG}{2\pi \ii} \left( 1 + g_s c^I \frac{\partial }{\partial X^I}\right) \right]^n \cZmod(X^I - ng_s c^I) \,.\nonumber
\end{align}
Note that each derivative $g_s \tfrac{\partial}{\partial X^I}$ gives rise to a leading contribution of order $1/g_s$. Hence, $\dotDelta_\omega^n \cZmod$, and equivalently $\dotDelta_\omega^n \cFmod$, have a leading $g_s^{-n}$ term multiplying the exponential $\e^{-{\omegaG}/g_s}$ factor, corresponding to a pole of order $n+1$ in the Borel plane at $\zeta = \omegaG$.

The expression \eqref{eq:deltan} immediately yields the elegant expression found in \cite{Iwaki_2024} for the Stokes automorphism acting on $\cZmod$, assuming that only one primitive contribution to the Stokes automorphism occurs along the ray $\IR^+ \omegaG$, at charge $\chargeV$: 
\begin{align} \label{eq:stokesZprimitive}
    \stokesAut_{\theta_\chargeV} \cZmod&=\e^{\sum_\ell \cC^\ell \dotDelta_{\ell \omegaG} }\cZmod \\
    &= \exp\left[\frac{\SG}{2\pi \ii}\sum_\ell \frac{1}{\ell^2}\left(1+ g_s \,\ell c^I \partial_I  \right) \cC^\ell \e^{-g_s \ell c_I \partial_I} \right] \cZmod \label{eq:sumell}\\
    &= \exp \left[\frac{\SG}{2\pi \ii} \left( \Litwo(\cC \e^{-g_s c^I \partial_I}) - g_s c^I \partial_I\log(1- \cC \e^{-g_s c^I \partial_I}) \right) \right] \cZmod \,.
\end{align}
Expression \eqref{eq:deltan} also allows for a straightforward derivation of the holomorphic limit of multi-instanton amplitudes, found in \cite{Gu:2022sqc, Gu:2023mgf} by solving the holomorphic anomaly equations: with the conventions regarding the case $c^I \equiv 0$ laid out in Section \ref{s:lightning}, $(\dotDeltagammaN{})^n$ acting on the partition function $\cZmodFrame$ in the distinguished frame with regard to the singularity $\omega_\gamma$ yields, according to \eqref{eq:deltan},
\begin{align}
    \left(\dotDelta_{\omega_\gamma}\right)^n \cZmodFrame = \left(\frac{S_{\omega_{\gamma}}}{2\pi \, \ii} \right)^n \left(1 + \frac{\omega_{\gamma}}{g_s}  \right)^n \e^{-\frac{\omega_{\gamma}}{g_s}}
    \cZmodFrame \,.
\end{align}
For an arbitrary frame, we need to evaluate the higher order derivatives occurring in \eqref{eq:deltan}. Invoking Faà di Bruno's formula for the chain rule applied to higher derivatives, we obtain
\begin{align}
    \left(g_s c^I \partial_I \right)^m \cZmod(X^I - ng_s c^I) = \sum_{\partk,\,d(\partk)=m} C_{\partk}  \prod_{i} \left(\left(g_s c^I \partial_I\right)^{i} \cFmod(X^I - ng_s c^I)\right)^{k_i}\cZmod(X^I - ng_s c^I) \,.
\end{align}
Here, $\partk = (k_1, \ldots,k_r)$ is a partition of $m$, i.e. $d(\partk) = \sum_{j=1}^r j k_j = m$, and 
\begin{align}
    C_{\partk} = \frac{d(\partk)!}{\prod_{j}k_j!(j!)^{k_j}}\,.
\end{align}
This precisely reproduces the prescription derived in \cite{Gu:2022sqc,Gu:2023mgf} for computing multi-instanton amplitudes: the central equation to compare to is  \cite[Equation (5.143)]{Gu:2023mgf}, where, in the holomorphic limit $\mathsf{D} = g_s c^I \partial_I$, and $X^{(n)} = g_s c^I \partial_I \cFmod(X^I - ng_s c^I)$, see \cite[Equation (5.76)]{Gu:2023mgf} and \cite[Equations (5.141), (5.149)]{Gu:2023mgf}, respectively. The generalization to including anti-instantons, i.e. contributions carrying instanton factors $\e^{+\ell \omega_\gamma/g_s}$, is immediate.\footnote{The references \cite{Gu:2022sqc,Gu:2023mgf} define specific boundary conditions, indicated by a subscript ${}_\ell$, for studying $\ell$-instanton amplitudes. We can circumvent this step by absorbing the associated multiplicative factor $\tfrac{1}{\ell^2}$ in the Stokes constant, see the discussion in Section \ref{s:nonprimitive}.}

\section{Resurgence and wall-crossing} \label{s:resurgenceWallCrossing}

The Stokes automorphism $\stokesAut_{\ray}$ relates Borel-Laplace summations of the topological string amplitude just below and just above the ray $\ray$ in the Borel plane, and is non-trivial precisely when 
$\ray$ is active, i.e. when singularities of the Borel transform lie on it. We can compose Stokes automorphisms to relate Borel-Laplace summations along two rays bounding a sector $V$ in the Borel plane. A natural question is how the composed Stokes automorphism varies as we vary the moduli of the Calabi-Yau threefold. As the moduli vary, the active rays shift continuously in the Borel plane; at special loci they can coincide, and upon further variation they cross. This picture is strongly reminiscent of the Kontsevich–Soibelman (KS) formalism for wall-crossing \cite{Kontsevich:2008fj}, which governs precisely the discontinuous jumps of generalized Donaldson-Thomas invariants at such loci (the walls in moduli space). Since we have conjecturally identified these invariants with the Stokes constants of our setup, that formalism becomes directly relevant to the present discussion.

Since the Stokes automorphism is completely determined by the pointed alien derivative, we begin by recasting equation \eqref{eq:alienZmod} in a form better suited to analyzing the effect of reordering composed Stokes automorphisms. We then use this result to compute the commutator of two alien derivatives, allowing us to make contact with the Kontsevich–Soibelman wall-crossing formalism.

\subsection{Successive alien derivatives evaluated at arbitrary points}\label{s:successiveAliens}
The higher powers \eqref{eq:deltan} of the alien derivative are necessary to calculate the Stokes automorphism due to a singularity at the point $\omegaG$ on the Borel plane. But more fundamentally, the alien derivative $\dotDelta_{\omegaG}$ encodes the singularity of the Borel transform at the point $\omegaG$. A second alien derivative then extracts the singularity of the coefficient of the log singularity. It is therefore meaningful to evaluate successive pointed alien derivatives, each evaluated at different points $\omega_{\chargeV_i}$. Expressing the alien derivative in the form \eqref{eq:firstInst} is inconvenient for addressing this question, as the contribution of $\cF_0$ to $\cFmod$ has been modified as specified in \eqref{eq:redefineF0F1} to depend on the charge $\chargeV$. In this section, we will therefore extract this dependence on $\chargeV$ from $\cFmod$ and incorporate it into the operator acting on $\cFpert$ to yield its pointed alien derivative.

In the following, $\cFpert$ and $\cZpert$ will hence denote
\begin{equation} \label{eq:defFg}
    \cFpert = \frac{1}{g_s^2} \cF_0 + \tilde{\cF}_1 + \sum_{g \ge 2} g_s^{2g-2} \cF_g \,, \quad \cZpert = \exp \cFpert \,,
\end{equation}
i.e. we retain the (singularity independent) genus 1 modification while discarding the (singularity dependent) genus 0 modification, such that (when $c^I \not \equiv 0)$,
\begin{align}
    \cZmod = \exp \left[\frac{1}{g_s^2}\deltamod \cF_0  \right] \cZpert \,.
\end{align}
The equation \eqref{eq:alienZmod} then becomes
\begin{align}
    \exp \left[\frac{1}{g_s^2}\deltamod \cF_0(X^I)   \right] \dotDelta_{\omegaG} \cZpert (X^I)&= \frac{\SoG}{2\pi \ii} \left( 1 + g_s c^I \frac{\partial }{\partial X^I}\right) \exp \left[\deltamod \cF_0(X^I -g_s c^I) \right] \nonumber \\
    &{}\quad \times \cZpert(X^I -g_s c^I) \nonumber \\
    &= \exp \left[\frac{1}{g_s^2}\deltamod \cF_0(X^I)  \right] \frac{\SoG}{2\pi \ii} \left( 1 + g_s c^I \frac{\partial }{\partial X^I} + \frac{1}{g_s} d_I X^I \right) \nonumber \\
    &{}\quad \times \exp \left[-\frac{1}{g_s} d_I X^I \right] \exp\left[\frac{1}{2}c^I d_I \right] \exp \left[-g_s c^I \frac{\partial}{\partial X^I} \right] \cZpert(X^I) \,,
\end{align}
where we have used the quadraticity and the defining equation of $\deltamod \cF_0$ to write
\begin{align}
    \deltamod \cF_0(X^I -g_s c^I) = \deltamod \cF_0(X^I) -g_s d_I X^I + \frac{1}{2}g_s^2 d_I c^I \,.
\end{align}
Setting 
\begin{equation} \label{eq:twistedDiff}
    \twistedDG = g_s c^I \partial_I + \frac{1}{g_s} d_I X^I 
\end{equation}
and invoking
\begin{equation}
    -\frac{1}{2} \left[ - \frac{1}{g_s} d_I X^I, -g_s c^I \partial_I \right] = -\frac{1}{2} c^I d_I 
\end{equation}
and the Baker-Campell-Hausdorff relation $\e^X \e^Y = \e^{X+Y+\frac{1}{2}[X,Y]}$ for $[X,Y]$ central, we arrive at
\begin{align} \label{eq:alienZ}
    \dotDelta_{\omegaG} \cZpert(X^I) &= \frac{\SoG}{2\pi \ii} \left( 1 + \twistedDG \right) \e^{-\twistedDG} \cZpert(X^I) \,.
\end{align}
Note that equation \eqref{eq:alienZ} also captures the correct behavior of $\dotDelta_{\omegaG}$ in a distinguished frame, in which by definition $c^I \equiv 0$ such that $\omega_\gamma = d_I X^I$:
\begin{align} 
    \dotDelta_{\omegaG} \cZpert_{\omegag}(X^I) &= \frac{\SoG}{2\pi \ii} \left( 1 + \frac{\omegag}{g_s} \right) \e^{-\frac{\omegag}{g_s}} \cZpert_{\omegag}(X^I) \,.
\end{align}
We can use equation \eqref{eq:alienZ} to rewrite the action of the Stokes automorphism given in \eqref{eq:stokesZprimitive} in terms of $\cZpert$ rather than $\cZmod$:
\begin{align}
    \stokesAut_{\theta_\chargeV} \cZpert
    &= \exp \left[\frac{\SG}{2\pi \ii} \left( \Litwo(\cC \e^{-\twistedDG}) - \twistedDG \log(1- \cC \e^{-\twistedDG}) \right) \right] \cZpert\,.
\end{align}
The relation \eqref{eq:alienZ} will be central to what follows.
Assuming the validity of the two claims put forth in Section \ref{ss:higherAlienPowers}, it implies that when acting on $\cZ$, we can, upon stripping away the Stokes constants, replace alien derivatives by the differential operators
\begin{align} \label{eq:alienDiffOp}
    \alienDiffOp_{\omegaG} = \frac{1}{2\pi \ii} (1 + \twistedDG) \e^{-\twistedDG} \,.
\end{align}
When evaluating a composition of alien operators, we begin with the innermost, replace it by the corresponding differential operator $\alienDiffOp_{\omegaG}$, then commute the remaining alien derivatives past $\alienDiffOp_{\omegaG}$ and repeat, thus arriving at the following reversal in the order of composition:
\begin{align} \label{eq:multiDeltaZ}
    \dotDelta_{\omega_1} \cdots \dotDelta_{\omega_n} \cZpert = \left(\prod_{i=1}^n S_{\omega_i}\right) \alienDiffOp_{\omega_n} \cdots \alienDiffOp_{\omega_1} \cZpert \,. 
\end{align}
We can use the derivation property of $\dotDeltagammaN{}$ to obtain $\dotDeltagammaN{}\cZpert$ from knowledge of $\dotDeltagammaN{}\cFpert$, to derive the action of multiple alien derivatives acting on $\cFpert$ from \eqref{eq:multiDeltaZ}. It is clear from the structure of \eqref{eq:twistedDiff} that each successive action of the alien derivative will increase the order of the leading pole of the formal Laurent series by one. It will be convenient to introduce the notation
\begin{equation}
    \dotDeltagammaN{1} \circ... \circ \dotDeltagammaN{n} \cFpert=\left(\prod_{i=1}^n{S_{\omega_{\gamma_i}}}\right)g_s^{-(n+1)}\e^{-\frac{\sum_{i=1}^n \omega_{\gamma_i}}{g_s}}\cF^{(\gamma_n\vert...\vert\gamma_1)}(g_s) \,,
    \label{eq:defFgammas}
\end{equation}
generalizing the notation in e.g. \cite{Douaud:2024khu,Couso-Santamaria:2014iia} for the instanton amplitude and rescaling such that 
\begin{align}
    \cF^{(\gamma_1\vert...\vert\gamma_n)}(g_s) \in g_s \,\IC[[g_s]]\,,
\end{align}
with coefficients
\begin{equation} \label{eq:FgammaCoeffs}
    \cF^{(\gamma_1\vert...\vert\gamma_n)}(g_s)=\sum_{k\geq0} \cF_k^{(\gamma_1\vert...\vert\gamma_n)}g_s^{k+1} \,.
\end{equation}
The alien derivative of $\cF^{(\gamma_1\vert...\vert\gamma_n)}$ follows immediately from equation  \eqref{eq:defFgammas},
\begin{align} \label{eq:alienDFinst}
    \dotDeltagammaN{} \cF^{(\gamma_1\vert...\vert\gamma_n)} = S_{\omega_{\gamma}} \, g_s^{-1} \e^{-\frac{\omega_{\gamma}}{g_s}} \cF^{(\gamma_1\vert...\vert\gamma_n|\gamma)} \,.
\end{align}
We will in particular be interested in the singularities of $\hatFpert$, the Borel transform of the topological string amplitude, along a ray $\ray_\gamma$ in the Borel plane. Let $\gamma$ be primitive, and let 
\begin{align}
    \Omega \cap [0,n \omega_\gamma] = \left\{0,\omega_\gamma, 2\omega_\gamma, \ldots, n\omega_\gamma \right\} 
\end{align}
for any $n \in \IN^*$, i.e. the only singularities along $\ray_\gamma$ are multiples of $\omega_\gamma$. In the notation of Appendix~\ref{app:alienOpsAndRays}, we thus have $\omega_r = r \omega_\gamma$, and the singularity of $\hatFpert$ at $n \omegag$, obtained via a path of analytic continuation above $\ray_{\gamma}$, is extracted via the alien operator $\Deltaplus_r$. By the exponential relation \eqref{eq:DeltaDandDeltaPlusD} between the Stokes automorphism $\Deltaplus_{\ray}=\mathrm{Id} + \sum_{r\in \IN^*}\dotDeltaPlusR$ and the pointed alien derivative $\dotDelta_r$, we have
\begin{align}
    \Deltaplus_{n\omega_\gamma} \left(\cFpert\right)&= \tau_{n\omega_\gamma}^{-1} \circ \sum_{k=1}^n \frac{1}{k!} \sum_{\substack{n_1, \ldots, n_k \in \IN^* \\n_1 + \cdots + n_k = n}} \dotDelta_{n_1\omegaG}\circ \cdots \circ  \dotDelta_{n_k\omegaG} \left(\cFpert\right) \\
    &= \sum_{k=1}^n \frac{g_s^{-(k+1)}}{k!} \sum_{\substack{n_1, \ldots, n_k \in \IN^* \\n_1 + \cdots + n_k = n}} \left(\prod_{i=1}^k{S_{n_i\omega_{\gamma}}}\right)\cF^{(n_k\gamma | \ldots |n_1\gamma)} \,.
\end{align}

\subsection{The commutator of two alien derivatives}
Now that we have recast the pointed alien derivative $\dotDelta_{\omegag}$ in terms of the operator $\alienDiffOp_{\omegag}$, computing the commutator of two alien derivatives reduces to a straightforward calculation. We begin by evaluating the composition of two such derivatives associated to two charges $\gamma$ and $\tilde{\gamma}$,
\begin{align}
    \alienDiffOp_{\omegaG} \alienDiffOp_{\omegatG} &= \frac{1}{(2\pi \ii)^2} \e^{\frac{1}{2}(c^I d_I + \tc^I \td_I)}\left(1+ \twistedDG \right)\e^{-\frac{1}{g_s} d_I X^I} \e^{-g_s c^I \partial_I} 
    \left(1+\twistedDtG\right)\e^{-\frac{1}{g_s} \td_I X^I} \e^{-g_s \tc^I \partial_I} \nonumber\\
    &= \frac{1}{(2\pi \ii)^2} \e^{\frac{1}{2}(c^I d_I + \tc^I \td_I)}\left(1+ \twistedDG \right) \nonumber \\
    &{}\quad\quad\quad \left(1+ \twistedDtG - c^I \td_I + d_I \tc^I \right)\e^{c^I \td_I} \e^{-\frac{1}{g_s} (d_I + \td_I) X^I} \e^{-g_s (c^I+\tc^I) \partial_I} \,.
\end{align}
Recalling the normalization $\omegaG, \omegatG \in \Aleph\,\Lambda$ of the singular points $\omegaG$, $\omegatG$ with regard to the period lattice of the Calabi-Yau threefold $X$, with $\Aleph = \ii \sqrt{2\pi \ii}$, we have
\begin{align}
    d_I \tc^I = c^I \td_I + 2\pi \ii \, \langle \chargeV, \tchargeV \rangle \quad \Rightarrow \quad \frac{1}{2} (c^I \td_I + d_I \tc^I) = c^I \td_I + \pi \ii \,\langle \chargeV, \tchargeV \rangle   \,.
\end{align}
Hence,
\begin{align}
     \alienDiffOp_{\omegaG} \alienDiffOp_{\omegatG} &= (-1)^{\langle \chargeV, \tchargeV \rangle } \frac{1}{(2\pi\ii)^2} \e^{\frac{1}{2}(c^I + \tc^I)(d_I +\td_I)} \left[ \left(1+ \twistedDG \right)\left(1+ \twistedDtG \right) - \left(1+ \twistedDG \right) \Aleph^2 \langle \chargeV ,\tchargeV \rangle \right] \nonumber \\
     & \hspace{6cm}\times\e^{-\frac{1}{g_s} (d_I + \td_I) X^I} \e^{-g_s (c^I+\tc^I) \partial_I}  \,.
\end{align}
Thus,
\begin{align} \label{eq:alienComm}
    \left[ \alienDiffOp_{\omegaG}, \alienDiffOp_{\omegatG} \right] &=  \frac{(-1)^{\langle \chargeV, \tchargeV \rangle }}{(2\pi \ii)^2} \e^{\frac{1}{2}(c^I + \tc^I)(d_I +\td_I)} \Big( \left[ \twistedDG , \twistedDtG \right] - \Aleph^2 \langle \chargeV, \tchargeV \rangle \left( 2 + \twistedDG + \twistedDtG \right)\Big) \nonumber \\
    & e^{-\frac{1}{g_s} (d_I + \td_I) X^I} \e^{-g_s (c^I+\tc^I) \partial_I} \nonumber \\
    &=  (-1)^{\langle \chargeV, \tchargeV \rangle }\langle \chargeV, \tchargeV \rangle \alienDiffOp_{\omegaG + \omegatG} \,. 
\end{align}
By \eqref{eq:multiDeltaZ}, we conclude that
\begin{align}
    [ \dotDelta_{\omega_\gamma}, \dotDelta_{\omega_{\tilde{\gamma}}}] = -(-1)^{\langle \chargeV, \tchargeV \rangle }\langle \chargeV, \tchargeV \rangle \dotDelta_{\omegaG + \omegatG} \,.
\end{align}

\subsection{Stokes automorphisms and the Kontsevich-Soibelman algebra}
The commutation relation \eqref{eq:alienComm}
is precisely that defining the Kontsevich-Soibelman Lie algebra \cite[page 11]{Kontsevich:2008fj}, with the identification $\dotDelta_{\omegag} \leftrightarrow -e_\gamma$. We give a very rough one paragraph summary of the KS formalism: these authors introduce a Lie algebra over $\IQ$ with basis elements $(e_\gamma)_{\gamma \in \Gamma}$, $\Gamma$ a free abelian group, with commutator
\begin{align}
    [e_\gamma ,e_{\gamma'}] = (-1)^{\langle \gamma,\gamma' \rangle}\langle \gamma, \gamma' \rangle e_{\gamma + \gamma'} \,.
\end{align}
The free abelian group $\Gamma$ is mapped by a central charge map $Z_t$, $t \in U$ for some appropriate parameter set $U$, to $\IC$: $Z_t: \Gamma \rightarrow \IC$. To each element $\gamma$, they introduce a Lie group element \begin{align} \label{def:Kgamma}
    K_\gamma = \exp\left(-\Omega(\gamma) \sum_{n=1}^\infty \frac{e_{n\gamma}}{n^2}\right) \,,
\end{align}
depending on integers $\Omega(\gamma)$. Given a sector $V$ in the complex plane with vertex at the origin and opening angle less than $\pi$, they consider the product
\begin{align} \label{eq:productK}
    \arrowProduct_{\gamma: Z(\gamma) \in V} K_\gamma \,,
\end{align}
with the arrow above the product indicating that the product is to be taken in clockwise order. The integers $\Omega(\gamma)$ are said to satisfy the KS wall-crossing relation if upon variation of the moduli $t$ such that no values of $Z(\gamma)$ enter or leave the sector $V$, they adjust such that the product \eqref{eq:productK} remains constant. Joyce and Song prove that the generalized Donaldson-Thomas invariants they introduce in \cite{Joyce} satisfy the KS wall-crossing formula.

The implications for our study are now all but clear, given that the relevant central charge map $Z(\gamma)$ in the case of the topological string on a Calabi-Yau threefold is a constant multiple of the period map $\omegag$: choose a sector $V$ in the Stokes plane bounded by the non-active rays $\ray_1$ and $\ray_2$. The Stokes automorphism $\stokesAut_V$ from $\ray_1$ to $\ray_2$ is given by the composition
\begin{align} \label{eq:stokesV}
    \stokesAut_V = \arrowProduct_{\gamma: \omega_\gamma \in V} \stokesAut_{\theta_\gamma} \,, \quad \theta_\gamma = \mathrm{phase}\left(\omegag\right) \,,
\end{align}
where for almost all $\theta$, $\stokesAut_{\theta} = 1$. Let us assume that only one primitive singularity $\omega_\gamma$ lies on the ray determined by $\gamma$ (if this is not the case, we consider a small variation of the moduli to arrive at this generic situation), and set
\begin{align}
    \alienDiffOpRay = \sum_{\ell=1}^\infty \frac{1}{\ell^2} \cD_{\ell \omegaG} \,.
\end{align}
Let us define the operator
\begin{align} \label{eq:productD}
    \cS_V(X^I;z)  = \parrowProduct_{\gamma: \omegag \in V} \left( \e^{\alienDiffOpRay} \right)^{\Omega(\gamma)}\,.
\end{align}
The variable dependence needs to be understood as follows: the $z$ dependence of the $X^I$ is implicit in the $X^I$ dependence of the operator. The only explicit $z$ dependence hence occurs due to the dependence on $z$ of the ordering of the composition. Now consider a variation of moduli from $z$ to $z'$ so that no singularities cross into or leave $V$. By \eqref{eq:alienComm}, the identification of $\Omega(\gamma)$ with DT invariants, and the KS formula (note the minus sign in the exponent in \eqref{def:Kgamma} and the reverse ordering of the products \eqref{eq:productK} and \eqref{eq:productD}, we can conclude that
\begin{align}
    \cS_V(X^I;z) = \cS_V(X^I;z') \,.
    \label{eq:defWcformula}
\end{align}
As a function of $X^I$, the image of $\cZ(X^I)$ under $\stokesAut_V$ hence remains unchanged upon varying the moduli from $z$ to $z'$. 

Ultimately, this observation could be an ingredient in the proof that the Stokes constants coincide with generalized Donaldson-Thomas invariants. One proof strategy would be the following: 
\begin{enumerate}
    \item {\it Demonstrate that the Borel-Laplace summation of the topological string partition function is continuous under variation of the moduli, as long as no singularities cross the integration path under this variation.} Given this, compare the two Borel-Laplace summations associated to the rays $\ray_1$ and $\ray_2$ spanning the sector $V$. As long as we consider variations under which no singularities pass into or pass out of the sector, the associated Stokes automorphism upon Borel transform maps a continuous to a continuous function, hence must itself be continuous under such variation of the moduli.
    \item {\it Establish that the Stokes constants are discrete.} Given this, as under change of ordering of active rays within $V$, the Stokes automorphism either jumps discontinuously or is constant, it must be constant. Hence, generalized Donaldson-Thomas invariants form a possible set of Stokes constants.
    \item {\it Establish the coincidence of Stokes constants and generalized Donaldson-Thomas invariants in some chamber.}
\end{enumerate}

\section{Examples and numerical evidence} \label{s:numerics}
\subsection{The geometries}
In \cite{Douaud:2024khu}, we studied the resurgence properties of the topological string amplitudes of all 13 one-parameter hypergeometric compact Calabi-Yau threefolds in some detail. In the present work, we will provide numerical evidence for the claims put forth in Section \ref{s:formalism} based on the quintic threefold, as well as the non-compact Calabi-Yau threefold which is the total space of the canonical bundle of the projective plane, $\KPt = \cO_{\IP^2}(-3)$, also referred to as local $\IP^2$. The quintic was studied in this context also in \cite{Gu:2023mgf}, while local $\IP^2$ has served as a case study of resurgence properties of the topological string amplitudes in numerous works, e.g. \cite{Couso-Santamaria:2014iia,Gu_2022,Gu:2022sqc,Marino:2024yme,Marino:2023gxy}.

\subsubsection{The quintic}
When we speak of the one-parameter hypergeometric Calabi-Yau models, we in fact have a pair of threefolds in mind for each model, related via mirror symmetry. In the case of the quintic, we will call $\quintic$ the threefold whose K\"ahler moduli space is one-dimensional, and its mirror with one-dimensional complex structure moduli space $\mirrorQ$. When we speak of the topological string amplitude $\cFpert$ of the quintic in the following, we have in mind the topological string amplitude of the A-model on $\quintic$ or equivalent that of the B-model on $\mirrorQ$.

As reviewed above, an important observation in \cite{Gu:2023mgf} was that singularities on the Borel plane occur at a constant multiple of the periods of the holomorphic 3-form $\Omega$ of $\mirrorQ$ evaluated on $H_3(\mirrorQ,\IZ)$. In the conventions of \cite{Douaud:2024khu}, the normalization constant is $\Aleph$ as given in \eqref{eq:aleph} above. The image of $H_3(\mirrorQ,\IZ)$ under the mirror map is a sublattice $\Lambda_{\quintic} \subset H^{\mathrm{even}}(\quintic,\IQ)$ which corresponds to D6-D4-D2-D0 charge (see \cite[Section 5.1]{Douaud:2024khu} for a quick summary of these matters). Mapping the integral period vector to the corresponding D-brane charge is central to our study, as the Stokes constants associated to a singularity in the Borel plane are conjectured to coincide with the Donaldson-Thomas invariants associated to this charge \cite{Gu:2023mgf,Douaud:2024khu}. This map is worked out by matching the expression for the central charge on the threefold and its mirror.

As all one-parameter hypergeometric models, the complex structure moduli space of $\mirrorQ$ has three distinguished points: the point of maximum unipotent monodromy (MUM), which we shall situate at $z=0$, the conifold point $\mu$, which in our coordinates will lie at $z=\tfrac{1}{3125}$, and finally a point at $z=\infty$, which in the case of $\mirrorQ$ is an R-point of order 5.
One can define a distinguished basis of periods $(X^0,X^1,P_0,P_1)$ (see e.g \cite{Gu:2023mgf} for more details), such that in this basis, the period 
\begin{equation}
  \int_{\gamma}\Omega=n_0X^0+ n_1 X^1+m^0P_0+m^1P_1,
\end{equation}
is the central charge of a bound states of $n_0$ D0-, $n_1$ D2-, $m^1$ D4- and $m^0$ D6-branes. The choice of $X^0$ and $X^1$ as A-periods defines the large radius frame.

The topological string amplitudes $\cF_g$ for the quintic are known up to genus $g_{\text{max}}=64$.
            
\subsubsection{Local $\IP^2$}
Topological string amplitudes on local Calabi-Yau threefolds, particularly those which are toric, are computationally more accessible than on their compact cousins. The holomorphic anomaly equations are simpler (only one rather than three sets of propagators is needed) and more importantly, a sufficient number of boundary conditions can be imposed to obtain the topological string amplitudes $F_g$ up to arbitrary high genus (with the only constraint being computer power and memory). Also, an independent approach, based on the topological vertex \cite{Aganagic_2004}, exists to compute the amplitudes without any need to impose boundary conditions. On the other hand, the theory of generalized Donaldson-Thomas invariants is more laborious to set up, see \cite[Section 6.7]{Joyce}, and the Hodge theory underlying the B-model is conceptually less straightforward. 

After the conifold, local $\IP^2$ is arguably the simplest local Calabi-Yau threefold on which to study the topological string. By $H_2(\IP^2, \IZ) = H_4(\IP^2)= \IZ$, the charge lattice of BPS states on $\KPt$ is three dimensional. The complex structure moduli space of the mirror has three singular points, typically chosen to lie at $z=0$ (the MUM point), $z=z_c = -\tfrac{1}{27}$ (the conifold point), and $z = \infty$ (the orbifold point).

The BPS spectrum on $\KPt$, or equivalently the bounded derived category of compactly supported coherent sheaves on $\KPt$, was studied in detail in \cite{Bousseau:2022snm}. We will follow their conventions in the following. The Chern vector $\gamma(E)$ of a compactly supported coherent sheaf $E$ on $\KPt$ takes the form
\begin{align}
    \gamma(E) = [r,d,\cht] \in \IZ \oplus \IZ \oplus \frac{1}{2}\IZ \,,
\end{align}
where $r$ and $d$ specify rank and degree of the sheaf respectively. The associated central charge in terms of periods of the meromorphic one-form on the mirror curve is given by
\begin{equation}
    Z(r,d,\cht)=-r\, T_D+d\,T- \cht \,,
    \label{central charge}
\end{equation}
where\footnote{The constants are imposed such that the periods of the meromorphic one-form satisfy the boundary conditions $(T,T_D) = (-\tfrac{1}{2},\tfrac{1}{3})$ at the orbifold point \cite{Diaconescu:1999dt}.} 
\begin{align}
   T= \frac{f_1(z)}{2\pi\ii}-\frac{1}{2} \,, \quad 
    T_D= \frac{f_2(z)}{(2\pi\ii)^2}-\frac{f_1(z)}{4\pi \ii}+\frac{1}{4} \,,
\end{align}
with the functions $f_1(z)$, $f_2(z)$ given in terms of Meijer G-functions as
\begin{align}
    &f_1(z)=-\frac{1}{\Gamma(\frac{1}{3})\Gamma(\frac{2}{3})}G^{\,2,2}_{3,3}\!\left(\,\begin{smallmatrix}\tfrac{1}{3}&\tfrac{2}{3}&1\\0&0&0\end{smallmatrix}\;\middle|\;27z\right)\\
    &f_2(z)=\frac{1}{2}\left(G^{\,3,1}_{3,3}\!\left(\,\begin{smallmatrix}\tfrac{1}{3}&\tfrac{2}{3}&1\\0&0&0\end{smallmatrix}\;\middle|\;27z\right)+G^{\,3,1}_{3,3}\!\left(\,\begin{smallmatrix}\tfrac{2}{3}&\tfrac{1}{3}&1\\0&0&0\end{smallmatrix}\;\middle|\;27z\right)\right)-\frac{\pi^2}{3}\,.
\end{align}
If we wish to express the central charge in terms of integer linear combinations of a basis of periods, we can replace $\cht$ in the Chern vector of the sheaf by its integral Euler characteristic, $\chi=r+\frac{3}{2}d+ \cht$. Following \cite{Bousseau:2022snm}, we denote the Chern vector $\gamma(E)$ of a coherent sheaf $E$ expressed in terms of the triplet rank, degree, and Euler characteristic of the as $\gamma(E)=[r,d,\chi)$. The central charge in terms of this triplet becomes
\begin{equation}
  Z[r,d,\chi)=r(-1 -T_D)+d(T+\frac{3}{2})-\chi.
\end{equation}
As motivated in the paragraph below equation \eqref{eq:aleph}, we will normalize the periods such that the constant period takes the value $(2\pi \ii)^{3/2}$. Together with the integrality requirement, this informs our choice of period vector $(\tilde{T}, \tilde{T}_D$, $x^0)$, where
\begin{align}
   & \tilde{T}=(2\pi\ii)^{3/2}(T+\frac{3}{2})\,,\\
    & \tilde{T}_D= (2\pi\ii)^{3/2}(T_D+1)\,,\\
    & x^0= (2\pi\ii)^{3/2}\,.
\end{align}
In this normalization, we will find that a singularity in the Borel plane associated to a charge ${\gamma=[r,d,\chi)}$ will lie at \begin{equation}
    \omega_{\gamma}=\alpha( r \tilde{T}_D-d\tilde{T}+\chi x^0)= 4\pi^2\ii Z_{\gamma} \,,
\end{equation}
with $\alpha$ as defined in \eqref{eq:aleph}.

The standard relation between the genus 0 amplitude $\cF^0$ and B-periods, $P_I = \partial_I\cF^0$, can involve a proportionality factor in the local setting. In our normalization, we have
\begin{equation}
    \tilde{T}_D=3\frac{\partial \cF_0^{\LR}}{\partial \tilde{T}} \,,
\end{equation}
where $\LR$ stands for large radius frame, the frame in which $\tilde{T}$ is chosen as an A-period. We retain this factor for all frames. The expression \eqref{eq:firstInst} is modified by this rescaling \cite{Gu:2022sqc}: we must replace $c^I \rightarrow 3c^I$, where the normalization of $c^I$ is fixed by \eqref{eq:instantonAction}. Looking back at the definition of the differential operator $\twistedDG$ in equation \eqref{eq:twistedDiff} leading to the fundamental relation \eqref{eq:alienZ}, we see that this modification leads to a rescaling of the derivative term in $\twistedDG$,
\begin{equation} \label{eq:twistedDifflocal}
    \twistedDG = 3 g_s c^I \partial_I + \frac{1}{g_s} d_I X^I \,. 
\end{equation}
Explicitly, in the large radius frame, this is
\begin{equation}
    D_{\omega_{\gamma}}^{\LR}=\alpha \left( 3 g_sr\frac{\partial}{\partial \tilde{T}}+\frac{-d\,\tilde{T}+\chi \,x^0}{g_s} \right)\,,
\end{equation}
while in the conifold frame, in which $\tilde{T}_D$ is chosen as an A-period,
\begin{equation}
      D_{\omega_{\gamma}}^\mathrm{c}=\alpha \left(3 g_s d\frac{\partial}{\partial \tilde{T}_D}+ \frac{r\tilde{T}_D+\chi x^0}{g_s} \right) \,.
\end{equation}
Note that to retain the form \eqref{eq:alienComm} of the commutator of two alien derivatives, this requires the pairing $\langle \cdot, \cdot \rangle$ between two charges to also be rescaled by a factor of 3. This precisely coincides with the definition of the pairing given in \cite{Bousseau:2022snm}:
\begin{equation}
    \langle [r,d,\chi),[r',d',\chi')\rangle=3(rd'-dr') \,.
\end{equation}

\subsection{Numerical tools and conventions}\label{sec:numtools}
\paragraph{Padé approximation}
Our numerical approach to resurgence requires studying the singularities of functions of the form \begin{equation}
    \hphi(\zeta)=\sum_{n\geq0}\frac{a_n}{n!}\zeta^n
\end{equation}
of which the coefficients are known only up to some maximum order $N$. The truncated series $\hphi_N(\zeta) =\sum_{n=0}^N\frac{a_n}{n!}\zeta^n$, being a polynomial, does not exhibit any singularities. A very efficient way to numerically predict the singularities of the full series $\hphi$ is to approximated it as the quotient of two polynomials $P_m$ and $Q_n$ of degree $m$ and $n$,
\begin{equation}
    \text{Padé}_{(m,n)}(\hphi)(\zeta)=\frac{P_m(\zeta)}{Q_n(\zeta)}\,,
\end{equation}
where the coefficients of $P_m$ and $Q_n$ are chosen so that the $\hatphi$ and the Taylor expansion of $\text{Padé}_{m,n}(\hatphi)$ coincide up to order $m+n$. The quotient is called the Padé approximant of $\hatphi$ of order $(m,n)$. In our applications, we typically choose $m=n$.

The roots of the Padé approximant aggregate around the poles and the branch cuts of $\hphi$. In the many plots that follow, we depict the $\zeta$ plane and mark the roots of the denominator $Q_n$ of the Padé approximant of the topological string amplitude as red points. We refer to such figures as depicting the Borel plane of the topological string amplitude.

\paragraph{Computing Stokes constants}

Assume that the charge $\gamma$ is indivisible. From equation \eqref{eq:defFgammas}, we see that we can compute the Stokes constant at $\gamma$ in terms of the alien derivative $\Delta_{\omegaG}(\cFpert)$ and the instanton correction $\cF^{(\gamma)}$ via\footnote{Note that compared to equation \eqref{eq:defFgammas}, expressing this relation in terms of the undotted alien derivative permits us to drop the exponential symbol in the denominator, cf. equation \eqref{eq:dottedUndotted} in the appendix.}
\begin{equation} \label{eq:solvingForS}
    S_{\gamma}=\frac{\DeltagammaN{}(\cFpert)}{g_s^{-2} \cF^{(\gamma)}}\,.
\end{equation}
As we have reviewed above, we can compute $\cF^{(\gamma)}$ in a $g_s$ expansion given the knowledge of the perturbative topological string amplitudes $\cFpert_g$. Given this knowledge, we can calculate the alien derivative numerically in multiple ways.\\

\begin{enumerate}[label=Method \arabic*:,leftmargin=*,itemindent=4em,ref=Method \arabic*] 
\item As we review in Appendix \ref{app:DeltaplusAndStokesAuto}, the Stokes automorphism $\stokesAut_\ray$ along a ray $\ray$ is the exponential of the sum of pointed alien derivatives evaluated at the singularities on $\ray$. Assuming that $\omega_\gamma$ and its multiples are the only singularities that lie on the ray $\ray_{\omega_\gamma}$, we have
\begin{align}
    \stokesAut_{\ray_{\omega_\gamma}} = \id + \dotDeltagammaN{} + \cO(\e^{-2\omega_\gamma/g_s}) \,.
\end{align}
Calculating the discontinuity of the Borel resummation $\sS$ (defined informally in the introduction and formally in equation \eqref{eq:BorelResum} of the appendix) of $\cFpert$ below and above\footnote{We follow the conventions of \cite{Sauzin} in which $\theta^+ < \theta < \theta^-$ (`-' indicating to the left, `+' to the right), and \mbox{$\disc_\theta = \sS^{\theta^+} - \sS^{\theta^-}$}.} the ray $\ray_{\omega_\gamma}$,
\begin{align}
    \disc_{\theta_{\omega_\gamma}}(\cFpert)= \sS^{\theta^-_{\omega_\gamma}}\left( \stokesAut_{\ray_{\omega_\gamma}} - \id\right)(\cFpert) \,,
\end{align}
thus yields the approximation
\begin{align} \label{eq:discApprox}
    \e^{-\tfrac{\omega_\gamma}{g_s}} \sS^{\theta^{-}_{\omega_\gamma}}\left(\Delta_{\omega_\gamma}(\cFpert) \right)\simeq \disc_{\theta_{\omega_\gamma}}(\cFpert) \,.
\end{align}
Note that we have expressed the left-hand side of \eqref{eq:discApprox} in terms of the undotted alien derivative in order to emphasize the leading asymptotic behavior of the discontinuity, cf. \eqref{eq:dottedUndotted} followed by \eqref{eq:symbolvsExponential}. Borel transforming also the denominator of \eqref{eq:solvingForS} thus yields
\begin{equation}
    S_{\gamma}\simeq  \frac{\disc_{\theta_{\omega_\gamma}}(\cFpert)}{\e^{-\frac{\omega_{\gamma}}{g_s}}g_s^{-2}\sS^{\theta_{\omega_\gamma}^-}(\cF^{(\gamma)})}
\end{equation}
for small values of $g_s$. Numerically, the Stokes discontinuity is obtained by summing over the residues of the poles of $\text{Padé}(\hatFpert)$ which lie on the ray $\ray_{\theta_{\gamma}}$. Depending on the degree of precision required, $\sS^{\theta_{\omega_\gamma}^-}(\cF^{(\gamma)})$ can be approximated by its asymptotic expansion $\cF^{(\gamma)}$, or by the numerical Laplace transform of the Padé approximant of $\hatFpert$. 

\item \label{method2}If $\omega_{\gamma}$ is a dominant singularity, i.e. if there is no other singularity $\omega_{\gamma'}$ with $\vert \omega_{\gamma'}\vert \simeq \vert\omega_{\gamma}\vert$, then it governs the asymptotics of $\cFpert_g$ by the discussion in Appendix \ref{app:asympCoeff}. Keeping track of the shift $m=2g-3$ when comparing the indexing of the coefficients $\cF_g$ of $\cFpert$ with that of the coefficients $c_m$ of $\tildephi$ as in Appendix \ref{app:asympCoeff}, and the shift $h=k+2$ when comparing the indexing of the coefficients $\cF^{(\gamma)}_h$ of $\cF^{(\gamma)}$ introduced in \eqref{eq:FgammaCoeffs} and that of the coefficients $c^\omega_k$ introduced in \eqref{eq:defckomega}, we arrive at\footnote{Such so-called large order relations were recently studied in detail in \cite{Marinissen:2023ttp}.}
\begin{equation} \label{eq:Fgasymp}
    \cFpert_g\sim -\frac{S_{\gamma}}{2\pi \ii}\sum_{h=0}^N\frac{(2g-2-h)!}{(\omega_{\gamma})^{2g-1-h}}\cF^{(\gamma)}_h 
\end{equation}
for $g \gg 1$ and $N$ fixed. To numerically determine $S_\gamma$, we consider the sequence
\begin{equation}
   s_{g,N}= -\frac{2\pi \ii \,\cFpert_g}{\sum_{h=0}^{N}\frac{(2g-2-h)!}{(\omega_{\gamma})^{2g-1-h}}\cF^{(\gamma)}_h}
\end{equation}
for some $N < g$. The quotient $s_{g,N}$ will approach $S_{\gamma}$ when $g$ becomes large. We can improve the estimate by using Richardson extrapolation.

As discussed in Appendix \ref{app:asympCoeff}, the right-hand side of \eqref{eq:Fgasymp}, divided by $(2g-3)!$, is an approximation to the integral $I(\cH_{\omega_\gamma})$,
\begin{align} 
    I(\cH_{\omega_\gamma}) = -\frac{1}{2\pi \ii}\sum_{k=-2,-1}  \frac{(2g-4-k)!}{(2g-3)!}\frac{\hatF^{\omega_\gamma}_{k}}{\omega_{\gamma}^{2g-3-k}} - \frac{1}{2\pi \ii} \int_{\omega_{\gamma}}^{R\e^{\ii \theta_{\omega_\gamma}}} \frac{\hatF^\omegaG(\zeta-\omega_{\gamma})}{\zeta^{2g-2}} d\zeta \,,
\end{align}
around the Hankel-like contour $\cH_{\omega_\gamma}$ depicted in Figure \ref{fig:asymCoeff} in the appendix. If the error terms $\cE_\infty$ and $\cE_{r_{N+1}}$ underlying this approximation, introduced in \eqref{def:Einf} and \eqref{def:Ernp1} prove too large, we can instead approximate this integral numerically by modeling $\hatF^\omega$ by its Padé approximant, such that
\begin{align} \label{eq:asympPade}
    I(\cH_{\omega_\gamma}) \sim -\frac{1}{2\pi \ii}\sum_{k=-2,-1}  \frac{(2g-4-k)!}{(2g-3)!}\frac{\hatF^{\omega_\gamma}_{k}}{\omega_{\gamma}^{2g-3-k}} -  \frac{1}{2\pi \ii} \int_{\omega_{\gamma}}^{R\e^{\ii \theta_{\omega_\gamma}}} \frac{\pade\left(\hatF^{\omega}(\zeta-\omega_\gamma)\right)}{\zeta^{2g-2}} d\zeta \,,
\end{align}
In practice, the approximation of $I(\cH_{\omega_{\gamma}})$ underlying \eqref{eq:Fgasymp} proved sufficient for many purposes. We have recourse to \eqref{eq:asympPade} however when subtracting asymptotics, see below.

\item \label{method3} When multiple singularities with different Stokes constants contribute comparably to the asymptotics of the perturbative amplitudes $\cFpert_g$, this asymptotics cannot directly be used to extract Stokes constants. Instead, we can leverage the two facts that 
\begin{enumerate} 
    \item the Padé approximant of $\hatFpert$ approximates the analytic continuation of $\hatFpert$ around each of the singularities,
    \item for any singularity $\omega_\gamma$, we can compute the coefficients $\hatF^{\omega_{\gamma}}_k$, which upon multiplication by the Stokes constant determine the analytic continuation of $\hatFpert$ in a vicinity of the singularity, to a given order.
\end{enumerate}
We can then approximate the integral 
\begin{align}
    I_{g,\omegag}(\cFpert) = \frac{(2g-3)!}{2\pi\, \ii} \int_{\cH_{\omegag}} \frac{\hatFpert(\zeta)}{\zeta^{2g-2}}d\zeta
\end{align}
(where by abuse of notation, we are identifying $\hatFpert$ and its analytic continuation to a neighborhood of $\omegag$) in two ways:
\begin{enumerate}
    \item By modeling $\hatFpert(\zeta)$ via its Padé approximant. As the Padé approximant models the logarithm via a series of poles, this amounts to summing the residues of the integrand along the interval enveloped by $\cH_{\omegag}$:
    \begin{align} \label{eq:intGomegaPade}
        I^{\mathrm{P}}_{g,\omegag}(\cFpert) = \frac{(2g-3)!}{2\pi\, \ii} \int_{\cH_{\omega_{\gamma}}} \frac{\pade(\hatFpert)}{\zeta^{2g-2}}d\zeta = -  (2g-3)!\sum_{\substack{\mathrm{poles}\, \zeta_i \\ \arg(\zeta_i) \simeq \arg(\omega_{\gamma})}} \res_{\zeta_i}\left(\frac{\pade(\hatFpert)}{\zeta^{2g-2}}\right) \,.
    \end{align}
    Here, the superscript `P' stands for `Padé'.
    \item By approximating $\hatFpert(\zeta)$ via a truncation of 
    \eqref{eq:FhatNearSing}. Following the steps in Appendix \ref{app:asympCoeff}, this yields
    \begin{align}
        I^{\mathrm{I},\,N}_{g,\omegag}(\cFpert) = -\frac{1}{2\pi \ii}\sum_{h=0}^N\frac{(2g-2-h)!}{(\omega_{\gamma})^{2g-1-h}}\cF^{(\gamma)}_h \,,
    \end{align}
\end{enumerate}
where the superscript `I' stands for `instanton'. As argued in Appendix \ref{app:asympCoeff}, the approximation
\begin{align}
    I^{\mathrm{P}}_{g,\omegag}(\cFpert) \simeq S_{\gamma} \, I^{\mathrm{I},\,N}_{g,\omegag}(\cFpert)
\end{align}
at fixed $N$ improves with increasing $g$, such that the sequence
\begin{align}
    s_{g,N,\omega_{\gamma}} = \frac{ I^{\mathrm{P}}_{g,\omegag}(\cFpert)}{I^{\mathrm{I},\,N}_{g,\omegag}(\cFpert)}    
\end{align}
should approach $S_\gamma$ as $g$ becomes large.
\end{enumerate}

\paragraph{Subtracting singularities} To study subleading singularities, we need to subtract more dominant ones. Based on the asymptotic formula \eqref{eq:Fgasymp}, the most straight-forward subtraction of leading asymptotics due to a singularity at $\omega_\gamma$ we can perform is by defining reduced amplitudes 
\begin{align} \label{eq:subtractAsymp}
    \FgredN = \cFpert_g + \frac{S_{\gamma}}{2\pi \ii}\sum_{h=0}^N\frac{(2g-2-h)!}{(\omega_{\gamma})^{2g-1-h}}\cF^{(\gamma)}_h \,.
\end{align}
At the level of the Borel transform $\hatFpert$, this corresponds to truncating the coefficient of the $\log$ singularity at $\omega_\gamma$ after a finite number of terms. This is reasonable, as leading order perturbative terms of an $(n+1)$-instanton contribution can numerically dominate higher order perturbative terms of an $n$-instanton contribution.

For our study of compact Calabi-Yau threefolds in \cite{Douaud:2024khu}, the subtraction \eqref{eq:subtractAsymp} worked remarkably well, in that the singularity at $\omega_{\gamma}$ was no longer visible in the Borel plane of $\FgredN$ for many examples. This procedure fails however for local $\IP^2$. Following the discussion under \ref{method2} above, what does work is approximating the asymptotics via 
\begin{align}\label{eq:exactasympt}
  \cF^{\mathrm{asymp}}_g= \frac{1}{2\pi\ii}\sum_{h=-k_{\omega}}^{-1}\frac{(2g-4-h)!}{\omega^{2g-3-h}}(\Delta_{\omega}^+\cFpert)_h+\frac{(2g-3)!}{2\pi \ii}\int_{\ray_{\omega}^-}\frac{(\mu\circ\Delta_{\omega}^+)\hatFpert(\zeta)}{(\zeta+\omega)^{2g-2}}d\zeta
\end{align}
with $\ray_{\omega}^-$ a contour slightly to the left of $d_{\omega}$. We are here using the notation introduced in \eqref{eq:notationCoeffSingPoles} and \eqref{eq:projection}. We then define
\begin{align}
    \Fgred = \cFpert - \cF_g^{\mathrm{asymp,\,P}} \,,
\end{align}
with $\cF_g^{\mathrm{asymp,\,P}}$ obtained from \eqref{eq:exactasympt} by replacing $(\mu\circ\Delta_{\omega}^+)\hatFpert(\zeta)$ by its Padé approximant.

$(\mu\circ\Delta_{\omega}^+)\hatFpert(\zeta)$ will typically have branch points and poles on the ray $\ray_\omega$. When we subtract beyond the leading asymptotics, it is important, as we have chosen to integrate along $\ray^-_{\omega}$ to the left of $\ray_\omega$, that we analytically continue to the right of all singularities encountered along $\ray_\omega$ before evaluating the alien derivative at $\omega$, so that our path of analytic continuation does not cross the integration contour.

\hspace{0.2cm}

\textit{Note: }Subtracting singularities requires the knowledge of $S_{\omega_{\gamma}}$. The subtraction is numerically very sensitive to the correct choice of $S_{\omega_{\gamma}}$. Typically, the subtraction only leads to the targeted singularity disappearing from the Borel plane if the integer invariants from which $S_{\omega_{\gamma}}$ is computed are identified precisely. If we have a prediction for the value of $S_{\omega_{\gamma}}$, we can thus verify it by checking whether the subtraction procedure with this choice succeeds.

\subsection{Local $\IP^2$ at large radius: beyond $n_0^d$}\label{sec : BPSnum}
At large radius in the large radius frame, the topological string amplitude is captured by the Gopakumar-Vafa formula \cite{Gopakumar:1998jq}. By considering the Laurent/Taylor series for the powers of the sine function occurring in this formula, one arrives at the following expression for the topological string amplitude at genus $g$ in terms of Gopakumar-Vafa invariants (see e.g. \cite[Sections 3.2.1 and 5.7.1]{Gu:2023mgf}:
\begin{equation}
    \cFpert_g=\sum_{d>0,n\in \IZ}[n^d_0\frac{B_{2g}\alpha^{2g-2}}{2g(2g-2)}+\mathcal{T}_{g,d,n}]\frac{1}{(d \tilde{T}+n x^0)^{2g-2}}
    \label{eq:GVformula}
\end{equation}
with $\mathcal{T}_{g,d,m}$ given by
\begin{equation}
   \mathcal{T}_{g,d,n}=\alpha^{2g-2}(2g-3)!(-1)^{g-1}[\frac{2(-1)^{g}n^d_2}{(2g-2)!}+...-\frac{g-2}{12}n^d_{g-1}+n^d_g] \,.
\end{equation}
The factorial behavior $B_{2g} \sim (2g)!$ of the Bernoulli numbers implies that the asymptotics of $\cFpert_g$ will be governed by the genus 0 GV invariants \cite{Gu:2023mgf}. The large $g$ asymptotics of the remaining contribution
\begin{align}
    \mathcal{T}_g=\sum_{d>0,m\in \IZ}\mathcal{T}_{g,d,m}\frac{1}{(d\, \tilde{T}+m \,x^0)^{2g-2}}
\end{align}
is more difficult to tackle due to the $g$-dependence of $n^d_g$. In this subsection, we will identify singularities in the Borel plane beyond those predicted by the leading Bernoulli dependence of $\cFpert_g$, thus demonstrating that also $\cT_g$ is encoded in the Borel plane.

We will work at $z=10^{-6}$ in the large radius frame. Subtracting the leading singularities, which arise at $\propto  \tilde{T}+n \,x^0$, $n =-2,\ldots,4$, and are associated to the GV invariants $n_0^d$,\footnote{Note that in all the following, we work with amplitudes $\cF_g$ for $\KPt$ from which the constant map contribution has been subtracted, see e.g. \cite{Gu:2022sqc}.} allows us to uncover four new singularities associated to the charges
\begin{equation}
   \gamma_1=[1,-2,0)\,,\quad \gamma_2=[1,-1,0)\,,\quad \gamma_3=[1,0,1)\,, \quad \gamma_4=[1,1,3) \,.
 \end{equation}
We numerically determine that the associated Stokes constants all coincide and equal 1. Upon subtracting these, we uncover a new set of singularities associated to the charges
\begin{align}
    {\gamma_{i,n} = \gamma_i - n[0,0,1)}\,, \quad i =1, \ldots, 4
\end{align} 
at $n=1$. We can repeat this procedure another two times before the loss of precision washes out the result. We record the numerically determined Stokes constants in Table \ref{tab:StokesHilbert}.
\begin{table}[h]
    \centering
    \begin{tabular}{c|c|c|c|c}
    & $\gamma_{i,1}$ & $\gamma_{i,2}$ & $\gamma_{i,3}$ & $\gamma_{i,4}$ \\ \hline
    $S_{\gamma_{i,n}}$ & 1 & 3 & 9 & 22 
    \end{tabular} \,.
    \caption{The Stokes constants of the singularities associated to the charges $\gamma_{i,n}$, for $i=1,\dots, 4$.} \label{tab:StokesHilbert}
\end{table}
The Borel planes upon each subtraction are depicted in Figures \ref{First and second layer} and \ref{third and fourth layer}. 

The charge $\gamma_3 - n[0,0,1) = [1,0,1-n)$ describes a bound state of a D4-brane with $n$ D0-branes. The moduli space can be identified with the Hilbert scheme of $n$ points (intuitively the position of the D0-branes) on $\IP^2$. The associated Poincaré polynomial was computed in \cite{Goettsche} and reproduces the values in Table \ref{tab:StokesHilbert}. These numbers were also reproduced via the scattering diagram method in \cite{Bousseau:2022snm}.

As the four charges $\gamma_{i,n}$, $i=1,\ldots,4$ at fixed $n$ are related via successive tensoring by a line bundle $\cO(1)$ (spectral flow),\footnote{Tensoring by $\cO(n)$ shifts the charge via $[r,d,\chi)\rightarrow [r,d+nr,\chi+n(d+\frac{3}{2}r)+\frac{n^2}{2}r)$.} their DT invariants coincide.

All numerically determined Stokes constants thus coincide with the respective DT invariants in the large radius chamber:
\begin{equation}
     S_{\omega_{\gamma_{i,n}}}=\bar{\Omega}(\gamma_{i,n},z=0),\quad   i=1,\ldots,4\,,\,\,\,n=0,\ldots,3 \,.
\end{equation}
\begin{figure}[H]
    \centering
    \includegraphics[scale=0.4]{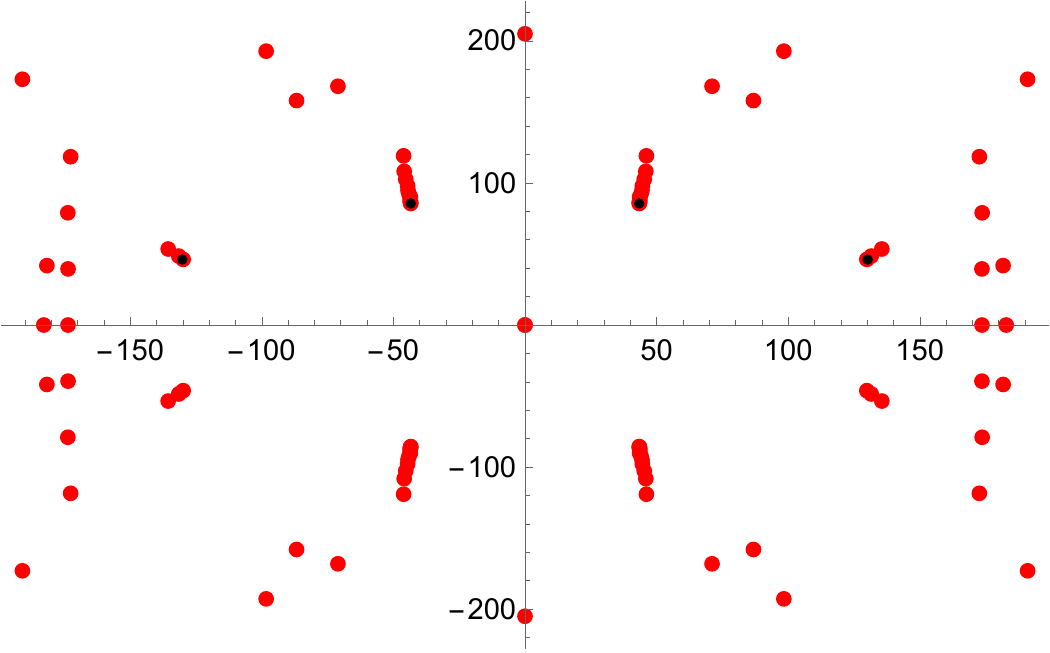}
    \includegraphics[scale=0.4]{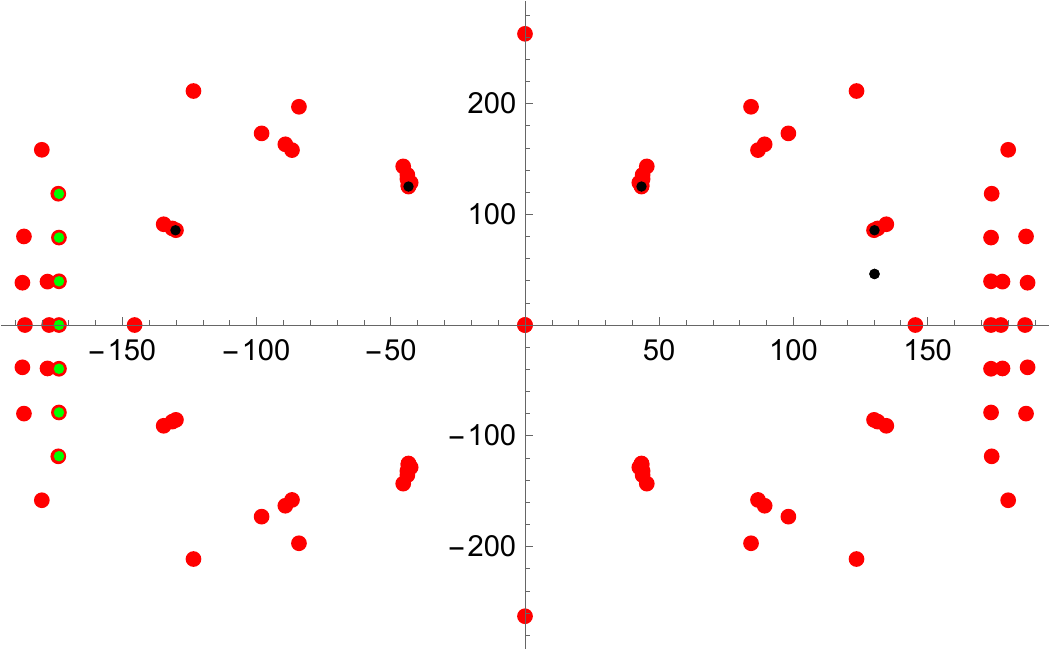}
    \caption{The Borel plane of local $\IP^2$ at $z=10^{-6}$ upon subtraction of leading (left) and leading and subleading (right) singularities. Black dots indicate the points $\omega_{\gamma_{i,n}}$ with $n=0$ on the left and $n=1$ on the right. Green dots are associated to D2-D0 charges $[0,d,n)$.}
    \label{First and second layer}
\end{figure}
 \begin{figure}[H]
    \centering
    \includegraphics[scale=0.4]{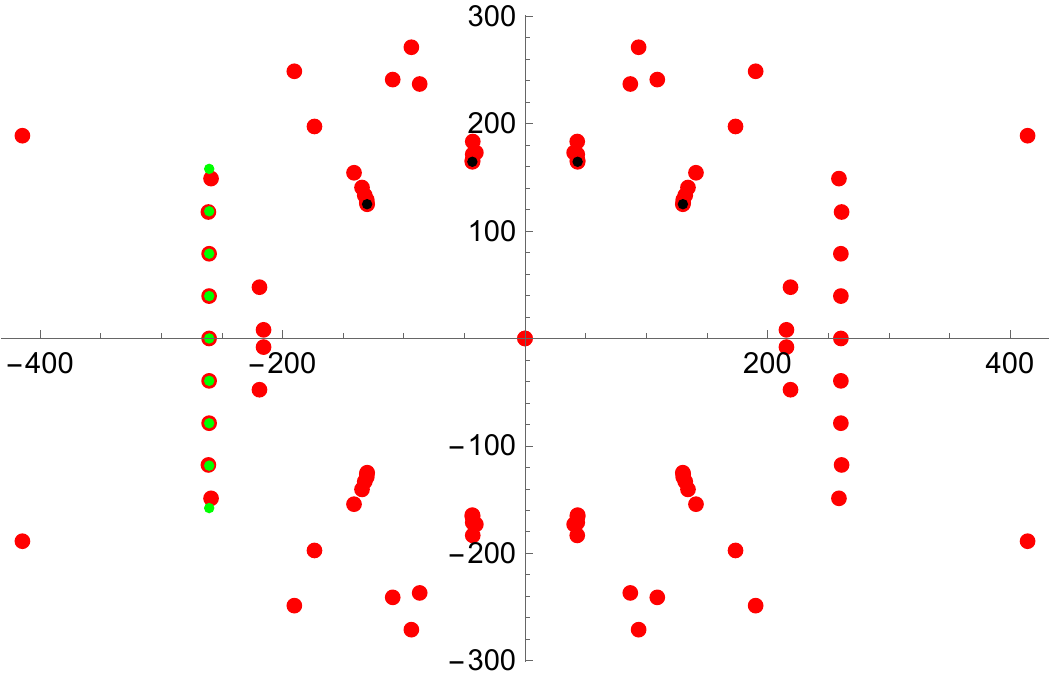}
    \includegraphics[scale=0.4]{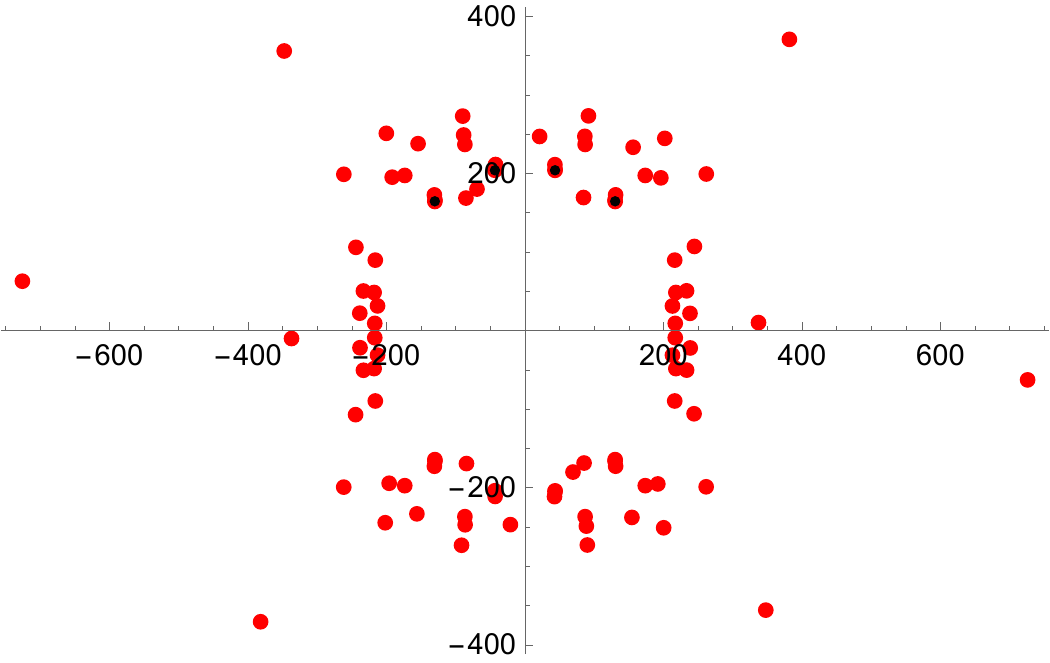}
    \caption{The Borel plane of local $\IP^2$ at $z=10^{-6}$ upon further subtractions. Black dots indicate the points $\omega_{\gamma_{i,n}}$ with $n=2$ on the left and $n=3$ on the right.  Green dots are associated to D2-D0 charges $[0,d,n)$.}
    \label{third and fourth layer}
\end{figure}

\subsection{$\quintic$ at large radius: beyond $n_0^d$}
We can repeat the computation of the previous section in the compact case as well \footnote{The Gopakumar-Vafa formula  for the quintic follows from the replacement $(\tilde{T},x^0)\rightarrow (X^1,X^0)$ in \eqref{eq:GVformula}. }. We consider the point $z=10^{-1}\mu$ in the complex structure moduli space of $\mirrorQ$ in the large radius frame. The left panel of Figure \ref{fig:X5LRp0} depicts the Borel plane before subtracting the asymptotics due to the singularities at $\omega_{d,n} = \alpha (d X^1+n X^0)$; the dominant singularities on the imaginary axis are the constant map contributions at  $d=0,n=\pm1$. Upon subtraction (right panel), the singularity at $\alpha P_0$ becomes more pronounced. Unfortunately, our numerical precision does not suffice to explore beyond the $\alpha P_0$ singularity.
\begin{figure}[H]
    \centering
    \includegraphics[scale=0.5]{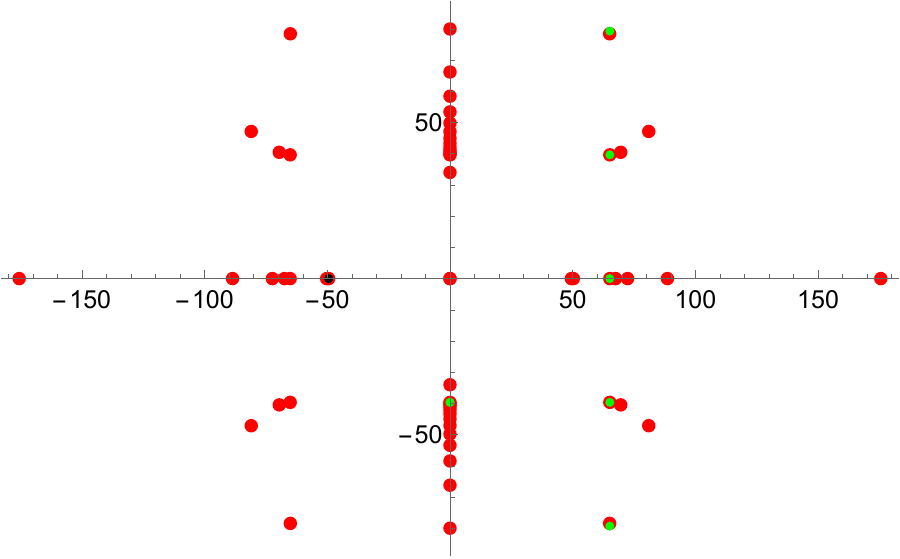}
    \includegraphics[scale=0.5]{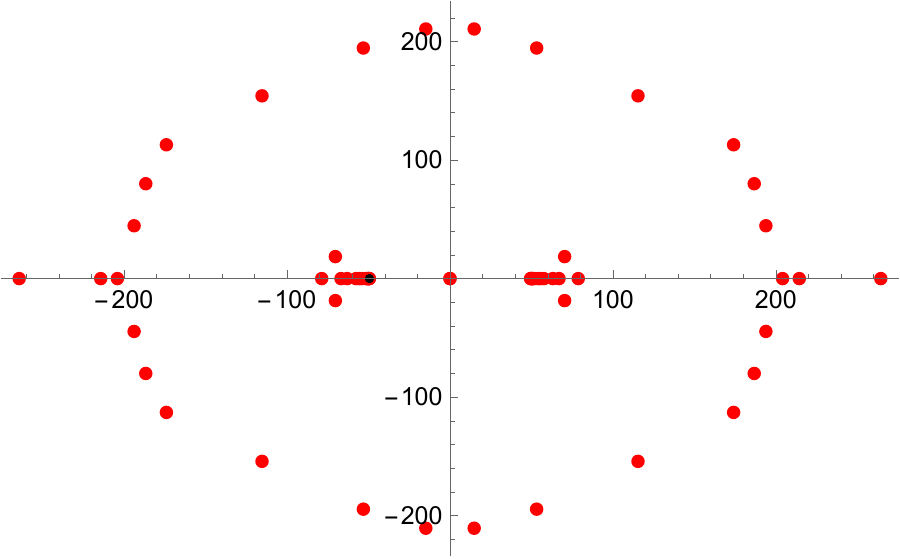}
    \caption{Borel plane of $\mirrorQ$ in the large radius frame at $z=10^{-1}\mu$ before (left) and after(right) subtracting the leading asymptotics. Green points indicate singularities of the form $\omega_{d,n}$: $(d,n)=(0,1)$ lies on the imaginary axis, $(1,n)$, $n= -2 ,\dots, 2$, lie beyond the $\alpha P^0$ point on the real axis. The black point lies at $\alpha P_0$. The singularity at this point becomes more pronounced upon subtracting those at the green points.}
    \label{fig:X5LRp0}
\end{figure}

\subsection{Numerical tests of the multi-instanton structure}
In this section, we will numerically test the results presented in the previous section regarding successive applications of alien derivatives. The multi-instanton sector of $\KPt$ was already studied in \cite{Couso-Santamaria:2014iia, Gu:2022sqc}.

\subsubsection{Local $\IP^2$ around conifold point: testing three-instanton corrections}
We begin by considering successive alien derivatives which are all evaluated at multiples of the same primitive singularity. This is not novel, as it amounts to e.g. computing higher instanton corrections to the Stokes automorphism. In fact, we will push the computation in \cite[Section 4.1]{Gu:2022sqc} of the discontinuity due to the conifold singularity at $\gamma_c = [1,0,1)$ from two to three instanton order. Operationally, however, we are using the differential operators $\alienDiffOp_\omega$ introduced in \eqref{eq:alienDiffOp} to compute the instanton amplitudes, so the analysis in this subsection can be seen as a warm-up exercise to the novel computations in Section \ref{s:resurgenceOneInstanton}, in which we consider successive applications of $\alienDiffOp_\omega$ at unaligned singularities to study the Borel plane of a one-instanton amplitude $\cF^{(\gamma)}$. 

We will work in the vicinity of the conifold point at $z=\frac{4}{3}z_c$ and in the large radius frame. The leading singularity here is associated to the charge $\gamma_c=[1,0,1)$, followed by those associated to $2\gamma_c$ and 3$\gamma_c$. We are interested in studying the higher instanton contributions to the Stokes discontinuity along $\ray_{\omega_{\gamma_c}}$ stemming from the charges $2\gamma_c$ and $3\gamma_c$. By \eqref{eq:discontinuity},  \eqref{eq:StokesAutoDeltaPlusExp} and \eqref{eq:Dexplike}, we have
\begin{align} \label{eq:discCfld3}
    \disc_{\theta_{\gamma_c}}\cFpert =\sS^{\theta_{\gamma_c}^-}\Big(& \big(\dotDeltagammaN{c}+\frac{1}{2}\dotDeltagammaN{c}^2+\dotDeltagammaKN{c}{2}+\frac{1}{6}\dotDeltagammaN{c}^3+ \\ 
    &+\frac{1}{2}(\dotDeltagammaKN{c}{2}\dotDeltagammaN{c}+\dotDeltagammaN{c}\dotDeltagammaKN{c}{2})+  \dotDeltagammaKN{c}{3}\big)\cFpert \Big)+O(\e^{\frac{-4\omega_{\gamma_c}}{g_s}}) \,. \nonumber
\end{align}
We can evaluate the successive alien derivatives via \eqref{eq:multiDeltaZ} up to knowledge of the Stokes constants. These can be e.g. extracted from the gapped form of the topological string amplitudes \cite{Couso-Santamaria:2013kmu,Couso-Santamaria:2014iia}: the only primitive contribution occurs at charge $\gamma_c$, where $S_{\gamma_c} = 1$. Hence, $S_{k\omega_{\gamma_c}} = \tfrac{1}{k^2}$ for $k\in \IN^*$. We can thus express \eqref{eq:discCfld3} in terms 
 of the functions $\cF^{(\gamma_1\vert...\vert \gamma_n)}$ introduced in \eqref{eq:defFgammas}:
\begin{align}
      \disc_{\theta_{\gamma_c}}\cFpert=\sS^{\theta_{\gamma_c}^-}\Bigg(&\frac{\e^{-\frac{\omega_{\gamma_c}}{g_s}}}{g_s^2}\cF^{(\gamma_c)}+\frac{\e^{-\frac{2\omega_{\gamma_c}}{g_s}}}{g_s^3}\Big(\frac{\cF^{(\gamma_c\vert\gamma_c)}}{2}+g_s\frac{\cF^{(2\gamma_c)}}{4}\Big)\\
     &+ \frac{\e^{-\frac{3\omega_{\gamma_c}}{g_s}}}{g_s^4} \Big(\frac{\cF^{(\gamma_c\vert\gamma_c\vert \gamma_c)}}{6}+g_s\frac{\cF^{(2\gamma_c\vert \gamma_c)}}{4}+g_s^2\frac{\cF^{(3\gamma_c)}}{9}\Big)\Bigg)+O(\e^{\frac{- 4\omega_{\gamma_c}}{g_s}})\,. \nonumber
\end{align}
To study the 2 instanton contribution stemming from the singularity at $2\omega_{\gamma_c}$, we consider the quotient
\begin{equation}
  Q_2(g_s)=  \frac{\disc_{\theta_{\gamma_c}}\cFpert-\sS^{\theta_{\gamma_c}^-}(\frac{\e^{-\frac{\omega_{\gamma_c}}{g_s}}}{g_s^2}\cF^{(\gamma_c)})}{\sS^{\theta_{\gamma_c}^-}(\frac{\e^{-\frac{2\omega_{\gamma_c}}{g_s}}}{g_s^3}(\frac{\cF^{(\gamma_c\vert\gamma_c)}}{2}+g_s\frac{\cF^{(2\gamma_c})}{4}))} \,,
\end{equation}
which satisfies $Q_2=1+O(\e^{-\frac{\omega_{\gamma_c}}{g_s}})$. We show the numerical values that we obtain for $Q_2$ for a wide range of values of $g_s$ in Table \ref{tab:Q2gs}. Our numerical error proves smaller than the theoretically expected deviation from the value 1 due to ignoring three instantons.
\begin{table}[h]
\centering
\begin{tabular}{c|c|c}
$g_s\ee^{-\ii\theta_{\gamma_c}}$ & $Q_2(g_s)$ & order of $\ee^{-\frac{\omega_{\gamma_c}}{g_s}}$ \\
\hline
$\frac{1}{2}$  & $0.990728 + 0.00310664\,\ii$         & $10^{-1}$ \\
$\frac{1}{5}$  & $0.999553 + 0.000129395\,\ii$        & $10^{-3}$ \\
$\frac{1}{10}$ & $0.999996 + 1.24626\times10^{-6}\,\ii$ & $10^{-5}$ \\
$\frac{1}{15}$ & $1. + 1.10194\times10^{-8}\,\ii$     & $10^{-7}$ \\
$\frac{1}{20}$ & $1. + 9.05822\times10^{-11}\,\ii$    & $10^{-9}$ \\
$\frac{1}{25}$ & $1. + 7.08262\times10^{-13}\,\ii$    & $10^{-11}$ \\
$\frac{1}{30}$ & $1. - 7.45886\times10^{-14}\,\ii$    & $10^{-13}$ \\
\hline
\end{tabular}
\caption{Values of $Q_2(g_s)$ for $g_s =\frac{1}{k}\e^{\ii\theta_{\gamma_c}}$.}
\label{tab:Q2gs}
\end{table}
We can now study the 3 instanton contribution stemming from the singularity at $3\omega_{\gamma_c}$ by evaluating the quotient 
\begin{equation}
    Q_3(g_s)= \frac{\disc_{\theta_{\gamma_c}}\cFpert-\sS^{\theta_{\gamma_c}^-}(\frac{\e^{-\frac{\omega_{\gamma_c}}{g_s}}}{g_s^2}\cF^{(\gamma_c)})+\frac{\e^{-\frac{2\omega_{\gamma_c}}{g_s}}}{g_s^3}(\frac{\cF^{(\gamma_c\vert\gamma_c)}}{2}+g_s\frac{\cF^{(2\gamma_c})}{4}))}{\sS^{\theta_{\gamma_c}^-}(\frac{\e^{-\frac{3\omega_{\gamma_c}}{g_s}}}{g_s^4}(\frac{\cF^{(\gamma_c\vert\gamma_c\vert \gamma_c)}}{6}+g_s\frac{\cF^{(2\gamma_c\vert \gamma_c)}}{4}+g_s^2\frac{\cF^{(3\gamma_c)}}{9}))}.
\end{equation}
Again, this quotient should satisfy $Q_3=1+O(\e^{-\frac{\omega_{\gamma_c}}{g_s}})$. From the values recorded in Table \ref{tab:Q3gs}, we see that $Q_3 \sim 1$; however, the numerical error is now larger than the correction expected from the 4 instanton contribution.
\begin{table}[h]
\centering
\begin{tabular}{c|c }
$g_s\ee^{-\ii\theta_{\gamma_c}}$ & $Q_3(g_s)$ \\
\hline
$\frac{1}{5}$  & $0.998178 - 0.00480932\,\ii$ \\
$\frac{1}{10}$ & $1.00074 - 0.00320779\,\ii$ \\
$\frac{1}{15}$ & $1.00105 - 0.00247649\,\ii$ \\
$\frac{1}{20}$ & $1.00098 - 0.00217809\,\ii$ \\
$\frac{1}{25}$ & $1.03234 + 0.0107923\,\ii$ \\
\hline
\end{tabular}
\caption{Values of $Q_3(g_s)$ for $g_s = \frac{1}{k}\ee^{\ii\theta_{\gamma_c}}$.}
\label{tab:Q3gs}
\end{table}

A second test of the the instanton contributions from $2\omega_{\gamma_c}$ and $3 \omega_{\gamma_c}$ is provided by the subleading asymptotics of $\cF_g$. In order to see the contribution from $2 \omega_{\gamma_c}$, we first subtract the leading asymptotics due to the singularity at $\omega_{\gamma_c}$ following the discussion in Section \ref{sec:numtools}. This yields the sequence $\Fgredn{1}$, whose leading asymptotics is governed by 
\begin{equation}
    \cF_{2\gamma_c}=g_s^2\Delta_{2\omega_{\gamma_c}}^+(\cF)=g_s S_{2\omega_{\gamma_c}}\cF^{(2\gamma_c)}+\frac{1}{2} S_{\omega_{\gamma_c}}^2\cF^{(\gamma_c \vert \gamma_c)}=\frac{g_s}{4}\cF^{(2\gamma_c)}+\frac{1}{2}\cF^{(\gamma_c \vert \gamma_c)}
\end{equation}
via the formula
\begin{equation}
    \Fgredn{1} \sim- \frac{1}{\pi\ii}\sum_{h\geq 0}\frac{(2g-1-h)!}{(2\omega_{\gamma_c})^{2g-h}}(\cF_{2\gamma_c})_h \,.
\end{equation}
We can test the contribution of $\cF_{2\gamma_c}$ to the asymptotics by considering the sequence
\begin{equation}
    s^{(2)}_{g,h}=-\frac{\pi \ii(2\omega_{\gamma_c})^{2g-h}}{(2g-h-1)!(\cF_{2\gamma_c})_h}(  \Fgredn{1} +\frac{1}{\pi \ii}\sum_{k=0}^{h-1}\frac{(2g-k-1)!}{(2\omega_{\gamma_c})^{2g-k}}(\cF_{2\gamma_c})_k)
    \label{eq:s2gh} \,,
\end{equation}
which should approach 1 as $g$ grows large. We use Richardson extrapolation to accelerate this approach. Table \ref{tab:sh2} records the result of this extrapolation for the first 16 values of $h$, with the Richardson extrapolation of the limit denoted as $\widehat{s^{\smash[t]{\scriptscriptstyle(2)}\vphantom{\rule{0pt}{0.5ex}}}_{g_{\mathrm{max},h}}}$. At large $g$, we have $s^{(2)}_{g,h}-1=\frac{2 \omega_{\gamma_c}}{2g}\frac{(\cF_{2\gamma_c})_{h+1}}{(\cF_{2\gamma_c})_h}+\cO(\frac{1}{g^2})$. Numerically, the term in $\frac{1}{g}$ at $g=70$ is of order $10^{-1}$. The Richardson extrapolation thus allows us to approach this limit with a precision far greater than this error.
\begin{table}[h!]
\centering
\begin{tabular}{c| c}
$h$ & $\widehat{s^{\scriptscriptstyle(2)}_{70,h}}$ \\
\hline
0  & $1.0000000000000068 + 1.36\times 10^{-14} \ii$ \\
1  & $0.9999999999999704 - 4.32\times 10^{-14} \ii$ \\
2  & $0.999999999999188 - 2.82\times 10^{-13} \ii$ \\
3  & $1.00000000000222 + 2.371\times 10^{-11} \ii$ \\
4  & $1.00000000370 + 1.24\times 10^{-9} \ii$ \\
5  & $1.0000000765 - 2.35\times 10^{-8} \ii$ \\
6  & $1.00000004 + 0\times 10^{-9} \ii$ \\
7  & $0.9999941 - 1.15\times 10^{-5} \ii$ \\
8  & $0.9999972 - 3.0\times 10^{-6} \ii$ \\
9  & $0.9999808 - 1.48\times 10^{-5} \ii$ \\
10 & $1.0000013 - 5.6\times 10^{-6} \ii$ \\
11 & $0.999980 - 2.4\times 10^{-5} \ii$ \\
12 & $0.999710 - 1.0\times 10^{-4} \ii$ \\
13 & $0.99841 - 3.5\times 10^{-4} \ii$ \\
14 & $0.99696 - 1.8\times 10^{-4} \ii$ \\
15 & $0.99783 - 5.2\times 10^{-4} \ii$ \\
\end{tabular}
\caption{Richardson extrapolation applied to the sequence $s^{(2)}_{g,h}$ for $0 \le h \le 15$.}\label{tab:sh2}
\end{table}

Next, we study the contribution from $3\omega_{\gamma_c}$ to the asymptotics. We subtract the contribution from $2\omega_{\gamma_c}$ to obtain $\Fgredn{2}$. Its leading asymptotics is governed by 
\begin{align}
    \cF_{3\omega_{\gamma_c}}&=g_s^3\Delta^+_{3\gamma_c}(\cF)=g_s^2\bar{\Omega}(3\gamma_c)\cF^{(3\gamma_c)}+g_s\bar{\Omega}(\gamma_c)\bar{\Omega}(2\gamma_c)\cF^{(\gamma_c \vert 2\gamma_c)}+\frac{1}{6}\bar{\Omega}(\gamma_c)^3\cF^{(\gamma_c \vert \gamma_c\vert \gamma_c)}\\
 &= \frac{g_s^2}{9}\cF^{(3\gamma_c)}+\frac{g_s}{4}\cF^{(\gamma_c \vert 2\gamma_c)}+\frac{1}{6}\cF^{(\gamma_c \vert \gamma_c\vert \gamma_c)}
\end{align}
via
\begin{equation}
    \Fgredn{2} \sim- \frac{1}{\pi \ii}\sum_{h\geq 0}\frac{(2g-h)!}{(3\omega_{\gamma_c})^{2g+1-h}}(\cF_{3\gamma_c})_h \,.
\end{equation} 
We test the coefficient $(\cF_{3\gamma_c})_h$ by computing
\begin{equation}
    s^{(3)}_{g,h}=\frac{\pi \ii(3\omega_{\gamma_c})^{2g-h+1}}{(2g-h)!(\cF_{3\gamma_c})_h}\left(  \Fgredn{2} +\frac{1}{\pi \ii}\sum_{k=0}^{h-1}\frac{(2g-k)!}{(3 \omega_{\gamma_c})^{2g-k+1}}(\cF_{3\gamma_c})_k\right).
\end{equation}
The results of Richardson extrapolation of this sequence are recorded in Table \ref{tab:sh3}. We clearly have less numerical precision than at two instanton level, yet for the first three values of $h$, the Richardson extrapolation of the sequence yields $\sim 1$ to second decimal order. The theoretical error is ${s^{(3)}_{g,h}=1+\cO(10^{-1})}$. For $h=3$, the errors from our numerical procedure (primarily the use of the Padé approximation) exceed this error.
\begin{table}[h!]
\centering
\begin{tabular}{c |c}
$h$ & $\widehat{s^{\scriptscriptstyle(3)}_{70,h}}$ \\
\hline
0  & $1.00087022 - 0.00110481\ii$ \\
1  & $1.00155 - 0.00785\ii$ \\
2  & $1.001 - 0.027  \ii$ \\
\end{tabular}
\caption{Richardson extrapolation applied to the sequence $s^{(3)}_{g,h}$ for $0 \le h \le 2$.}
\label{tab:sh3}
\end{table}

\subsubsection{$\quintic$ around conifold point: testing two-instanton corrections}
We next perform the analogous study for $\quintic$. As the $\cFpert_g$ are known to lower maximal genus, our numerical precision is only sufficient to study two-instanton corrections

We work at $z = (1-10^{-3})\mu$ and in the large radius frame. The leading and subleading singularity are at $\omega_{\gamma_c}=\alpha P_0$ and $2\omega_{\gamma_c}$ respectively. We obtain better results by studying the asymptotics of $\cFpert_g$, rather than studying the discontinuity of $\cFpert$. We subtract the leading asymptotics in exact analogy to the local $\IP^2$ analysis in the previous subsection, and compute the sequence $s^{(2)}_{g,h}$ defined by \eqref{eq:s2gh}. The Richardson extrapolations of the limit of this sequence for $h=0, \ldots, 6$ are recorded in Table \ref{tab:s2ghX5}. The expected theoretical error is again $\cO(10^{-1})$.

\begin{table}[h!]
\centering
\begin{tabular}{c| c}
$h$ & $\widehat{s^{\smash[t]{\scriptscriptstyle(2)}\vphantom{\rule{0pt}{0.5ex}}}_{64,h}}$ \\
\hline
0  & $1. - 1.7496\times 10^{-8}\,\ii$ \\
1  & $1. + 2.5431\times 10^{-8}\,\ii$ \\
2  & $1. - 1.61335\times 10^{-7}\,\ii$ \\
3  & $0.999998 - 9.92036\times 10^{-6}\,\ii$ \\
4  & $0.997989 + 4.33882\times 10^{-4}\,\ii$ \\
5  & $0.998818 + 9.97854\times 10^{-4}\,\ii$ \\
\end{tabular}
\caption{Richardson extrapolation applied to the sequence $s^{(2)}_{g,h}$ for $0 \le h \le 5$.}
\label{tab:s2ghX5}
\end{table}

\subsubsection{Local $\IP^2$ at large radius: resurgence of the one-instanton amplitude $\cF^{(\gamma_{1,3})}$} \label{s:resurgenceOneInstanton}
As we have argued in Section \ref{s:formalism}, the differential operator $\alienDiffOp_{\omega_{\gamma}}$ introduced in \eqref{eq:alienDiffOp} can be equally well applied to the topological string amplitude $\cFpert$ and to higher instanton amplitudes $\cF^{(\gamma_1| \ldots | \gamma_n)}$. In this subsection, we will study the Borel plane of such an instanton amplitude close to large radius. As in Section \ref{sec : BPSnum}, we choose the point $z=10^{-6}$ in moduli space, and we study the singularity at charge $\gamma = \gamma_{1,3} = [1,-2,-3)$ uncovered there, in the large radius frame. 

The Borel plane of $\cF^{(\gamma)}$ is depicted in Figure \ref{fig:BPF13}. It exhibits the same pattern of singularities as the Borel plane of $\cFpert$, reflecting the fact that the set of singularities $\Omega$ is closed under addition. An immediately visible departure from the Borel plane of $\cFpert$ is the absence of $\IZ_2$ symmetry under reflection through the origin. This symmetry in the case of $\cFpert$ is a consequence of its $g_s \rightarrow -g_s$ symmetry, absent for $\cF^{(\gamma)}$. Thus, singularities at $\omega$ and $-\omega$ will contribute differently to the asymptotics of $\cF^{(\gamma)}$. Numerically, for instance, 
$\cF^{(\gamma\vert \gamma_{2,0})} \ll \cF^{(\gamma\vert -\gamma_{2,0})}$; in the Borel plane depicted in Figure \ref{fig:BPF13}, this asymmetry is reflected in the more pronounced singularity at $-\gamma_{2,0}$.
\begin{figure}[h]
    \centering
    \includegraphics[scale=0.6]{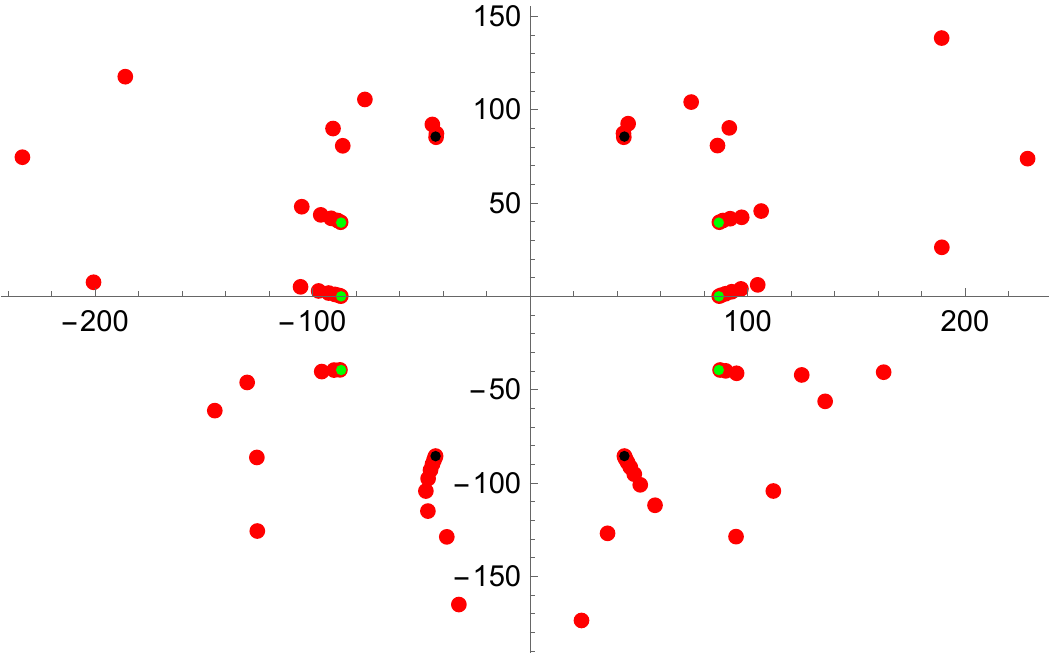}
    \caption{Borel plane of $\cF^{(\gamma_{1,3})}$ of local $\IP^2$ at $z=10^{-6}$ in the large radius frame. Black dots indicate the points $\pm \omega_{\gamma_{2,0}}$ and $\pm \omega_{\gamma_{3,0}}$. Green dots indicate the singularities associated to D2-D0 charges $\pm [0,1,m)$ for $m=-2,-1,0$ .}
    \label{fig:BPF13}
\end{figure}

We wish to numerically test the equation \eqref{eq:alienDFinst}, which predicts
\begin{equation} \label{eq:alienDpredict}
    \pd_{\omega_{\gamma'}}\cF^{(\gamma)}=S_{\omega_{\gamma'}}\frac{\e^{-\frac{\omega_{\gamma'}}{g_s}}}{g_s}\cF^{(\gamma\vert \gamma')} \,.
\end{equation}
We will study the singularity associated to the charge $\gamma'=-\gamma_{2,0}$.  As multiple singularities contribute comparably to the asymptotics of $\cF^{(\gamma)}$, we have recourse to \ref{method3}, introduced in Section \ref{sec:numtools}: the integral\footnote{The shift in the factorial and the power of $\zeta$ in the denominator compared to \eqref{eq:intGomegaPade} is due to the different indexing conventions of the $g_s$-coefficients of $\cFpert$ and $\cF^{(\gamma)}$.} 
\begin{align}
        I^P_{n,\omega_{\gamma'}}(\cF^{(\gamma)}) = \frac{n!}{2\pi\, \ii} \int_{\cH_{\omega_{\gamma'}}} \frac{\pade(\cF^{(\gamma)})}{\zeta^{n+1}}d\zeta = -  n!\sum_{\substack{\mathrm{poles}\, \zeta_i \\ \arg(\zeta_i) \simeq \arg(\omega_{\gamma})}} \res_{\zeta_i}\left(\frac{\pade(\cF^{(\gamma)})}{\zeta^{n+1}}\right) \,.
    \end{align}
can be approximated by
\begin{align}
    I_{n,\omega_{\gamma'}}^{I,h}(\cF^{(\gamma)}) = -\frac{S_{\omega_{\gamma'}}}{2\pi \ii} \sum_{k=0}^h \frac{(n-k)!}{\omega_{\gamma'}^{n+1-k}} \cF_k^{(\gamma|\gamma')}\,.
\end{align}

To study the successive contributions of the instanton coefficients to this approximation, we introduce the sequence
\begin{equation} \label{eq:sggpnh}
    s^{\gamma,\gamma'}_{n,h}=
    \frac{2 \pi \ii\, \omega_{\gamma'}^{n+1-h}}
    {S_{\omega_{\gamma'}} \cF_h^{(\gamma \vert \gamma')}(n-h)!}
    \left(
    I_{n,\omega'}^P (\cF^{(\gamma)})
    + \frac{S_{\omega_{\gamma'}}}{2\pi \ii}
    \sum_{k=0}^{h-1}
    \frac{(n-k)!}{\omega_{\gamma'}^{n+1-k}}
    \cF_k^{(\gamma \vert \gamma')}
    \right).
\end{equation}
We record the values of the Richardson extrapolation of the sequence $s_{n,h}$ for $h = 0, \ldots, 10$ in Table~\ref{tab:limsnh}. For the range of $h$ we consider, the expected theoretical error is $s^{\gamma,\gamma'}_{155,h}=1+\cO(10^{-2})$.
\begin{table}[h]
\centering
\begin{tabular}{c|c}
$h$ & $\widehat{s^{\gamma,\gamma'}_{155,h}}$ \\
\hline
0  & $0.9999532 + 1.5\times10^{-6}\, \ii$ \\
1  & $0.9999713 - 1.2\times10^{-6}\, \ii$ \\
2  & $0.9999365 - 8.0\times10^{-6}\, \ii$ \\
3  & $0.99998705 + 2.9\times10^{-7}\, \ii$ \\
4  & $0.9999922 + 2.6\times10^{-6}\, \ii$ \\
5  & $1.000028 + 5.0\times10^{-5}\, \ii$ \\
6  & $1.000224 - 1.07\times10^{-4}\, \ii$ \\
7  & $0.998947 - 6.10\times10^{-4}\, \ii$ \\
8  & $0.999005 - 8.7\times10^{-5}\, \ii$ \\
9  & $0.998767 + 2.4\times10^{-5}\, \ii$ \\
10 & $0.99852 - 4.8\times10^{-4}\, \ii$ \\
\end{tabular}
\caption{Richardson extrapolation applied to the sequence $s^{\gamma,\gamma'}_{155,h}$ for $0 \le h \le 10$.}
\label{tab:limsnh}
\end{table}

We can also test equation \eqref{eq:alienDpredict} by comparing the discontinuity it implies for $\cF^{(\gamma)}$,
\begin{equation} \label{eq:discPredict}
    \disc_{\theta_{\gamma'}}(\cF^{(\gamma)})
    =\sS^{\theta_{\gamma'}^-}(\pd_{\omega_{\gamma'}}\cF^{(\gamma)})
    +O(\e^{-\frac{2\omega_{\gamma'}}{g_s}}) \,,
\end{equation}
with the numerically computed discontinuity. We record the quotient $Q$ of the numerically determined discontinuity by $\sS^{\theta_{\gamma'}^-}(\pd_{\omega_{\gamma'}}\cF^{(\gamma)})$ in Table \ref{tab:discFgamma}, with the Laplace transform in $\sS$ performed numerically. We do not attempt to estimate the error arising from Padé approximation and numerical integration.
\begin{table}[H]
\centering
\begin{tabular}{c|c}
$g_s\ee^{-\ii \theta_{-\gamma_{2,0}}}$ &
Q\\
\hline
$3$ & $0.999826 - 3.85135\times10^{-5}\,\ii$ \\
$2$ & $0.999999 + 1.73965\times10^{-6}\,\ii$ \\
$\frac{3}{2}$ & $1 + 3.29437\times10^{-9}\,\ii$ \\
$1$ & $1 - 1.04575\times10^{-12}\,\ii$ \\
$\frac{1}{2}$ & $1 + 9.92491\times10^{-10}\,\ii$ \\
$\frac{2}{3}$ & $1 - 3.2401\times10^{-12}\,\ii$ \\
$\frac{1}{3}$ & $1.00001 - 4.72328\times10^{-7}\,\ii$ \\
$\frac{1}{5}$ & $1.00036 - 2.49717\times10^{-4}\,\ii$ \\
$\frac{1}{10}$ & $1.00407 - 5.79685\times10^{-3}\,\ii$ \\
\end{tabular}
\caption{The quotient $Q$ of the numerically determined Stokes discontinuity by prediction \eqref{eq:discPredict} for various values of $g_s$.}
\label{tab:discFgamma}
\end{table}
We performed analogous numerical checks for the pairs
\begin{align}
    (\gamma,\gamma') \in \,\left\{ (\gamma_{1,3},-\gamma_{3,0}),(-\gamma_{1,2},\gamma_{2,0}),(\gamma_{1,3},\gamma^1) \right\}\,,\,\,\,\gamma^1=[0,1,1)\,,
\end{align}

with comparable results.

\subsubsection{Local $\IP^2$ at large radius: resurgence of the two-instanton amplitude $\cF^{(\gamma_{1,3}\vert -\gamma_{0,2})}$}
We can also study the Borel plane of $\cF^{(\gamma_{1,3}\vert -\gamma_{0,2})}$ with the same numerical parameters as in the previous section. The  corresponding Borel plane is shown in Figure \ref{fig:BPF1302}. It shares the same features as the Borel plane of $\cF^{(\gamma_{1,3})}$, albeit with less pronounced singularities at $\pm \omega_{ \gamma_{2,0}}$ and $\pm \omega_{ \gamma_{3,0}}$.
\begin{figure}[H]
    \centering
    \includegraphics[scale=0.6]{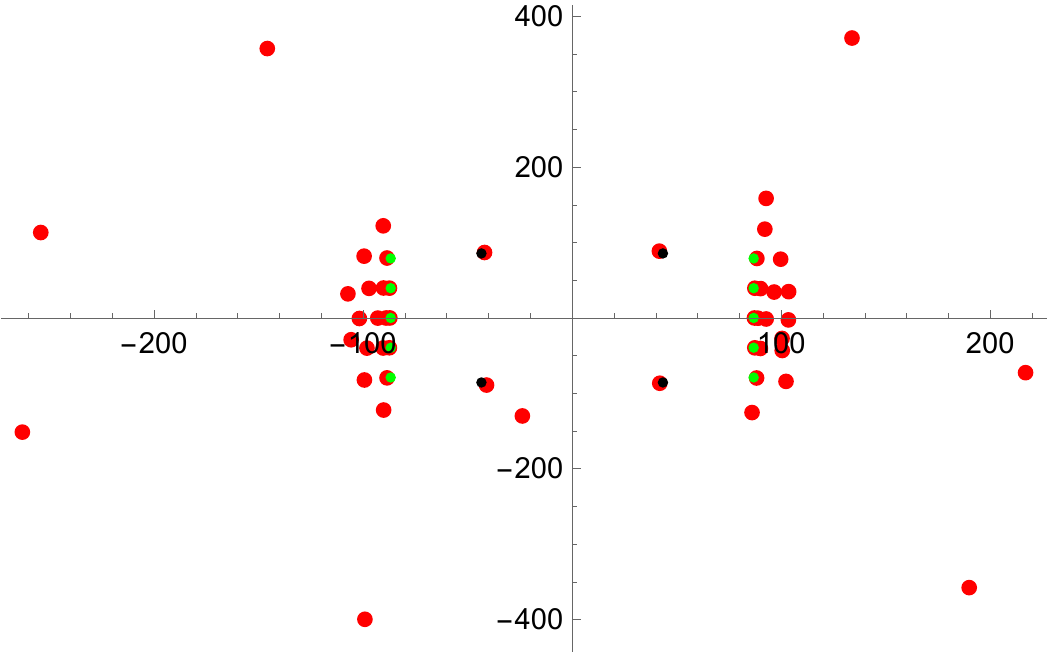}
    \caption{Borel plane of $\cF^{(\gamma_{1,3}\vert -\gamma_{0,2})}$ of local $\IP^2$ at $z=10^{-6}$ in the large radius frame. Black dots indicate the points $\pm \omega_{\gamma_{2,0}}$ and $\pm \omega_{\gamma_{3,0}}$. Green dots indicate the singularities associated to D2-D0 charges $\pm [0,1,m)$ for $m=-3,-2,-1,0,1$.}
    \label{fig:BPF1302}
\end{figure}
We numerically test the equation \eqref{eq:alienDFinst}, which predicts
\begin{equation} \label{eq:alienDpredict2}
    \pd_{\omega_{\gamma''}}\cF^{(\gamma\vert \gamma')}=S_{\omega_{\gamma''}}\frac{\e^{-\frac{\omega_{\gamma''}}{g_s}}}{g_s}\cF^{(\gamma\vert \gamma'\vert\gamma'')} \,,
\end{equation}
where $\gamma''=\gamma^1=[0,1,1)$, with $\omega_{\gamma^1}$ on the real axis, and $S_{\omega_{\gamma}''}=n^0_1=3$, by comparing the discontinuity it implies for $\cF^{(\gamma\vert \gamma')}$,
\begin{equation} \label{eq:discPredict2}
    \disc_{\theta_{\gamma''}}(\cF^{(\gamma\vert \gamma')})
    =\sS^{\theta_{\gamma''}^-}(\pd_{\omega_{\gamma''}}\cF^{(\gamma\vert \gamma')})
    +O(\e^{-\frac{2\omega_{\gamma''}}{g_s}}) \,,
\end{equation}
with the numerically determined discontinuity. We record the quotient $Q'$ of the numerically determined discontinuity by $\sS^{\theta_{\gamma''}^-}(\pd_{\omega_{\gamma''}}\cF^{(\gamma\vert \gamma ')})$ in Table \ref{tab:discFgammagammap}. 
\begin{table}[h]
\centering
\begin{tabular}{c|c}
$g_s\ee^{-\ii \theta_{\gamma^1}}$ &
$Q'$\\
\hline
$1$ & $1.0000033 - 4.68267\times10^{-7}\,\ii$ \\
$2$ & $0.99993 - 3.11781\times10^{-5}\,\ii$ \\
$3$ & $1.00044 + 4.26617\times10^{-4}\,\ii$ \\
\end{tabular}
\caption{The quotient $Q'$ of the numerically determined Stokes discontinuity by prediction \eqref{eq:discPredict2} for three values of $g_s$. At smaller values of $g_s$, the discrepancy between $Q'$ and 1 increases.}
\label{tab:discFgammagammap}
\end{table}
Surprisingly, when we consider $g_s$ closer to 0 than the values given in Table \ref{tab:discFgammagammap}, we appear to lose precision.

\subsubsection{$\quintic$ at large radius: resurgence of the one-instanton amplitude $\cF^{(\gamma_c)}$}
For $\quintic$, we depict the Borel plane of the instanton amplitude $\cF^{(\gamma_c)}$ near large radius in Figure \ref{fig:borelplaneX5Fgammac}. Singularities at $\alpha X^0=\omega_{\gamma^0}$ and $-\alpha P_0=-\omega_{\gamma_c}$ are clearly visible -- these are also leading singularities in the Borel plane of $\cFpert_g$. The absence of a visible singularity at $\omega_{\gamma_c}$ is consistent with  $\cF^{(\gamma_c,-\gamma_c)}\gg\cF^{(\gamma_c,\gamma_c)}$.
\begin{figure}[h]
    \centering
    \includegraphics[scale=0.6]{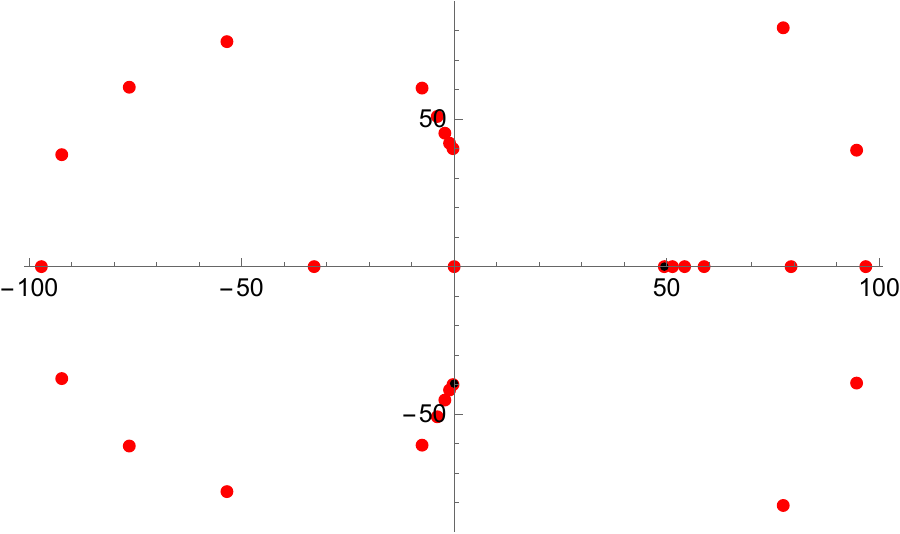}
    \caption{Borel plane of $\cF^{(\gamma_c)}$ for $\mirrorQ$ at $z=10^{-1}\mu$ in the large radius frame. Black dots depict the point $-\alpha P_0$ on the  real axis and $\alpha X^0$ on the imaginary axis.}
    \label{fig:borelplaneX5Fgammac}
\end{figure}
We can study the contribution of $\alpha X^0$ and $-\alpha P_0 $ to the asymptotics of $\cF^{(\gamma_c)}$. We again form the sequence $s^{\gamma,\gamma'}_{n,h}$ introduced in equation \eqref{eq:sggpnh}, with $\gamma = \gamma_c$ and $\gamma' \in \{\gamma^0, -\gamma_c\}$. Note that \cite{Gu:2023mgf,Douaud:2024khu}
\begin{align}
    S_{\omega_{\gamma^0}} =\bar{\Omega}(\gamma^0)= \chi(\mirrorQ) = 200 \,, \quad S_{\omega_{\gamma_c}} =\bar{\Omega}(\gamma_c)=1  \,.
\end{align}
We record the Richardson extrapolation for the sequences $s^{\gamma_c,\gamma'}_{n,h}$ in Table \ref{table:snhX5} for $h=0,\ldots,3$. Note that for $\gamma'=\gamma_c$,  $\cF^{(\gamma_c\vert-\gamma_c)}$ is even in $g_s$, hence $\cF^{(\gamma_c\vert-\gamma_c)}_{2k+1}$=0 and we only give the values of the approximant for $h$ even. For our range of $h$ we should have  $s^{\gamma,\gamma'}_{60,h}=1+\cO(10^{-1})$.
\begin{table}[h]
\centering
\begin{tabular}{c|c|c}
$h$ & $\widehat{s^{\gamma_c,\gamma^0}_{60,h}}$ & $\widehat{s^{\gamma_c,-\gamma_c}_{60,h}}$ \\
\hline
0 & $*$ & $0.99947 + 0\,\ii$ \\
1 & $1.09004 - 0.01831\,\ii$ & $*$ \\
2 & $1.02094 + 0.05396\,\ii$ & $1.02122 + 0\,\ii$ \\
3 & $1.00669 - 0.01580\,\ii$ & $*$ \\
\end{tabular}
\caption{Richardson extrapolation applied to the sequence $s^{\gamma_c,\gamma'}_{n,h}$ for $0 \le h \le 3$. $*$ corresponds to $\cF^{(\gamma_c\vert\gamma')}_h=0$.}
\label{table:snhX5}
\end{table}

\subsection{Wall-crossing in local $\IP^2$: D2-D0 decay away from large radius}
The refined subtraction method introduced in Section \ref{sec:numtools} enables us to study wall-crossing in the Borel plane of local $\IP^2$. Wall-crossing in the context of resurgence has previously been illustrated numerically in \cite{Marino:2024yme} via the example of W-boson and dyon decay from weak to strong coupling in pure Seiberg-Witten theory (Figure 5 in that reference shows a similar phenomenon in the case of local $\IF_0$), and in \cite{Douaud:2024khu} in the case of the decay of a D6-D0 bound state as one approaches the conifold point in all hypergeometric one-parameter Calabi-Yau threefolds (save $X_{3,2,2}$, for which the requisite numerical precision was lacking). In this section, we are going to trace the decay of a D2-D0 bound state in the Borel plane. The decay we focus on was analyzed in \cite[Section 4.4]{Bousseau:2022snm}: coming in from large radius, the state $\gamma_1=[0,d,\chi)=[0,1,1)$, associated to the class of the hyperplane divisor of $\IP^2$ decays. When $|\text{Arg}(- \ii Z_{\gamma_1})|$ is bounded by $\sim 0.824$, as is the case for the point $z_*$ in moduli space we shall consider, the decay takes the form \cite[Section 6.5]{Bousseau:2022snm}
\begin{align}
    \gamma_1 \rightarrow \gamma_2 + \gamma_3 \,,
\end{align}
with $\gamma_2=[1,0,1)=\gamma_c$ and $\gamma_3=[-1,1,0)$. As the charges involved are primitive and the decay only involves two constituents, the wall-crossing formula reads
\begin{equation}
     \Delta \Omega(\gamma_1)= \Omega(\gamma_1,z_*^+)- \Omega(\gamma_1,z_*^-)=\langle \gamma_2,\gamma_3\rangle(-1)^{\langle \gamma_2,\gamma_3\rangle}\Omega(\gamma_2,z_*)\Omega(\gamma_3,z_*) \,,
 \end{equation}
consistent with  
$\Omega(\gamma_2,z_*)=\Omega(\gamma_3,z_*)=1$,  $\Omega(\gamma_1,z_*^-)=3$ (the genus 0 GV invariant of degree 1), and, decisively, $ \Omega(\gamma_1,z_*^+)=0$, indicating the decay of the state of charge $\gamma_1$ beyond the wall.

In Figure \ref{fig:wall}, we trace out the associated wall of marginal stability in the moduli space: the dashed line indicates the locus on which the central charges $Z_{\gamma_2}(z_*)$ and $Z_{\gamma_3}(z_*)$ are aligned. 
  \begin{figure}[H]
    \centering 
    \includegraphics[scale=1]{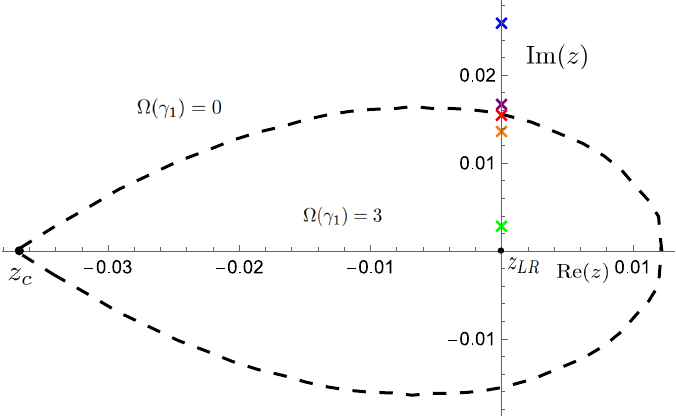}
    \caption{The wall of marginal stability in the $z$ plane for the decay $[0,1,1)\rightarrow [1,0,1)+[-1,1,0)$. Green/Orange/Red/Purple/Blue markers are  $z_1/z_2/z_*/z_3/z_4$.}
    \label{fig:wall}
\end{figure}
To trace the decay, we fix a point on the intersection of the wall with the imaginary axis,
\begin{equation}
    z_*\simeq 0.0154\ii \,,
\end{equation}
and compare the Borel plane at two points above and two below the wall: \begin{equation}
    z_1=\frac{2}{10}z_*,\,  z_2=\frac{9}{10}z_*,\, z_3=\frac{11}{10}z_*,\,   z_4=\frac{17}{10}z_* \,.
\end{equation}
These four points are indicated in Figure \ref{fig:wall}. Note that $z_1$ and $z_2$ are located inside the wall that encircles the large radius point $z_{\mathrm{LR}}$, $z_3$ and $z_4$ outside this wall. By the nature of the decay, $\omega_{\gamma_1}$ should be the position of a singularity in the Borel plane at the values $z_1$ and $z_2$ of the modulus, but a regular point (and hence invisible in the Borel plane) at $z_3$ and $z_4$.

For the numerical analysis of this decay to be feasible, it is important that the three charges involved in the decay are dominant in the vicinity of the point $z_*$, $\omega_{\gamma_2}$, $\omega_{\gamma_3}$, $\omega_{\gamma_1}$, see Figure \ref{fig:wczstar}.
\begin{figure}[H]
    \centering 
    \includegraphics[scale=0.4]{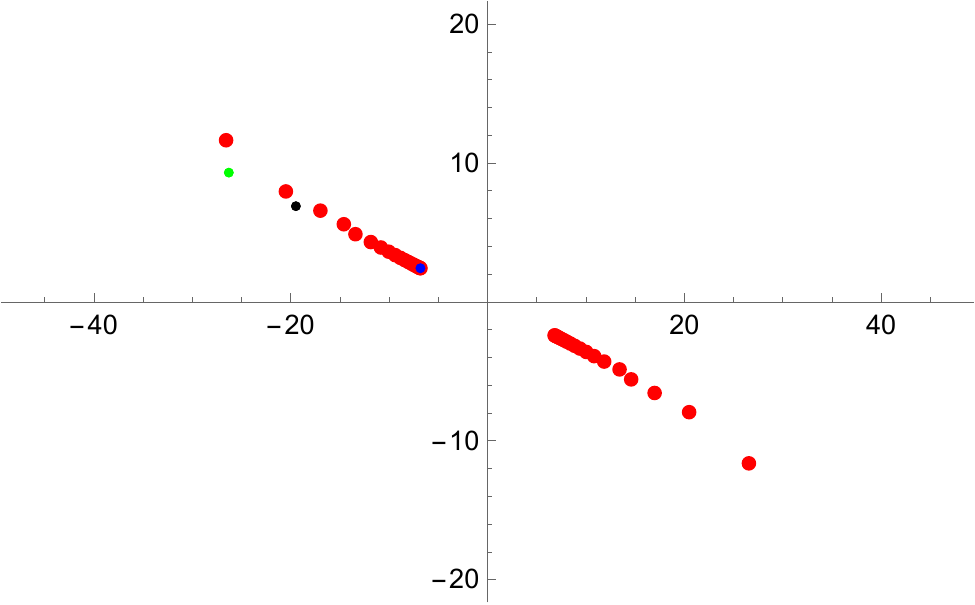}
    \caption{The Borel plane of $\cFpert$ at $z = z^*$. The three indicated points at increasing distance from the origin are $\omega_{\gamma_2}$ (blue), $\omega_{\gamma_3}$ (black), $\omega_{\gamma_1}$ (green).}
    \label{fig:wczstar}
\end{figure}

If this were not the case, the requirement of subtracting more dominant singularities would result in a loss of precision, possibly washing out the desired decay entirely.

As the periods satisfy $|\omega_{\gamma_1}| > |\omega_{\gamma_3}| > |\omega_{\gamma_2}|$ in the vicinity of $z_*$ that we are studying, we will work in a conifold frame, i.e. a frame with $\gamma_2$ as an A-cycle, such that the singularity associated to $\gamma_2$ can easily be subtracted. The resulting amplitudes $\Fgred$ are the starting point of our analysis summarized in Figure \ref{fig:sequenceWC}:
	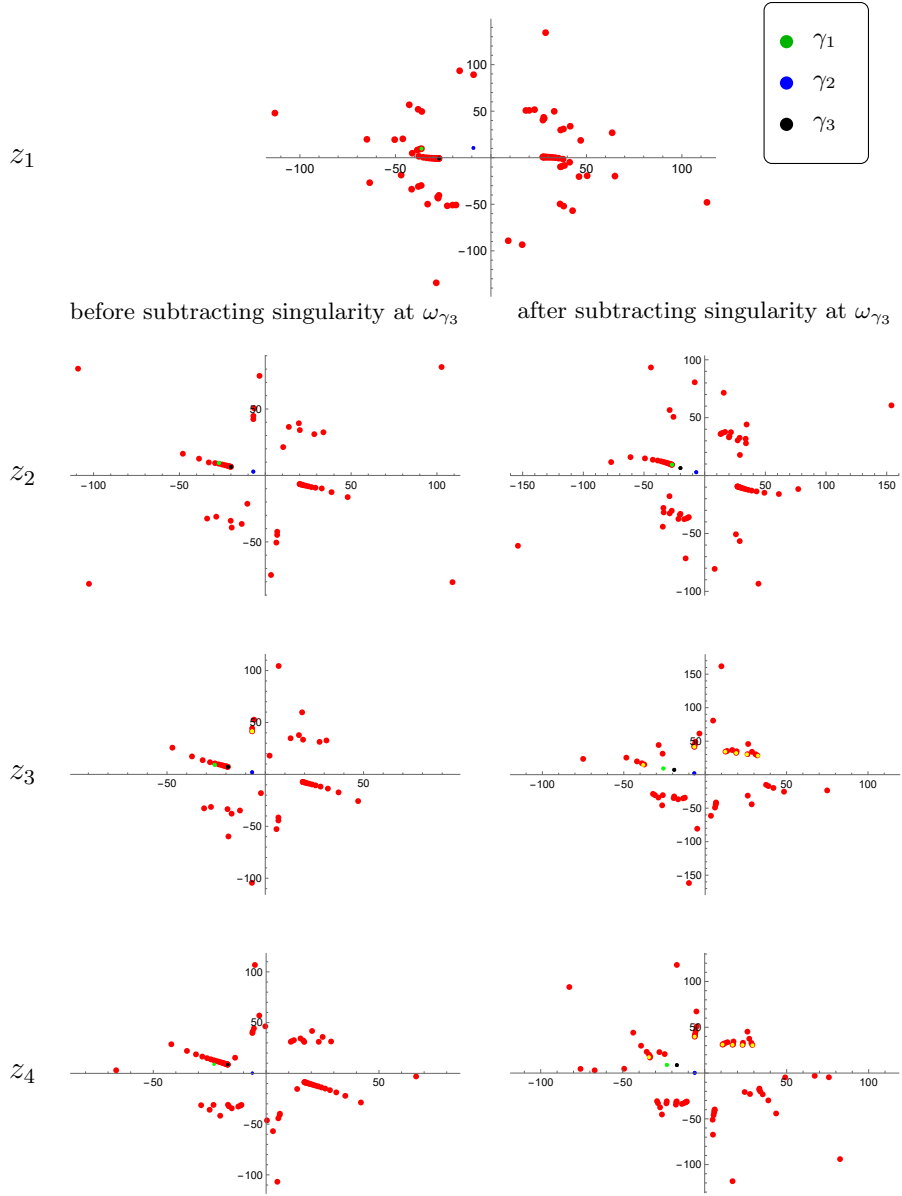
\begin{figure}[h]
		\centering
		\begin{tikzpicture}[
			figwidth/.store in=\figw, figwidth=5.2cm,
			node distance=0.3cm
			]
			
			\node (z1) {\includegraphics[width=6cm]{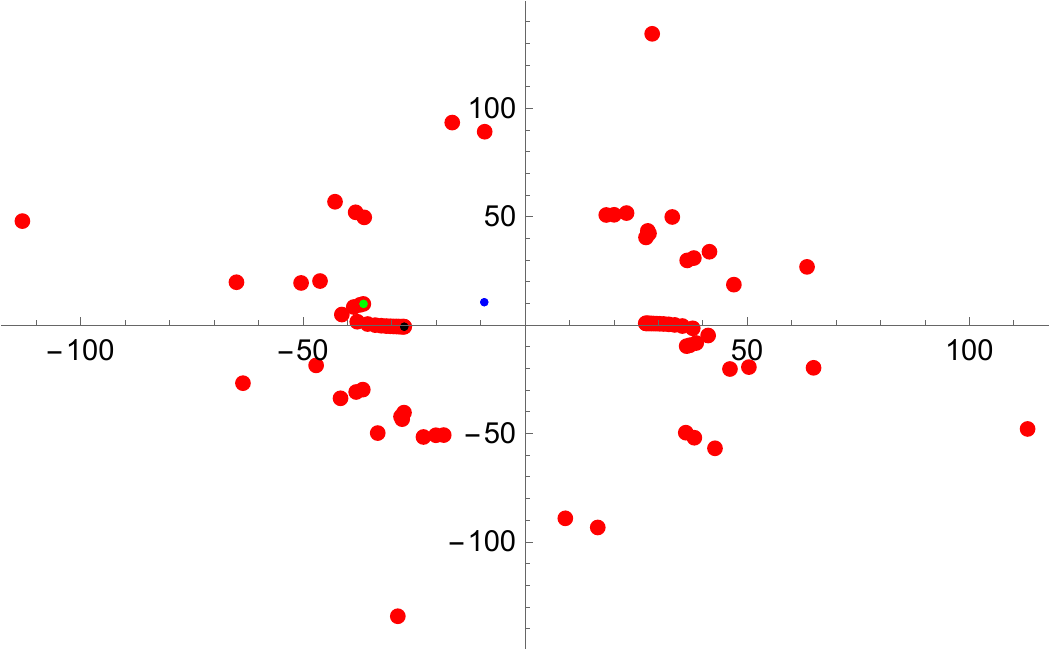}};
			
			\node (z2b) [below=0.5cm of z1, xshift=-3.0cm]
			{\includegraphics[width=\figw]{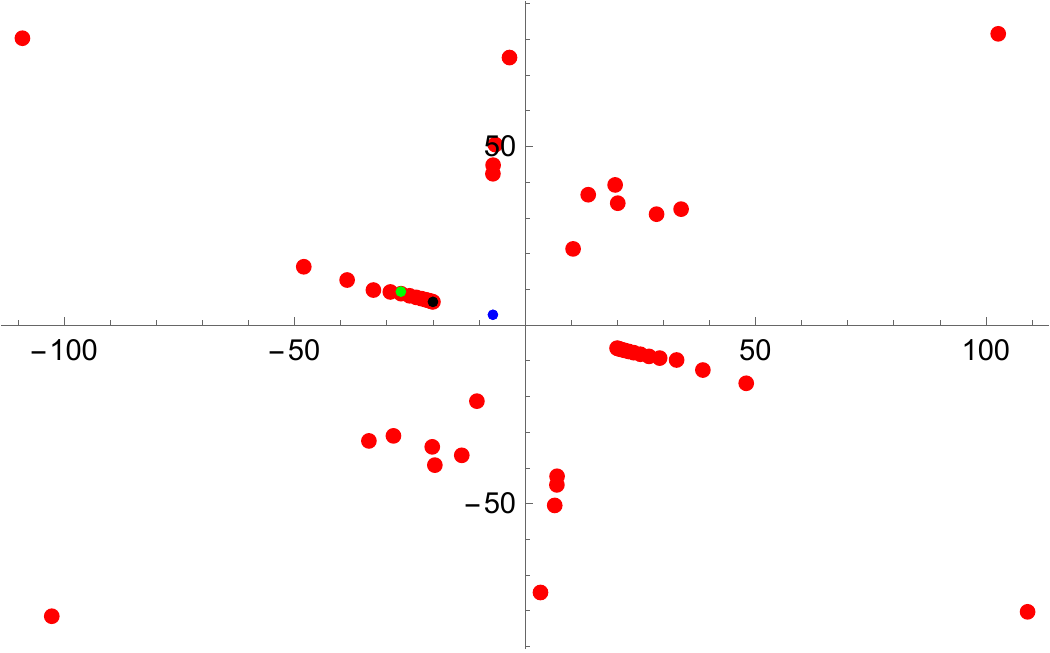}};
			\node (z2a) [right=0.4cm of z2b]
			{\includegraphics[width=\figw]{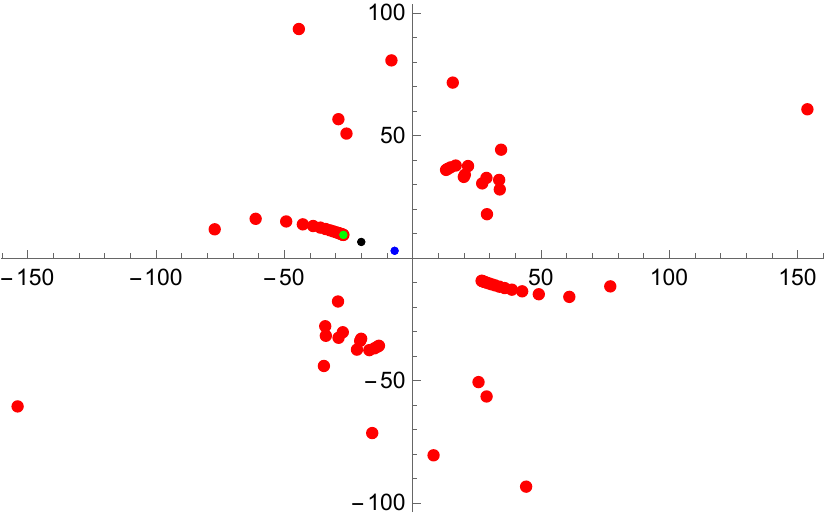}};
			
			\node (z3b) [below=0.5cm of z2b]
			{\includegraphics[width=\figw]{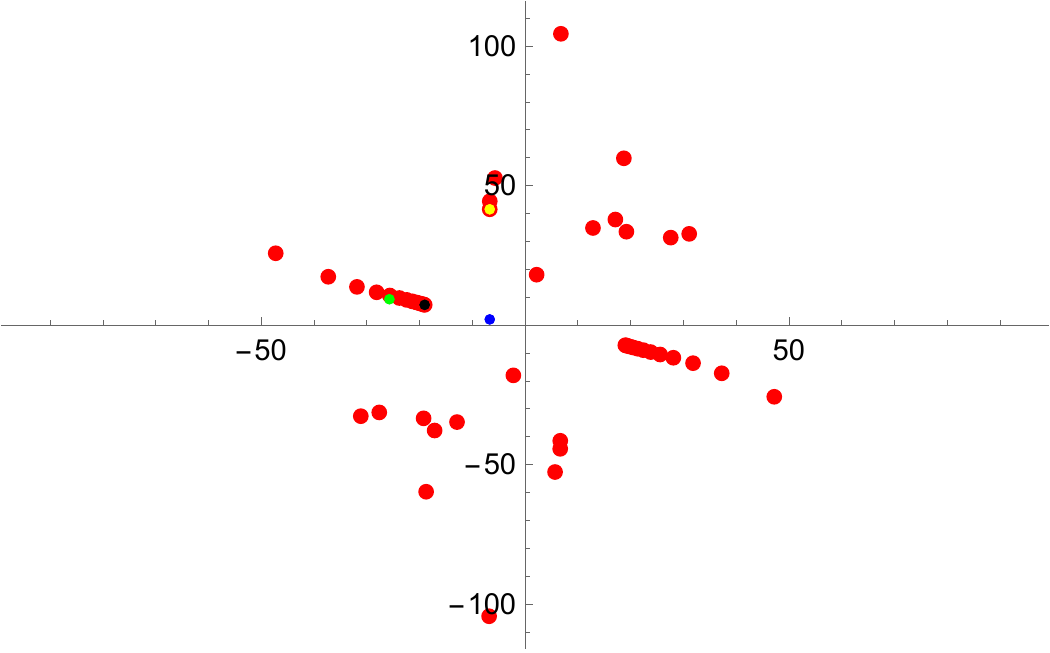}};
			\node (z3a) [right=0.4cm of z3b]
			{\includegraphics[width=\figw]{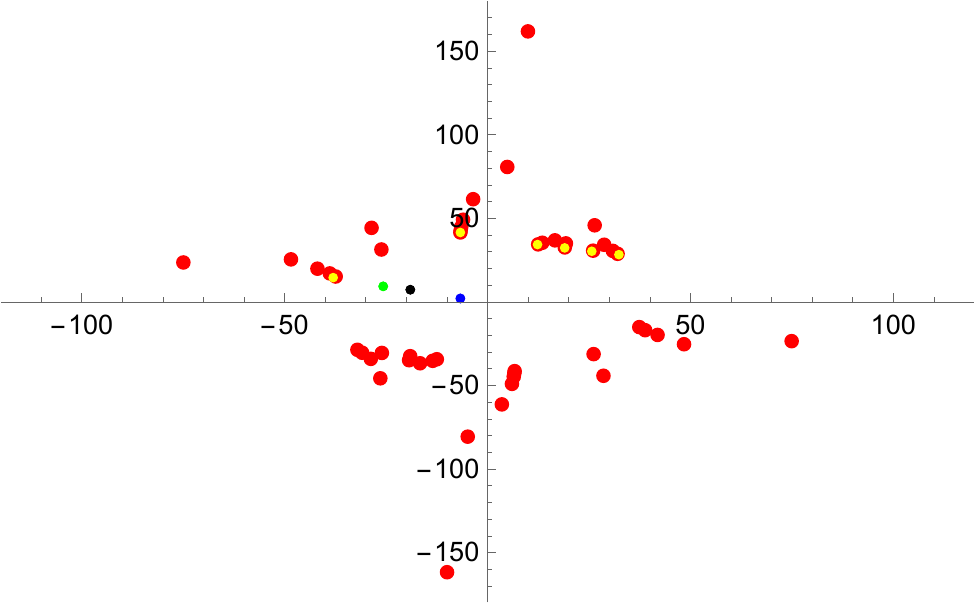}};
			
			\node (z4b) [below=0.5cm of z3b]
			{\includegraphics[width=\figw]{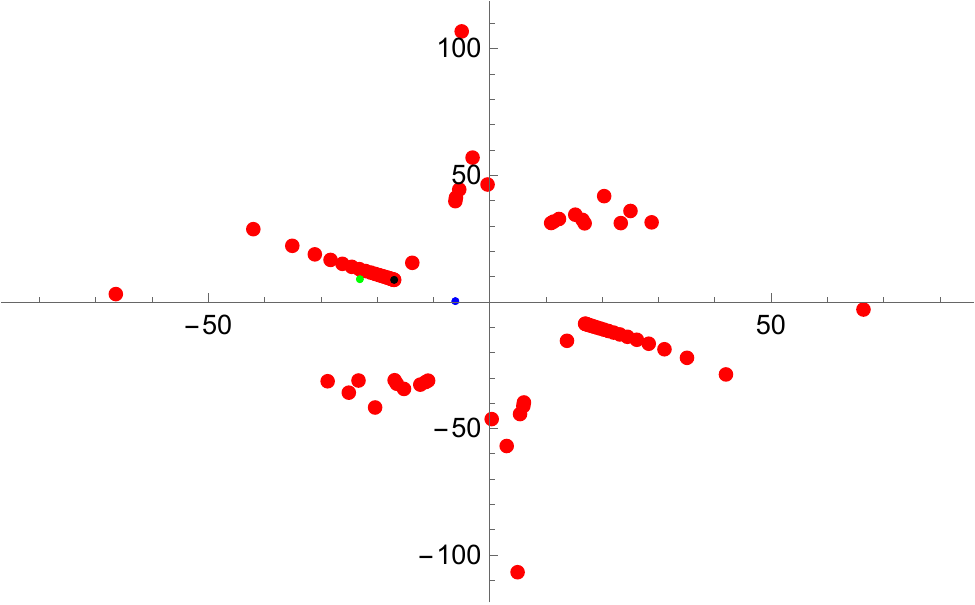}};
			\node (z4a) [right=0.4cm of z4b]
			{\includegraphics[width=\figw]{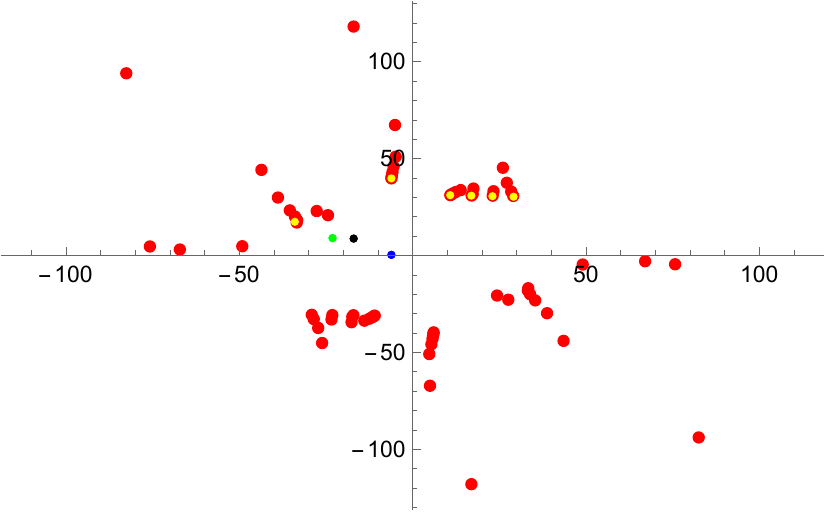}};
			
			\coordinate (labelcol) at ([xshift=-0.5cm] z2b.west);
			\node[font=\large] at (labelcol |- z1.center)  {$z_1$};
			\node[font=\large] at (labelcol |- z2b.center) {$z_2$};
			\node[font=\large] at (labelcol |- z3b.center) {$z_3$};
			\node[font=\large] at (labelcol |- z4b.center) {$z_4$};
			
		\node[above=0.15cm of z2b.north, font=\small]  {before subtracting singularity at $\omega_{\gamma_3}$};
		\node[above=0.15cm of z2a.north, font=\small]  {after subtracting singularity at $\omega_{\gamma_3}$};

			\node[
			draw, rounded corners=3pt,
			inner sep=6pt,
			font=\small,
			right=0.5cm of z1.east,
			anchor=west,
			yshift=1cm
			] (legend) {%
				\begin{tikzpicture}[x=0.9cm, y=0.55cm]
					\fill[green!70!black] (0,2) circle (2.5pt);
					\node[right, font=\small] at (0.15,2) {$\gamma_1$};
					\fill[blue]             (0,1) circle (2.5pt);
					\node[right, font=\small] at (0.15,1) {$\gamma_2$};
					\fill[black]            (0,0) circle (2.5pt);
					\node[right, font=\small] at (0.15,0) {$\gamma_3$};
				\end{tikzpicture}
			};
			
		\end{tikzpicture}
        \caption{The Borel plane at $z=z_i$, $i=1, \ldots, 4$. The singularity at $\omega_{\gamma_3}$ is removed on the right panels, to reveal a singularity at $\omega_{\gamma_1}$ still present at $z=z_2$, but absent at $z=z_3$ and $z=z_4$.The yellow points are the singularities identified in \eqref{eq:yellowsing}.} \label{fig:sequenceWC}
	\end{figure}
\begin{enumerate}
    \item At $z=z_1$, singularities in the Borel plane of $\Fgred$ are clearly visible at $\omega_{\gamma_1}$ and $\omega_{\gamma_3}$. We can confirm numerically that the Stokes constant at $\omega_{\gamma_1}$ is indeed 3.
    \item At $z=z_2$, the two periods $\omega_{\gamma_1}$ and $\omega_{\gamma_3}$ are almost aligned. To confirm that a singularity at $\omega_{\gamma_1}$ is present, we subtract the more dominant singularity at $\omega_{\gamma_3}$. 
    \item At $z=z_3$, the two periods remain nearly aligned. Upon subtracting the singularity at $\omega_{\gamma_3}$, we note that the singularity at $\omega_{\gamma_1}$ has disappeared. A host of singularities are visible in the Borel plane, associated to the charges 
    \begin{equation}
        [1, 0, 0),\quad [0, -1, -2),\quad [-1, -1, -3),\quad [2, -1, 0), \quad [1, -1, -1),\quad [-2, 2,0),
  \label{eq:yellowsing}
    \end{equation}
    Some of these were not visible before the subtraction. The dominant singularity among these is the one associated to the charge $[-2,2,0)$. By
    \begin{equation}
         27.3 \approx\vert \omega_{\gamma_1}(z_3)\vert \leq \vert  \omega_{[-2,2,0)}(z_3)\vert \approx 40.7 \,,
    \end{equation}
    the fact that a singularity at $\omega_{[2,1,0)}$ is visible gives us confidence that a singularity at $\omega_{\gamma_1}$, were it present, would also fall within the sensitivity of our analysis.
    \item At $z=z_4$, we observe $\omega_{\gamma_1}$ and $\omega_{\gamma_3}$ beginning to move out of alignment. The Borel plane upon subtracting the singularity at $\omega_{\gamma_3}$ resembles the plane at $z=z_3$, with still no singularity at $\omega_{\gamma_1}$ visible.
\end{enumerate}

\section{Conclusions} \label{s:conclusions}
We have argued that the resurgent structure of the topological string is captured by simple differential operators $\alienDiffOp_{\omegag}$, $\gamma \in \Gamma$. These operators provide a realization of the Kontsevich-Soibelman Lie algebra, suggesting an intimate relation between wall-crossing and the resurgent structure of the topological amplitudes.

Numerically, we continued the exploration of the Borel plane of the topological string amplitude $\cFpert$ for the quintic and local $\IP^2$. In the case of $\KPt$, we uncover singularities in the Borel plane associated to bound states involving D4-branes and match the associated Stokes constants and generalized Donaldson-Thomas invariants. The introduction of the operators $\alienDiffOp_\omegag$ allowed us to extend our exploration to the Borel plane of instanton amplitudes; here, we numerically confirmed our expressions for the action of successive alien derivatives on $\cFpert$. Finally, we identified the manifestation of a predicted decay in the Borel plane, and verified that it matches theoretical predictions.

Continuing this line of work, a natural next step is to identify further walls in the Borel plane of local $\mathbb{P}^2$ -- studying them numerically and comparing systematically with the scattering diagram perspective of \cite{Bousseau:2022snm}.
It is also desirable to push this comparison further to other local threefolds, such as local $\IP^1 \times \IP^1$, for which a scattering diagram analysis is available \cite{LeFloch:2024cwl}.

Interesting results regarding subleading singularities around the conifold point of local $\IP^2$ have been obtained in \cite{Gu_2022} based on the TS/ST formalism (see \cite{Marino:2015nla} for a review); these should be accessible from the Borel plane perspective as well.

The form of the operator $\twistedDG$ which enters into the definition of $\alienDiffOp_{\omegag}$ is reminiscent of the wave-function approach to the holomorphic anomaly equations \cite{Witten:1993ed,Aganagic:2006wq}, in which B-periods are mapped to derivative operators. Indeed, the references \cite{Marino:2024yme,Marino:2024tbx} provide a derivation of the holomorphic limit of instanton amplitudes on local Calabi-Yau threefolds in this framework. We will pursue this line of reasoning and extend it to the compact case in upcoming work \cite{Douaud2026}.

In the text, we put forth a tentative proof strategy for identifying Stokes constants with generalized Donaldson-Thomas invariants, though several steps remain technically formidable. Perhaps combining our results with ideas of Bridgeland \cite{Bridgeland:2016nqw, Bridgeland:2024ecw} can help overcome some of these difficulties.

\section*{Acknowledgements}
We would like to thank Jie Gu, Gengbei Guo, Marcos Mari\~no, Boris Pioline, David Sauzin and Maximilian Schwick for useful discussions, encouragement and collaboration on related projects. We would like to thank Marcos Mari\~no for sharing his unpublished notes regarding the numerical computation of instanton amplitudes with us, and David Sauzin for sharing his unpublished notes regarding parametric resurgence.   
\vspace{0.3cm}

\noindent
The work of the authors was supported under ANR grant ANR-21-CE31-0021.

\appendix
\section{Resurgent tools} \label{app:resurgentTools}
\subsection{Notation}
\begin{tabular}{ll}
$D(a,R)$ & the open disk with center $a\in \IC$ and radius $R>0$ \\
$D^*(a,R)$ & $D(a,R) - \{0\}$ \\
$C_r(a)$ & a circular path of radius $r>0$ traversed positively around the point $a\in\IC$, \\
& around the origin when the argument $a$ is omitted  \\
$\overline{C_r(a)}$ & the image of $C_r(a)$ \\
$\theta_\omega$ & the argument of $\omega \in \IC$ \\
$\ray_\omega$ & the ray emanating from the origin through $\omega$ \\
$[\omega_1,\omega_2]$ & the oriented straight path from $\omega_1 \in \IC$ to $\omega_2 \in \IC$ \\
$\tilde{\varphi}$     & a formal power series without constant term \\
$\hat{\varphi}$ & a holomorphic germ, its analytic continuation in a open set of the Borel plane,\\
&can also denote the minor of $\singphi$, $\hat{\varphi}$ = $\var(\singphi)$ \\
$\singphi$   & an element of SING \\
$\check{\varphi}$ & a major of $\singphi$, two representatives of the class differ by a holomorphic function near the origin, \\
 & lives in ANA, that is on $\tC$
\end{tabular}
\vspace{0.2cm}

The space $\SINGsimp$ is isomorphic to $\IC \delta \oplus \IC \{\zeta \}$. When we write an element of $\SINGsimp$ as $a \delta + \hat{\varphi}$, this isomorphism is implied. We proceed analogously for the space $\SINGsimpRam$.

\subsection{Basic definitions}
The starting point of our considerations is a formal power series in a parameter $z$, $\tildephi(z)$, without constant term:
\begin{align}
    \tildephi(z) = \sum_{n=0}^\infty c_n z^{n+1} \in z\,\IC[[z]]\,.
\end{align}
\begin{definition}
    The formal Borel transform is defined as the following linear map between formal power series:
    \begin{eqnarray} \label{eq:Boreltransform}
        \sB :\quad z \,\IC[[z]]\quad &\longrightarrow& \quad \IC [[\zeta]] \\
         \sum_{n=0}^\infty c_n z^{n+1} &\longmapsto& \sum_{n=0}^\infty c_n \frac{\zeta^n}{n!} \nonumber \,.
    \end{eqnarray}
\end{definition}
If and only if $\tildephi(z)$ is a Gevrey-1 series -- we write $\tildephi(z) \in z\, \IC[[z]]_1$ -- the radius of convergence of $\hatphi = \sB(\tildephi)$ is non-zero. If $\hatphi(\zeta)$ exists and its analytic continuation along a ray drops off sufficiently quickly, we can compute its Laplace transform along this ray. The condition of dropping off sufficiently quickly is captured by the following definition:
\begin{definition}
    Let $\cN_{c_0}(\e^{\ii \theta}\IR^+)$ be the set of all convergent power series $\hatphi(\zeta) \in \IC\{\zeta\}$ which extend analytically to a half-strip of width $\delta$, $\delta >0$, around the ray $\e^{\ii \theta} \IR^+$ and are bounded by
    \begin{align}\label{eq:boundIntegrandLaplace}
        |\hatphi(\zeta)| \le A \e^{c_0 |\zeta|} \quad \mathrm{for} \quad \zeta \in  \e^{\ii \theta}\IR^+ 
    \end{align}
    for some appropriate choice of $A>0$ and $c_0 >0$. Set
    \begin{align}
        \cN(\e^{\ii \theta}\IR^+) = \bigcup_{c_0 \in \IR} \cN_{c_0}(\e^{\ii\theta}\IR^+) \,.
    \end{align}
\end{definition}
We then have the following theorem (see e.g. \cite[Theorem 5.20]{Sauzin}):
\begin{theorem} \label{thm:LaplaceToAsymptotic}
    Let $\theta \in [0,2\pi)$, and let $\hatphi \in \cN_{c_0}(e^{\ii \theta}\IR^+)$. Then the Laplace transform $(\sL^{\theta} \hatphi)(z)$ of $\hatphi$, defined by the integral
    \begin{equation} \label{eq:laplaceTransform}
        (\sL^{\theta} \widehat{\varphi})(z) = \int_0^\infty \e^{-\frac{\xi \e^{\ii \theta}}{z}} \widehat{\varphi}(\xi \e^{\ii \theta}) \e^{\ii \theta} \dd \xi 
    \end{equation}
    is holomorphic for $\re(z) > c_0$ and has the formal series $\widetilde{\varphi} := \sB^{-1} \widehat{\varphi} \in z\,\IC[[z]]_1$ as its uniform asymptotic expansion for any $z : \re(z) >c_1$, where $c_1$ is any positive constant such that $c_1 > c_0$.
\end{theorem}
We can therefore define an operator
\begin{equation} \label{eq:BorelResum}
    \sS^\theta = \sL^{\theta} \circ \sB 
\end{equation}
which under favorable circumstances and appropriate choice of $\theta$ maps a Gevrey-1 formal power series $\tildephi$ to a function $\varphi$ of which it is the asymptotic expansion. When this succeeds, $\varphi$ is called a Borel resummation of $\widetilde{\varphi}$, and $\varphi$ is called fine-summable in the direction $\theta$.

Varying $\theta$ and studying how the image of $\sS^\theta$ changes will play a central role in our discussion. For an open interval $I$, we will let $\cN(I)$ denote the generalization of $\cN(\e^{\ii \theta} \IR^+)$ to holomorphic functions defined in a neighborhood of the origin which analytically extend to the open sector 
\begin{align}
    \{\xi \e^{\ii \theta} \,| \,\xi>0,\, \theta\in I\} \,.
\end{align}
The precise generalization of the bound \eqref{eq:boundIntegrandLaplace} can be found in \cite[Definition 5.29]{Sauzin}. The image of $\sS^\theta$ as $\theta$ varies depends sensitively on the singularities that $\hatphi$ develops as it is analytically continued away from the origin of the Borel plane.

We can define compatible products on the space of Gevrey-1 series and on its image under Borel transform. Let $\tildephi(z), \tildepsi(z) \in z \, \IC[[z]]$, 
\begin{align}
    \tildephi(z)=\sum_{n=0}^\infty a_n z^{n+1} \,, \quad \tildepsi(z)=\sum_{n=0}^\infty b_n z^{n+1} \,,
\end{align}
and set $\hatphi=\sB(\tildephi)$, $\hatpsi = \sB(\tildepsi)$, $\hatphi(\zeta), \hatpsi(\zeta) \in \IC[[\zeta]]$. The Cauchy product on $z\, \IC[[z]]$, 
\begin{align}
    \tildephi \cdot \tildepsi = \sum_{n=0}^\infty c_n z^{n+2}\,, \quad c_n = \sum_{p+q=n} a_p b_q 
\end{align}
pushes forward to the convolution product on $\IC[[\zeta]]$,
\begin{align}
    (\hatphi * \hatpsi)(\zeta) = \sum_{n=0}^\infty c_n \frac{\zeta^{n+1}}{(n+1)!} \,, \quad c_n = \sum_{p+q=n} a_p b_q \,.
\end{align}
When $\hatphi, \hatpsi$ have finite radius of convergence around the origin, i.e. $\hatphi, \hatpsi \in \IC\{\zeta\}$, their convolution product can be defined analytically \cite[Lemma 5.14]{Sauzin}: Letting $\Phi$, $\Psi$ denote the holomorphic functions defined by $\hatphi, \hatpsi$ in the disc $D(0,R) = \{\zeta \in \IC \, | \, |\zeta| < R\}$ with $R>0$ smaller than the radius of convergence of both, we can define, for for $\zeta \in D(0,R)$, the holomorphic function 
\begin{align}
    (\Phi * \Psi)(\zeta) &= \int_0^\zeta \Phi(\xi) \Psi(\zeta - \xi) \, d\xi \,.
\end{align}
$(\Phi * \Psi)(\zeta)$ is the sum of the power series $\hatphi * \hatpsi$, whose radius of convergence is thus at least $R$.

\subsection{Some formal definitions}
We will be dealing with holomorphic functions whose analytic continuations exhibit infinitely many branch points of logarithms. As we will only be interested in the behavior of these functions close to the singular points, we do not need to riddle the complex plane with branch cuts. Instead, we proceed as follows:
\begin{enumerate}
    \item The Borel transform $\hatphi(\zeta)$ of a 1-Gevrey series is a holomorphic function with non-zero radius of convergence. To eliminate the need to specify this radius, we introduce the space of {\it holomorphic germs} by identifying all functions holomorphic on some open set containing the origin which coincide on some such open set. This space can be identified with the space $\IC\{\zeta\}$ of convergent power series at the origin. Hence, we write $\hatphi(\zeta) \in \IC\{\zeta \}$. The space of holomorphic germs at any point of the complex plane is defined analogously.
    \item The holomorphic germs we will be interested in can be continued along any path $\gamma:[0,1] \rightarrow \IC- \Omega$ in the complex plane which avoids a discrete set of points $\Omega$. We call the holomorphic germ at the endpoint $\gamma(1)$ of the path $\gamma$
    \begin{align}
        \cont_{\gamma} \hatphi \,.
    \end{align}
    The germ $\cont_\gamma(\hatphi)$ will generically depend on the homotopy class of $\gamma$ in $\IC - \Omega$.
    \item We will need to consider holomorphic functions which develop poles or single logarithms. To extend the notion of germs to logarithms, we need to address the fact that the logarithm cannot be defined on any punctured disk centered at its branch point. One solution is to work with the universal cover $\tC$ of $\IC^*$ with base point 1. This is the space of equivalence classes of paths $\gamma : [0,1] \rightarrow \IC^*$ with base point $\gamma(0) = 1$ under homotopy with fixed endpoints. Concretely, two points $[\gamma_1], [\gamma_2] \in \tC$ coincide if $\gamma_1(1) = \gamma_2(1)$ and both paths wrap the same number of times around the origin. We define $\underline{\e}^{i\theta}$ for $\theta \in \IR$ as the homotopy class of the path which wraps around the origin $(\theta \mod 2\pi \ii) $ times and ends at $\e^{\ii\theta}$. Any point $\zeta$ of $\tC$ can be written as $\zeta = r \underline{\e}^{\ii \theta}$ for $r\in (0,\infty)$, $\theta \in \IR$. There is a natural projection $\pi : \tC \rightarrow \IC$ given by $[\gamma]  \mapsto \gamma(1)$. $\tC$ is called the Riemann surface of the logarithm, as the logarithm defined via
    \begin{align}
        \log: \tC &\longrightarrow \IC \\
        [\gamma]   &\longmapsto \int_\gamma \frac{d\zeta}{\zeta}
    \end{align}
    is single-valued on it. To replicate the definition of holomorphic germs, we need to generalize the notion of a punctured disk of arbitrary radius to $\tC$. The relevant definition can be found in \cite[Definition 6.37]{Sauzin}. The corresponding space of {\it singular germs} is called $\ana$.
    \item We will need to lift holomorphic functions defined on a disk $D \subset \IC^*$ to which the origin is adherent to an element of the space $\ana$. This requires choosing a connected component among the infinite components $\tD$ of $\pi^{-1}(D)$. We specify the component by indicated $\xi \in \tD$ which maps to the center of the disk $D$, and we write $\tD = \tD_D(\xi)$. $\pi|_{\tD}$ defines a biholomorphism between $\tD$ and $D$, and hence allows us to lift a holomorphic function $f:D \rightarrow \IC$ to a function $\fana : \tD \rightarrow \IC$ via $\fana = f \circ \pi$. If $f$ has the spiral continuation property, then $\fana$ can be extended to an element of $\ana$ by analytic continuation.
    \item It will be useful to mod out the space $\ana$ by contributions which do not contribute to Stokes automorphisms. These are precisely the lift of contributions which are holomorphic at the origin. Using the canonical embedding
    \begin{align} \label{eq:convPowerSING}
        \IC\{\zeta\} \hookrightarrow \ana \,,
    \end{align}
    we define the space of singularities and the associated projection
    \begin{align}
        \SING = \ana/\holgerms \,, \quad \sing_0 : \ana \rightarrow \SING \,.
    \end{align}
    
    \end{enumerate}

\subsection{Simply ramified singularities} \label{app:simplyRamSing}
The Borel transform of a 1-Gevrey formal series is typically not an entire function. The singularities that its analytic continuation develops -- we speak of singularities in the Borel plane -- render the non-perturbative completion of the formal series ambiguous and play a central role in the analysis to follow. The singularities that we have found to be relevant for the non-perturbative analysis of the topological string partition function are poles and logarithms. We hence define, following \cite[eq. (6.28)]{Sauzin},
\begin{align}
    \delta = \sing_0 \left( \frac{1}{2 \pi \ii \,\zeta}\right) \,, \quad \delta^{(k)} = \sing_0 \left( \frac{(-1)^k k!}{2 \pi \ii \,\zeta^{k+1}} \right) \,,\,\,k\ge 0\,.
\end{align}
The normalizations are chosen such that
\begin{align}
    \delta^{(k)} = \left(\frac{d}{d\zeta}\right)^k \delta \,.
\end{align}
We then define a simply ramified singularity of order $n \in \IN$ to be any singularity of the form \cite[p50]{sauzin2007}
\begin{align}
    \singphi = \sum_{k=1}^n a_{-k}\delta^{(k-1)}+\sing_0 \left(\hatphi(\zeta) \frac{\log \zeta}{2 \pi \ii} \right) 
\end{align}
with $a_i \in \IC$ for $i=1, \ldots, n$ and $\hatphi(\zeta) \in \holgerms$.
We denote the subspace of all simply ramified singularities of arbitrary order $n$ as $\SINGsimpRam$. Note that we have a $\IC$-linear isomorphism
\begin{align} \label{eq:isoSingSimpInf}
    \bigoplus_{n=0}^\infty \IC \delta^{(n)} \oplus \IC\{\zeta \} \quad&\longrightarrow \quad \SINGsimpRam \\
    \sum_{k=1}^n a_{-k}\delta^{(k-1)}+ \hatphi(\zeta) \quad &\longmapsto \quad \sum_{k=1}^n a_{-k}\delta^{(k-1)}+\sing_0 \left(\hatphi(\zeta) \frac{\log \zeta}{2 \pi \ii} \right) \nonumber
\end{align}
with inverse map given by the variation map \cite[Definition 6.46]{Sauzin}. The subspace of $\SINGsimpRam$ given by the image of $\IC \delta^{(0)} \oplus \IC\{\zeta \}$ is called the space of simple singularities $\SING^{\simp}$. It will be convenient to introduce the notation
\begin{align} \label{eq:notationCoeffSingPoles}
    (\singphi)_k := a_k \,, 
\end{align}
as well as the projection
\begin{align} \label{eq:projection}
    \mu : \bigoplus_{n=0}^\infty \IC \delta^{(n)} \oplus \IC\{\zeta \} \rightarrow \IC\{\zeta \} \,,
\end{align}
such that
\begin{align} \label{eq:logCoeff}
    \hatphi(\zeta) = \mu(\singphi) \,,
\end{align}
with the isomorphism \eqref{eq:isoSingSimpInf} on the right-hand side of \eqref{eq:logCoeff} left implicit.

The space $\SINGsimpRam$, unlike the space $\SING^{\simp}$, is closed under the operator $\tfrac{d}{d\zeta}$, which on the right hand side of the isomorphism \eqref{eq:isoSingSimpInf} acts as differentiation, and on the left hand side acts as
\begin{align} \label{eq:diffIsoSing}
    \frac{d}{d\zeta}\,: \quad \bigoplus_{n=0}^\infty \IC \delta^{(n)} \oplus \IC\{\zeta \} \quad&\longrightarrow \quad \bigoplus_{n=0}^\infty\IC \delta^{(n)} \oplus \IC\{\zeta \}\\
    \sum_{k=1}^n a_{-k}\delta^{(k-1)}+ \hatphi(\zeta) \quad &\longmapsto \quad \sum_{k=1}^n a_{-k}\delta^{(k)}+ \hatphi(0)\delta^{(0)} + \hatphi'(\zeta) \nonumber
\end{align}
The space $\SING$ can be equipped with a commutative and associative product, the convolution \mbox{product}~$*$, with unit $\delta$ \cite[Proposition 8]{sauzin2007}. The space $\SINGsimpRam$ is closed under this product, such that $\SINGsimpRam \subset \SING$ is an algebra inclusion. The rules for computing the convolution product on $\SINGsimpRam$ follow from
\begin{enumerate}
    \item 
    \begin{align}
        \sing_0\left(\hatphi(\zeta) \frac{\log \zeta}{2\pi \ii} \right) * \sing_0\left(\hatpsi(\zeta) \frac{\log \zeta}{2\pi \ii} \right) = \sing_0\left(\hatphi(\zeta)*\hatpsi(\zeta) \frac{\log \zeta}{2\pi \ii} \right)
    \end{align}
    \item 
    \begin{align}
        \left(\frac{d}{d \zeta}\right)^n \left( \singphi * \singpsi \right) =  \left( \left(\frac{d}{d \zeta}\right)^n \singphi \right) * \singpsi  = \singphi * \left(\frac{d}{d \zeta}\right)^n  \left(\singpsi \right)
    \end{align}
\end{enumerate}
E.g.,
\begin{align}
    \delta^{(m)} * \delta^{(n)} = \left(\frac{d}{d \zeta}\right)^{m+n} \delta = \delta^{(m+n)} \,.
\end{align}

\subsection{Moving away from the origin}
The theory of resurgence requires analytically continuing the holomorphic germ $\hatphi \in \IC\{\zeta\}$ away from the origin. 

\begin{definition} \cite[Definition 6.1]{Sauzin}
    Let $\Omega$ be a non-empty closed discrete subset of $\IC$. We say that $\hatphi \in \IC\{\zeta\}$ is an $\Omega$-continuable germ if it can be analytically continued along any path in $\IC - \Omega$ from within a pointed disc $D^*(0,R)$ of radius smaller than the radius of convergence of $\hatphi$ and satisfying $D^*(0,R) \cap \Omega = \emptyset$. We denote the set of all $\Omega$-continuable germs by $\sRhat_{\Omega}$.
\end{definition} 
Using the canonical embedding \eqref{eq:convPowerSING}, we can interpret $\sRhat_{\Omega}$ as a subspace of $\SING$. We then introduce the direct sum
\begin{align} \label{eq:omResF}
    \IC \delta^{(n-1)} \oplus \cdots \oplus \IC \delta \oplus \sRhat_{\Omega} \,,
\end{align}
whose elements we call $\Omega$-resurgent function of order $n$. We call the set of all such functions for arbitrary $n$ the set of $\Omega$-resurgent functions.

It will sometimes prove convenient to embed $\sRhat_{\Omega} \hookrightarrow   \IC\{\zeta\}$ without invoking \eqref{eq:convPowerSING}. To make sense of \eqref{eq:omResF} in such instances, we need an alternative home for $\delta^{(n)}$. We return to this in Section \ref{app:backToz}.

\subsection{Alien operators in general}
Singularities of the analytic continuation of holomorphic germs at the origin, arising as the image of Borel transforms, are at the origin of the Stokes automorphism phenomenon. An alien operator extracts the singularity of the analytic continuation of a holomorphic germ at the origin to a point $\omega \in \IC$. In full generality, 
\begin{definition} \label{def:Agen}
    We define
\begin{equation}
    \begin{aligned}
        \sA_\omega^{\gamma, \xi} : \IC \delta^{(n-1)} \oplus \cdots \oplus \IC \delta \oplus \sRhat_{\Omega} & \rightarrow \mathrm{SING} \\
        \sum_{k=1}^n a_{-k} \delta^{(k-1)} + \hatphi & \mapsto \mathrm{sing}_0 (\fana) \,,
    \end{aligned}
\end{equation}
where
\begin{equation} \label{eq:deffanaforA}
    \fana(\zeta) = (\mathrm{cont}_\gamma \hatphi) (\omega +\pi(\zeta)) \quad \mathrm{for} \quad \zeta \in \tD \subset \tC \,.
\end{equation}
Here, $\gamma$ is a path from a point with the convergence radius of $\hatphi$ around the origin to a point $\zeta_1 \in \IC - \Omega$ at which a disk $D \subset \IC-\Omega$ is centered to which $\omega$ is adherent. $-\omega + D$ hence defines a disk with center $-\omega + \zeta_1$ to which the origin is adherent. We choose $\xi \in \tC$ in the preimage of $-\omega + \zeta_1$ under the projection $\pi$, and set $\tD = \tD_{-\omega+D}(\xi)$. By $\Omega$-continuability of $\hatphi$, $\fana$ as defined in \eqref{eq:deffanaforA} can be analytically continued to yield an element of $\ana$. 
\end{definition}
$\sA^{\gamma,\xi}_\omega$ is an example of an alien operator. A general alien operator is defined as compatible compositions of such operators, see below. We note the following:
\begin{itemize}
    \item $\sA_\omega^{\gamma, \xi}$ essentially extracts the singularity of the analytic continuation of $\hatphi$ to a neighborhood of $\omega$. The lift to $\tC$ is necessary to accommodate the definition of $\SING$.
    \item $\sA_\omega^{\gamma,\xi} \delta^{(i)} = 0$ by definition. This definition is natural, as a pole at the origin has a unique analytic continuation to a holomorphic function on all of $\IC^*$.
    \item The lift of a pole via \eqref{eq:deffanaforA} to an element of $\ana$   does not depend on the choice of $\xi$. The lift of a function $\hatpsi(\zeta) \anyLog(\zeta)$, for $\hatpsi \in \IC\{\zeta\}$ and $\anyLog$ denoting a specific branch of the logarithm, does depend on $\xi$. Different choices differ by $2\pi \ii n \hatpsi$ for some $n \in \IZ$, hence project onto the same element in $\SING$.
\end{itemize}
We can now consider the subset of $\Omega$-resurgent functions defined by
\begin{definition}
    We call a simply ramified $\Omega$-resurgent function of order $n$ any $\Omega$-resurgent function $\phiOR$ such that, for all $(\omega, \gamma, \xi)$ as in Definition \ref{def:Agen}, $\sA_\omega^{\gamma, \xi} \phiOR$ is a simply ramified singularity of order $n$. We denote the subset of all simply ramified $\Omega$-resurgent functions of order $n$ contained in $\sRhat_\Omega$ by $\sRhat^{\sr,n}_\Omega$, and denote by $\sRhat^{\sr}_\Omega$ the set of simply ramified $\Omega$-resurgent functions of arbitrary order contained in $\sRhat_\Omega$.
\end{definition}
The set of all simply ramified $\Omega$-resurgent functions is thus
\begin{align}
    \Esing_0 :=\bigoplus_{n\in \IN} \IC \delta^{(n)} \oplus \sRhat^{\sr}_\Omega \,.
\end{align}
Note that by the comment above, we can omit the superscript $\xi$ from $\sA^{\gamma,\xi}_\omega$ when acting on simple ramified $\Omega$-resurgent functions. We will do so in the following.

In simple terms, if $\hatphi \in \sRhat^{\sr}_\Omega$, then for any $\omega \in \Omega$, any analytic continuation of $\hatphi$ along a path $\gamma$ leading close to $\omega$ has the form
\begin{equation} \label{eq:anContPhiSram}
     (\mathrm{cont}_\gamma \hat\varphi) (\omega +\zeta) = \sum_{k=1}^n \frac{(-1)^{k-1} (k-1)! \,a_{-k}}{2 \pi \ii \,\zeta^{k}} + \hat{\psi}(\zeta) \frac{\anyLog(\zeta)}{2\pi \ii} + R(\zeta) \,,
\end{equation}
for $\zeta$ close to zero and $R(\zeta) \in \IC\{\zeta\}$. The alien operator along $\gamma$ at $\omega$ acting on $\hatphi$ then yields
\begin{equation}
    \sA_\omega^{\gamma} \hphi = \sum_{k=1}^{n} a_{-k} \delta^{(k-1)} + \hat{\psi} \,.
\end{equation}
To allow for composing alien operators of the form $\sA^\gamma_\omega$, it is convenient to invoke the isomorphism \eqref{eq:isoSingSimpInf} to regard $\Omega$-resurgent functions as simply ramified singularities. With this understanding, it follows that \cite[Lemma 6.54]{Sauzin}
\begin{align}
    \sA^\gamma_\omega \quad : \quad \omegaResurgentSimplyRamified \quad \longrightarrow \quad \bigoplus_{n\in \IN} \IC \delta^{(n)} \oplus \sRhat^{\sr}_{-\omega+\Omega} \,.
\end{align}
It will prove convenient to speak of singularities based at $\omega \in \IC$ rather than the origin. We can straightforwardly generalize the definition of the spaces $\ana$ and $\SING$ to their analogues based at $\omega \in \IC$, and define the corresponding projection map $\sing_\omega$. We then define
\begin{align}
    \delta_\omega = \sing_\omega \left( \frac{1}{2 \pi \ii \,(\zeta-\omega)}\right) \,, \quad \delta^{(k)}_\omega = \sing_\omega \left( \frac{(-1)^k k!}{2 \pi \ii \,(\zeta-\omega)^{k+1}} \right) \,,\,\,k\ge 0\,.
\end{align}
We define a translation operator $\tau_\omega$ whose action on $\Esing_0$ is induced by
\begin{align}
    \tau_\omega(\delta^{(n)}) := \delta^{(n)}_\omega \,, \quad \tau_\omega(\hatphi)(\zeta) := \hatphi(\zeta - \omega) \,\,\,\mathrm{for}\,\,\,\hatphi\in \IC\{\zeta\}\,.
\end{align}
and use it define the space of simply ramified $\Omega$-resurgent functions based at $\omega$ via
\begin{align}
    \Esing_\omega := \tau_\omega \left(\bigoplus_{n\in \IN} \IC \delta^{(n)}\oplus \sRhat^{\sr}_{-\omega + \Omega}\right) \,.
\end{align}
We finally define the pointed alien operator $\sA^\gamma_\omega$ as
\begin{align}
    \dot{\sA}^\gamma_\omega := \tau_\omega \circ \sA^\gamma_\omega \quad : \quad \Esing_0 \rightarrow \Esing_\omega \,.
\end{align}
Using the translation operator, we extend the definition of the convolution product on singularities to $\Esing_\omega$ via
\begin{align} \label{eq:pointedConvolution}
    (\varphi,\psi) \in (\Esing_{\omega'},\Esing_{\omega''}) \quad \Longrightarrow \quad \varphi * \psi := \tau_{\omega'+\omega''} \left(\tau_{\omega'}^{-1} \varphi * \tau_{\omega''}^{-1}\psi \right) \in \Esing_{\omega'+\omega''} \,.
\end{align}
The fact that the product is an element of $\Esing_{\omega'+\omega''}$ is a consequence of Theorem \ref{thm:towardsStokesAuto} below. Note that we may have to replace $\Omega$ by $\Omega \cup (\Omega + \Omega)$. In the following, we will assume that $\Omega$ is stable under addition. 

\subsection{Alien operators and rays of singularities} \label{app:alienOpsAndRays}
We will in particular be interested in considering singularities along a ray $d$, as this notion will directly enter in the consideration of the Stokes automorphism. Let $\omega \in \Omega$ and define $[0,\omega] \subset d$ as the straight line from $0$ to $\omega$. We introduce
\begin{equation}
    \Omega \cap [0,\omega] = \{\omega_0, \ldots, \omega_r \} \,, \quad \omega_0 = 0\,,\,\, \omega_r = \omega \,,
\end{equation}
where the $\omega_i$ are ordered by absolute value. We define canonical paths labeled by $\epsilon = (\epsilon_1 , \ldots ,\epsilon_{r-1}) \in \{+,-\}^{r-1}$ from 0 to $\omega$, where we deviate from the straight line connecting 0 to $\omega$ to avoid $\omega_i$ via a small semi-circle to the left or to the right (as we move in the direction from 0 to $\omega$) depending on whether $\epsilon_i = -$ or $+$. Based on this class of paths, we define two alien operators
\begin{equation}
    \Delta_\omega^+ = \mathscr{A}^{(+,\ldots,+)}_\omega \,, \quad \Delta_\omega = \sum_{\epsilon \in \{+,-\}^{r-1}} \frac{p(\epsilon)q(\epsilon)}{r!}  \mathscr{A}^{\epsilon}_\omega \,,
\end{equation}
where $p(\epsilon)$ and $q(\epsilon)$ count the number of `$+$' and `$-$' symbols respectively in the tuple $\epsilon$. These two classes of operators are not independent. They satisfy the relation \cite[Theorem 6.72]{Sauzin}
\begin{align} \label{eq:Dloglike}
    \Delta_\omega &= \sum_{s=1}^\infty \frac{(-1)^{s+1}}{s} \sum_{(\eta_1, \ldots, \eta_{s-1}) \in \Sigma(s, \omega, \Omega)} \Delta^+_{\omega-\eta_{s-1}} \circ \cdots \circ \Delta^+_{\eta_{2}-\eta_{1}} \circ \Delta^+_{\eta_1} \,, \\ \label{eq:Dexplike}
    \Delta^+_\omega &= \sum_{s=1}^\infty \frac{1}{s!} \sum_{(\eta_1, \ldots, \eta_{s-1}) \in \Sigma(s,\omega,\Omega)} \Delta_{\omega-\eta_{s-1}} \circ \cdots \circ \Delta_{\eta_{2}-\eta_{1}} \circ \Delta_{\eta_1} \,, 
\end{align}
where $\Sigma(s, \omega, \Omega)$ is the set of all increasing sequences (ordered by absolute value) $(\eta_1, \ldots, \eta_s) \in (0,\omega) \cap \Omega$ and for $s=1$, the composite operators in the summands are interpreted as $\Delta^+_\omega$, $\Delta_\omega$ respectively. These relations can be written in more convenient form as follows:
\begin{itemize}
    \item We eliminate the need to shift $\Omega$ upon acting with an alien operator by considering pointed alien operators,
    \begin{align}
        \dot{\Delta}^+_\omega := \tau_\omega \circ \Delta^+_\omega \,, \quad \dot{\Delta}_\omega := \tau_\omega \circ \Delta_\omega \quad:\quad \Esing_0 \rightarrow \Esing_\omega \,.
    \end{align}
    \item We straightforwardly generalize the definition of $\dot{\Delta}^+_\omega$, $\dot{\Delta}^+_\omega$ to obtain, for $i<j$,
    \begin{align}
        \dotDelta^+_{\omega_i, \omega_j} := \tau_{\omega_j} \circ \Delta^+_{\omega_i+\omega_j} \circ \tau_{\omega_i}^{-1}\,, \quad \dotDelta_{\omega_i, \omega_j} := \tau_{\omega_j} \circ \Delta_{\omega_i+\omega_j} \circ \tau_{\omega_i}^{-1}\quad: \quad \Esing_{\omega_i} \rightarrow \Esing_{\omega_j} \,.
    \end{align}
    \item We set
    \begin{align}
        \Esing(d) = \hat{\bigoplus}_{\omega \in \Omega \cap d} \Esing_\omega
    \end{align}
    (the hat above the direct sum symbol indicates completion, i.e. infinite sums are allowed). We call $\singphi \in \Esing_\omega \hookrightarrow \Esing(d)$ a homogeneous element of $\Esing(d)$ of degree $\omega$ (or, by keeping track of the ordering along $d$, of degree $m$ for $\omega = \omega_m$). Note that $\Esing(d)$ is an algebra with regard to the convolution product \eqref{eq:pointedConvolution}. We define, for $r\in \IN^*$, operators 
    \begin{align}
        \dotDeltaPlusR\,, \,\, \dotDeltaR \quad : \quad \Esing(d) \rightarrow \Esing(d)
    \end{align}
    such that
    \begin{align}
        \dotDeltaPlusR |_{\Esing_{\omega_i}} = \dot{\Delta}^+_{\omega_i, \omega_{i+r}} \,, \quad         \dotDeltaR |_{\Esing_{\omega_i}} = \dot{\Delta}_{\omega_i, \omega_{i+r}} \,.
    \end{align}
\end{itemize}
With these definitions, the relations \eqref{eq:Dloglike} and \eqref{eq:Dexplike} yield the following \cite[Theorem 6.73]{Sauzin}\footnote{Note that Sauzin proves the theorem for simple singularities, but the extension to simply ramified singularities is immediate.} 
\begin{theorem}
    With the notation from above, define the two operators
    \begin{align}
        \Delta^+_d := \mathrm{Id} + \sum_{r\in \IN^*}\dotDeltaPlusR \,, \quad \Delta_d :=\sum_{r \in \IN^*} \dotDeltaR \,.
    \end{align}
    Then the following two relations hold:
    \begin{align} \label{eq:DeltaDandDeltaPlusD}
        \Delta_d = \sum_{s\in \IN^*} \frac{(-1)^{s-1}}{s} \left(\Delta_d^+ - \mathrm{Id} \right)^s \,, \quad \Delta^+_d = \sum_{s\in \IN} \frac{1}{s!} (\Delta_d)^s \,.
    \end{align}
\end{theorem}



\subsection{Back to formal series in the $z$-plane} \label{app:backToz}
The Borel transform $\sB$ as defined in \eqref{eq:Boreltransform} induces a linear isomorphism between $z\,\IC[[z]]_1$ and $\IC \{\zeta\}$. Under this isomorphism, we have
\begin{align} \label{eq:isoDzetaDivz}
    \sB\left(\frac{1}{z} \tildephi\right) = \frac{d}{d\zeta}\sB(\tildephi) \,, \quad \tildephi \in z^2\,\IC[[z]]_1 \,,
\end{align}
i.e. the operation multiplication by $\tfrac{1}{z}$ intertwines with $\tfrac{d}{d\zeta}$ under Borel transform. By enlarging $z \, \IC[[z]]_1$ to\footnote{\label{footnote:defGevreyLaurent}Elements of this space are sums of polynomials in $\tfrac{1}{z}$ and Gevrey-1 series in $z$. The notation $\IC((z))$ for formal Laurent series is standard, the notation $\IC((z))_1$ to our knowledge is not.}
\begin{align} \label{eq:extendingGevrey}
    z\, \IC[[z]]_1 \hookrightarrow \varinjlim_n \, z^{-n}\, \IC[[z]]_1 =: \IC((z))_1 \,, 
\end{align}
we obtain a space which is closed under multiplication by $\tfrac{1}{z}$. We correspondingly extend $\IC\{\zeta\}$ by adjoining formal generators $\delta^{(n)}$,
\begin{align} \label{eq:extendingConvPower}
    \IC\{\zeta\} \hookrightarrow \bigoplus_{n\in \IN} \IC \,\delta^{(n)}\oplus \IC\{\zeta\} =: \IC\{\zeta\}_{\delta}\,,
\end{align}
such that we can extend the definition of the Borel transform $\sB$ via
\begin{align}
    \sB\left(\frac{1}{z^n}\right) = \delta^{(n)} \,, \quad n \in \IN \,.
\end{align}
In the following, we will also write $\delta := \delta^{(0)}$. To retain the intertwining operator identity between $\times \tfrac{1}{z}$ and $\tfrac{d}{d\zeta}$, we must then extend the action of $\tfrac{d}{d\zeta}$ to $\IC\{\zeta\}_{\delta}$ by defining it as differentiation on $\zeta \IC\{\zeta\}$, and imposing furthermore
\begin{align}
    \frac{d}{d\zeta} 1 = \delta \,, \quad \frac{d}{d\zeta} \delta^{(n)} = \delta^{(n+1)} \,\,\,\, \mbox{for} \,\, n \in \IN\,.
\end{align}
Note that this action intertwines with the action of $\tfrac{d}{d\zeta}$ upon embedding $\IC\{\zeta\}_{\delta}$ into $\SING$, see in particular \eqref{eq:diffIsoSing}.\footnote{Indeed, to avoid a proliferation of symbols and following the conventions in the literature, we have used the same symbols $\delta^{(n)}$ and $\tfrac{d}{d\zeta}$ on both sides of this map.} 

Note also that there exists a unique extension of the convolution product to $\IC\{\zeta\}_{\delta}$ such that the Borel transform $\sB$ induces an algebra homomorphism between $\IC((z))_1$ and $\IC\{\zeta\}_{\delta}$: $\delta$ must act as the convolution product identity, and the product rule $\tfrac{d}{d\zeta} (f*g) = (\tfrac{d}{d\zeta} f)*g=f* (\tfrac{d}{d\zeta} g)$ must be retained.

We extend the definition of alien operators to the space $\IC((z))_1$ by composing the definitions above by $\sB^{-1}$. We similarly introduce the counterpart of sets defined on the Borel plane by taking their inverse image under the Borel transform. E.g., we define
\begin{align}
    \sRtilde^{\sr}_\Omega = \sB^{-1}\sRhat^{\sr}_\Omega \,.
\end{align}
To apply the pointed alien operators to formal power series in $z$, we can formally map the space $\tau_\omega \,\IC\{\zeta\}$ to the $z$-plane via
\begin{align}
    \sB^{-1} \left(\tau_\omega \hatphi \right) = \sB^{-1} \left( \e^{-\omega \,\tfrac{d}{d \zeta}} \right) \sB^{-1}(\hatphi) = \e^{-\omega/z} \tildephi \,,
\end{align}
where we have applied \eqref{eq:isoDzetaDivz} termwise. Note that the right-hand side, upon expanding the exponential, is not an element of $C((z))$, as formal Laurent series have a leading pole of finite order. It is therefore best to think of it as a symbol rather than a power series \cite[Section 6.12.4]{Sauzin}. We also define
\begin{align}
    \sB^{-1} (\delta^{(n)}_\omega) = \e^{-\omega/z} \sB^{-1}(\delta^{(n)}) \,. 
\end{align}
The notation is suggestive, as for $\tildephi$ fine-summable,
\begin{align} \label{eq:symbolvsExponential}
    \sL^\theta \circ \sB(\e^{-\omega/z} \tildephi) = \sL^\theta (\tau_\omega \hatphi) = \e^{-\omega/z} \sS^\theta (\tildephi) \,,
\end{align}
with $\e^{-\omega/z}$ on the left-hand side indicating the exponential symbol, on the right-hand side the exponential function. In terms of this symbol, we have e.g., for $\tildephi \in \IC((z))_1$,
\begin{align} \label{eq:dottedUndotted}
    \dotDelta_\omega \tildephi = \e^{-\omega/z} \Delta_\omega \tildephi \,.
\end{align}

\subsection{The asymptotics of the the coefficients of the formal power series $\tildephi(z)$} \label{app:asympCoeff}
    \usetikzlibrary{arrows.meta,decorations.markings}
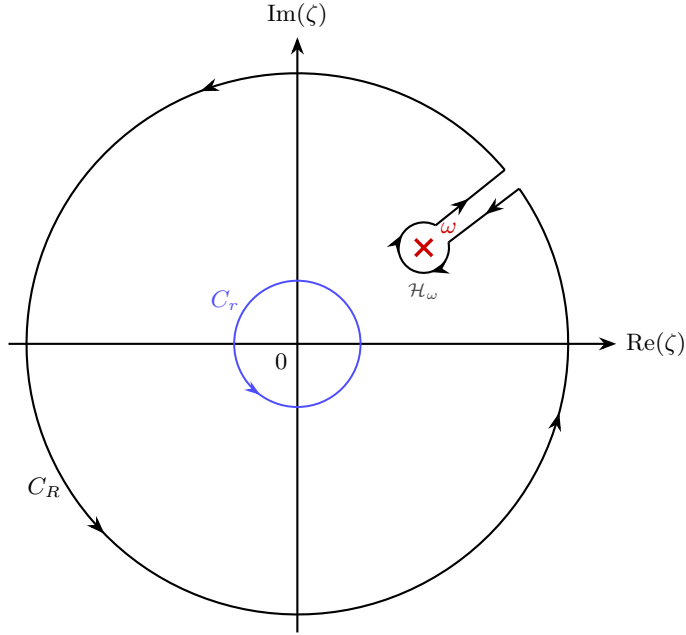
\begin{figure}[H]
    \centering\begin{tikzpicture}[
  scale      = 0.80,
  line width = 0.80pt,
  >=Stealth,
]

\draw[->] (-4.8,0)--(5.3,0)
  node[right,font=\small]{$\operatorname{Re}(\zeta)$};
\draw[->] (0,-4.8)--(0,5.1)
  node[above,font=\small]{$\operatorname{Im}(\zeta)$};
\node[below left,font=\small] at (0,0){$0$};


\draw[red!80!black,line width=1.3pt]
  (1.96,1.46)--(2.24,1.74)
  (1.96,1.74)--(2.24,1.46);
\node[red!80!black,above right,font=\small]
  at (2.22,1.68){$\omega$};

\draw[blue!70,line width=0.75pt,
  postaction={decorate},decoration={markings,
    mark=at position 0.65 with
      {\arrow[blue!70]{Stealth[length=6pt,width=4pt]}}}
] (0,0) circle[radius=1.05];
\node[blue!70,font=\small] at (-1.2,0.72){$C_r$};


\draw[
  postaction={decorate},decoration={markings,
    mark=at position 0.14 with{\arrowreversed{Stealth[length=7pt,width=5pt]}},
    mark=at position 0.48 with{\arrowreversed{Stealth[length=7pt,width=5pt]}},
    mark=at position 0.80 with{\arrowreversed{Stealth[length=7pt,width=5pt]}}}
] ({4.5*cos(35)},{4.5*sin(35)})
  arc[start angle=35, delta angle=-355, radius=4.5];

\node[font=\small] at ({4.5*cos(213)-.43},{4.5*sin(213)+.06})
  {$C_R$};

\draw[
  postaction={decorate},decoration={markings,
    mark=at position 0.52 with{\arrowreversed{Stealth[length=7pt,width=5pt]}}}
] ({4.5*cos(40)},{4.5*sin(40)}) -- (2.297,1.971);

\draw[
  postaction={decorate},decoration={markings,
    mark=at position 0.28 with{\arrowreversed{Stealth[length=7pt,width=5pt]}},
    mark=at position 0.72 with{\arrowreversed{Stealth[length=7pt,width=5pt]}}}
] (2.297,1.971)
  arc[start angle=62, delta angle=310, radius=0.42];

\draw[
  postaction={decorate},decoration={markings,
    mark=at position 0.52 with{\arrowreversed{Stealth[length=7pt,width=5pt]}}}
] (2.511,1.687) -- ({4.5*cos(35)},{4.5*sin(35)});


\node[font=\scriptsize,gray!55!black]
  at (2.1,{1.6-0.42-0.30}){$\mathcal{H}_\omega$};


\end{tikzpicture}
    \caption{Integration contour for asymptotic estimate of coefficients of $\tildephi$} \label{fig:asymCoeff}
\end{figure}
Let $\tildephi(z) = \sum_{m=0}^\infty c_m z^{m+1} \in \sRtilde^{\sr}_\Omega$. Assume that there exists an $\omega \in \Omega$ closest to the origin, and let $\omega'$ be a singularity second closest to the origin, with $|\omega'|\le2|\omega|$. The Borel transform $\hatphi(\zeta) = \sum_{m=0}^\infty \frac{c_m}{m!} \zeta^{m}$ of $\tildephi(z)$ satisfies
\begin{align}
    \dotDelta_\omega \hatphi= \sum_{k=1}^n c^\omega_{-k} \delta_\omega^{(k-1)} + \hatphi^{\omega}(\zeta-\omega)  
\end{align}
for some $n$. Let $C_r$ be the integration contour indicated in Figure~\ref{fig:asymCoeff} for some $r < |\omega|$. Then
\begin{align}
    \frac{c_m}{m!} = \frac{1}{2\pi \ii} \oint_{C_r} \frac{\hatphi(\zeta)}{\zeta^{m+1}} d\zeta=: I(C_r) \,.
\end{align}
Next, we push out the integration contour to $C_R$ as depicted in Figure \ref{fig:asymCoeff}, with $|\omega| < R < |\omega'|$ and let $C_R' = C_R - \cH_\omega$, such that
\begin{align}
    \frac{c_m}{m!} = I(C_R) =  I(\cH_\omega) + I(C'_R)\,.  
\end{align}
The integral $I(C'_R)=:\cE_{C'_R}$ constitutes one of several error terms to the result we wish to establish. It can be bounded by
\begin{align} \label{eq:ECR}
    |\cE_{C'_R}| \le \frac{M_R}{R^m} \,, \quad M_R = \max_{\zeta\in \overline{C'_R}} |\hatphi(\zeta)|\,,
\end{align}
where by abuse of notation, we are using $\hatphi$ to denote the analytic continuation, avoiding $\ray_\omega$, of $\hatphi \in \sRhat^{\sr}_\Omega$  to an open set containing $C'_R$. The integration around the Hankel-like contour $\cH_\omega$ yields
\begin{align} \label{eq:twoContributionsHankel}
    I(\cH_\omega) = -\frac{1}{2\pi \ii}\sum_{k=1}^n  \frac{(m+k-1)!}{m!}\frac{c^\omega_{-k}}{\omega^{m+k}} - \frac{1}{2\pi \ii} \int_\omega^{R\e^{\ii \theta_\omega}} \frac{\hatphi^\omega(\zeta-\omega)}{\zeta^{m+1}} d\zeta \,.
\end{align}
By assumption on $R$, $\hatphi^\omega(\zeta-\omega)$ can be written as a  power series
\begin{align} \label{eq:defckomega}
    \hatphi^\omega(\zeta - \omega) = \sum_{k=0}^\infty \frac{c^\omega_k}{k!} (\zeta-\omega)^k
\end{align}
which converges in the disk $D(\omega, |\omega'-\omega|)$. Choosing $N$ such that $N+1<m$, we write
\begin{align}
    \hatphi^\omega(\zeta - \omega) = \sum_{k=0}^N \frac{c^\omega_k}{k!} (\zeta-\omega)^k + r_{N+1}(\zeta) \,.
\end{align}
We can estimate the error term $r_{N+1}(\zeta)$ for $\zeta \in D(\omega,|\omega'-\omega|)$ by choosing an integration contour $C_\rho(\omega)$ of radius $\rho$ satisfying 
\begin{align} \label{eq:Randrho}
    R< |\omega|+\rho <|\omega'| \,,
\end{align}
such that $[\omega, R\,e^{\ii \theta_\omega}] \subset D(\omega,\rho)$, to write for $\zeta \in D(\omega,\rho)$
\begin{align}
    r_{N+1}(\zeta) &= \frac{1}{2\pi \ii}\left[\oint_{C_\rho} \frac{\hatphi^\omega(\xi-\omega)}{\xi-\zeta}d\xi-\sum_{k=0}^N (\zeta-\omega)^k \oint_{C_\rho}\frac{\hatphi^\omega(\xi-\omega)}{(\xi-\omega)^{k+1}}d\xi \right]\\
    &= \frac{(\zeta-\omega)^{N+1}}{2\pi \ii}\oint_{C_\rho} \frac{\hatphi^\omega(\xi-\omega)}{(\xi-\omega)^{N+1}(\xi-\zeta)}d\xi \,,
\end{align}
yielding the estimate
\begin{align} \label{eq:rEstimate}
    |r_{N+1}(\zeta)| \le \left(\frac{|\zeta-\omega|}{\rho}\right)^{N+1}\frac{M^\omega_\rho}{1-\frac{|\zeta-\omega|}{\rho}} \,, \quad M^\omega_\rho = \max_{\xi \in \overline{C_\rho}}|\hatphi^\omega(\xi-\omega)| \,.
\end{align}
Then
\begin{align}
    - \frac{1}{2\pi \ii} \int_\omega^{R\e^{\ii \theta_\omega}} \frac{\hatphi^\omega(\zeta-\omega)}{\zeta^{m+1}} d\zeta = I + \cE_{r_{N+1}}
\end{align}
with 
\begin{align} \label{def:Ernp1}
    I = - \frac{1}{2\pi \ii} \sum_{k=0}^N \frac{c_k^\omega}{k!} \int_{\omega}^{R\e^{\ii \theta_\omega}} \frac{(\zeta-\omega)^k}{\zeta^{m+1}} d\zeta \,, \quad     \cE_{r_{N+1}} = - \frac{1}{2\pi \ii}  \int_{\omega}^{R\e^{\ii \theta_\omega}} \frac{r_{N+1}(\zeta)}{\zeta^{m+1}} d\zeta \,.
\end{align}
Introducing the variables $s$ and $\epsilon$ via $\zeta = \omega(s+1)$ and $\epsilon = R/|\omega|-1$, we obtain
\begin{align}
    \int_{\omega}^{R \e^{\ii \theta_\omega}} \frac{(\zeta-\omega)^k}{\zeta^{m+1}} d\zeta &= \frac{1}{\omega^{m-k}}\int_0^\epsilon \frac{s^k}{(1+s)^{m+1}} ds \\
    &= \frac{1}{\omega^{m-k}}\big( B(m-k,k+1)- \cE_\infty(k) \big) \,,
\end{align}
where the beta function $B(z_1,z_2)$ satisfies
\begin{align}
    B(z_1,z_2) = \int_0^1 t^{z_1-1}(1-t)^{z_2-1} dt = \frac{\Gamma(z_1)\Gamma(z_2)}{\Gamma(z_1+z_2)} \quad \mathrm{for} \quad \re(z_1), \re(z_2) >0 \,,
\end{align}
and 
\begin{align}
    \cE_\infty(k) =  \int_\epsilon^\infty \frac{s^k}{(1+s)^{m+1}} ds 
\end{align}
satisfies the estimate
\begin{align}
    |\cE_\infty(k)| \le \int_\epsilon^\infty (1+s)^{k-m-1} ds = \frac{1}{m-k}\frac{1}{(1+\epsilon)^{m-k}} = \frac{1}{m-k} \left(\frac{|\omega|}{R}\right)^{m-k}  \,.
\end{align}
The correction term due to extending the integral in $I$ to infinity is thus given by
\begin{align} \label{def:Einf}
    \cE_\infty = \frac{1}{2\pi \ii} 
    \frac{1}{\omega^m}\sum_{k=0}^N \frac{c^\omega_k \,\omega^k}{k!} \cE_\infty(k) \,,
\end{align}
and satisfies the estimate
\begin{align} \label{eq:Einf}
    |\cE_\infty| \le \frac{C_\infty}{(m-N)R^{m-N}} \,, \quad C_\infty = \frac{1}{2\pi} \max_{k=0,\ldots,N} \frac{|c^\omega_k|}{k!} \,.
\end{align}
It remains to estimate $\cE_{r_{N+1}}$. With \eqref{eq:rEstimate}, we obtain
\begin{align}
    |\cE_{r_{N+1}}| &\le \frac{1}{2\pi} \frac{M^\omega_\rho}{\rho^{N+1}} \int_{\omega}^{R \e^{\ii \theta_\omega}} \frac{1}{1-\frac{|\zeta-\omega|}{\rho}}\frac{(\zeta-\omega)^{N+1}}{\zeta^{m+1}}d\zeta \le \frac{1}{2\pi} \frac{M^\omega_\rho}{\rho^{N} (\rho+|\omega|-R)} \int_{\omega}^{R \e^{\ii \theta_\omega}} \frac{(\zeta-\omega)^{N+1}}{\zeta^{m+1}}d\zeta \nonumber \\
    &\le \frac{1}{2\pi} \frac{1}{|\omega|^m}\frac{M^\omega_\rho}{1-\frac{R-|\omega|}{\rho}} \left(\frac{|\omega|}{\rho}\right)^{N+1} \frac{(m-N-2)!(N+1)!}{m!} \,. \label{eq:ERN}
\end{align}
Putting this all together, we finally arrive at
\begin{align} \label{eq:pertAsymp}
    \frac{c_m}{m!} = -\frac{1}{2\pi \ii} \frac{1}{m!} \sum_{k=-n}^N \frac{(m-k-1)!}{\omega^{m-k}}c^\omega_k + \cE_\infty + \cE_{r_{N+1}} + \cE_{C'_R} \,,
\end{align}
with the error estimates \eqref{eq:Einf}, \eqref{eq:ERN} and \eqref{eq:ECR}.
It is clear that in the $m\gg1$ limit at constant $N$, all three error terms become small, yielding the classical estimate of the asymptotics of the coefficients $c_m$. To go beyond this estimate, note that $\cE_{r_{N+1}}$ is suppressed by a factor $ \tfrac{1}{|\omega|^m}$, vs. $\tfrac{1}{R^m}$ and $\tfrac{1}{R^{m-N}}$ for $\cE_{C'_R}$ and $\cE_\infty$ respectively. The term $1-\tfrac{R-|\omega|}{\rho}$ in the denominator of the estimate for $\cE_{r_{N+1}}$ reflects the fact that the estimate of the error term $r_{N+1}$ deteriorates as $\rho$ approaches the boundary of analyticity of $\hatphi^\omega$. Given \eqref{eq:Randrho}, we should hence resist the temptation of choosing $R$ too close to $|\omega'|$. As we discuss in the text, we can mitigate the errors $\cE_\infty$ and $\cE_{r_{N+1}}$ by performing the second contribution to the integral $I(\cH_\omega)$ in \eqref{eq:twoContributionsHankel} numerically, modeling $\hatphi^{\omega}$ by its Padé approximant.

\subsection{The Laplace transform of majors}
To trivially extend Theorem \ref{thm:LaplaceToAsymptotic} to the space 
\begin{align} \label{eq:CzetaDeltaLaplaceable}
    \bigoplus_{n \in \IN} \IC \, \delta^{(n)} \oplus \cN(\e^{\ii \theta}\IR^+) \,,
\end{align}
we define
\begin{align}
    \sL^{\theta}\left(\delta^{(n)}\right) = \frac{1}{z^n} \,.
\end{align}
Note that upon embedding \eqref{eq:CzetaDeltaLaplaceable} into $\SING$, we can obtain the action of $\sL^{\theta}$ by the following operation \cite[Section 6.9.2]{Sauzin}: we choose as major $\phimajor$ of the singularity 
\begin{align}
    \singphi = \sum_{k=1}^n a_{-k}\delta^{(k-1)} + \hatphi(\zeta) 
\end{align}
the singular germ
\begin{align} \label{eq:LaplaceMajors}
    \phimajor = \sum_{k=1}^n \frac{(-1)^{k-1} (k-1)! \,a_{-k}}{2 \pi \ii \,\zeta^{k}} + \hatphi(\zeta) \frac{\log(\zeta)}{2\pi \ii} \,,
\end{align}
and compute
\begin{align}
    \int_{\Gamma_{\theta,\epsilon}} \e^{-\frac{\zeta}{z}} \phimajor(\zeta)d\zeta \,,
\end{align}
where the integration contour $\Gamma_{\theta, \epsilon}$ comes in from infinity along the ray $\underline{\e}^{\ii (\theta-2\pi)}[\epsilon, \infty)$, encircles the origin counterclockwise along the circle of radius $\epsilon$, and then extends back to infinity along the ray $\underline{\e}^{\ii \theta} [\epsilon, \infty)$. The integral \eqref{eq:LaplaceMajors} is what Sauzin calls the Laplace transform of majors. Note that it is exactly this type of integral, albeit with integration path shifted away from the origin, which arises when computing Stokes automorphisms, see Section \ref{app:DeltaplusAndStokesAuto}.

\subsection{The operators $\Delta^+_d$ are automorphisms, and the operators $\Delta_d$ derivations}
The main rationale behind introducing the operator $\Delta_\omega$ alongside the more straightforward $\Delta^+_\omega$ is that the former define derivations with regard to the convolution product. This is shown in \cite[Section 6.13]{Sauzin} in the case of simple singularities. The proof carries over directly to the case of simply ramified singularities, the main difficulty being dealing with the logarithmic component of the singularity. The first step is to show that $\Delta^+_d$ is an automorphism of the algebra $\Esing(d)$. Directly generalizing \cite[Theorem 6.83]{Sauzin}, we have
\begin{theorem} \label{thm:towardsStokesAuto}
    Let $\tildephi \in \sRtilde^{\sr}_\Omega$, $\tildepsi \in \sRtilde^{\sr}_\Omega$. Then $\tildephi \tildepsi \in \sRtilde^{\sr}_\Omega$ and, for every $\omega \in \Omega - \{0\}$,
    \begin{align} \label{eq:deltaPlusProduct}
        \Deltaplus_{\omega} (\tildephi \tildepsi) &= (\Deltaplus_{\omega} \tildephi)\tildepsi + \sum_{\substack{\omega = \omega' + \omega'' \\ \omega', \omega'' \in (0,\omega)}} (\Deltaplus_{\omega'} \tildephi)(\Deltaplus_{\omega''}\tildepsi) + \tildephi (\Deltaplus_{\omega} \tildepsi) \,.
    \end{align}
\end{theorem}
\begin{proof}
    Upon Borel transformation, the identity \eqref{eq:deltaPlusProduct} maps to the relation
    \begin{align} \label{eq:deltaPlusProductBorel}
        \Deltaplus_\omega(\hatphi * \hatpsi) = \sum_{\eta \in \{0,\omega\} \cup \Sigma_\omega} (\Deltaplus_{\eta} \hatphi)*(\Deltaplus_{\omega-\eta} \hatpsi) \,,
    \end{align}
    with the understanding that $\Deltaplus_0 \hatphi = \hatphi$, $\Deltaplus_0 \hatpsi = \hatpsi$. 
    The proof of this relation proceeds by induction. Let us introduce the notation
    \begin{align}
        \EsingOmegaN{0}{n} :=\bigoplus_{k=0}^{n-1} \IC \delta^{(k)} \oplus \sRhat^{\mathrm{s.ram,n}}_\Omega \,.
    \end{align}
    Sauzin shows the relation for $\hatphi, \hatpsi \in \EsingOmegaN{0}{1}$ \cite[Theorem 6.83]{Sauzin}. Assume that it holds for $\hatphi, \hatpsi \in \EsingOmegaN{0}{n-1}$. Now let $\hatphi, \hatpsi \in \EsingOmegaN{0}{n}$, such that
    \begin{align}
        \Deltaplus_\eta \hatphi = a^\eta_{-n}\delta^{(n-1)} + \ldots + a^\eta_{-1}\delta + \hatphi_\eta \,, \quad \Deltaplus_{\omega-\eta} \hatpsi = b^{\omega-\eta}_{-n}\delta^{(n-1)} + \ldots + b^{\omega-\eta}_{-1}\delta + \hatpsi_{\omega-\eta} \,.
    \end{align}
    Define
    \begin{align}
        \hatphi^* = 1*\hatphi \,, \quad \hatpsi^* = 1*\hatpsi \,.
    \end{align}
    Then $\hatphi^*, \hatpsi^* \in \EsingOmegaN{0}{n-1}$, and
    \begin{align}
        \Deltaplus_\eta \hatphi^* = a^\eta_{-n}\delta^{(n-2)} + \ldots + a^\eta_{-1} + 1*\hatphi_\eta \,, \quad \Deltaplus_{\omega-\eta} \hatpsi^* = b^{\omega-\eta}_{-n}\delta^{(n-2)} + \ldots + b^{\omega-\eta}_{-1} + 1*\hatpsi_{\omega-\eta} \,.
    \end{align}
    By the induction hypothesis, $\Deltaplus_\omega(\hatphi^* * \hatpsi^*)$ satisfies \eqref{eq:deltaPlusProductBorel}. Hence, $\left( \tfrac{d}{d\zeta}\right)^2 \Deltaplus_\omega(\hatphi^* * \hatpsi^*)$  contains three types of terms: 
    \begin{enumerate}
        \item \begin{align}
                \left( \tfrac{d}{d\zeta}\right)^2 \left(a^\eta_{-j} b^{\omega-\eta}_{-k} \,\delta^{(j-2)}*\delta^{(k-2)} \right) & = a^\eta_{-j} b^{\omega-\eta}_{-k} \,\left( \tfrac{d}{d\zeta}\right)^2 \delta^{(j+k-4)} =a^\eta_{-j} b^{\omega-\eta}_{-k} \, \delta^{(j+k-2)} 
            \end{align}
        \item \begin{align}
                \left( \tfrac{d}{d\zeta}\right)^2 \left(a^\eta_{-j} \,\delta^{(j-2)}* 1 * \hatpsi_{\omega-\eta} \right) & = a^\eta_{-j} \,\delta^{(j-1)} * \hatpsi_{\omega-\eta} 
            \end{align}
        \item \begin{align}
            \left( \tfrac{d}{d\zeta}\right)^2 \left( (a^\eta_{-1} + 1*\hatphi_\eta)*(b^{\omega-\eta}_{-1} + 1* \hatpsi_{\omega-\eta}) \right) &= \left( \tfrac{d}{d\zeta}\right)^2 \big(a^\eta_{-1} b^{\omega-\eta}_{-1} \zeta + a^\eta_{-1} \zeta * \hatpsi_{\omega-\eta} \nonumber \\ 
            &\hspace{2cm} +b^{\omega-\eta}_{-1}\zeta * \hatphi_\eta +\zeta * \hatphi_\eta * \hatpsi_{\omega-\eta} \big)  \nonumber \\
            &= a^\eta_{-1} b^{\omega-\eta}_{-1} \delta + a^\eta_{-1}  \hatpsi_{\omega-\eta} +b^{\omega-\eta}_{-1}\hatphi_\eta + \hatphi_\eta * \hatpsi_{\omega-\eta} \,,
        \end{align}
        where we have used \eqref{eq:diffIsoSing}.
    \end{enumerate}
These terms assemble to yield \eqref{eq:deltaPlusProductBorel}, completing the argument.
\end{proof}
The relation \eqref{eq:deltaPlusProductBorel}, applied termwise to $(\Phi,\Psi) \in \Esing(d) \times \Esing(d)$, immediately yields \cite[Theorem 6.84]{Sauzin}
\begin{align}
    \Deltaplus_d (\Phi * \Psi) = \Deltaplus_d(\Phi)*\Deltaplus_d(\Psi) \,.
\end{align}
The algebraic result that the logarithm of an automorphism is a derivation \cite[Lemma 6.87]{Sauzin} now permits us to conclude that $\Delta_d$ is a derivation. From this, we can easily conclude that the (unpointed) operator $\Delta_\omega$ is also a derivation: choose $\Phi$ and $\Psi$ to be homogeneous of degree 0, and then extract their image under $\Delta_\omega$ from the degree $\omega$ component.

\subsection{$\Deltaplus_d$ is the Stokes automorphism operator} \label{app:DeltaplusAndStokesAuto}
Consider a ray $d = \{ t \,\e^{\ii\theta}\,|\, t>0 \}$ in the direction $\theta$, a small open interval $I$ centered at $\theta$ and decomposed as $I = I^- \cup \{\theta\}\cup I^+$, $I^{+} = \{\theta' \in I \, | \, \theta' < \theta\}$, $I^{-} = \{\theta' \in I \, | \, \theta' > \theta\}$. We want to compare the Laplace transform to the left and to the right of the ray $d$ of an element $\hatphi$ of the set
\begin{align}
    \EsingOmegaN{0}{I} = \bigoplus_{n\in \IN} \IC \delta^{(n)} \oplus \left( \sRhat^{\sr}_\Omega \cap \cN(I^+) \cap \cN(I^-) \right)\,.
\end{align}
The comparison involves a contour deformation which picks up contributions from $\dotDeltaPlusR \hatphi \in \Esing_{\omega_r}$. To write down these contributions, we extend the definition of the Laplace transform to $\hatphi^{\omega} \in \tau_{\omega}(\EsingOmegaN{0}{I})$ via
\begin{align}
    \sL^{\theta^{\pm}} \hatphi^{\omega} := \e^{-\omega/z} \sL^{\theta^{\pm}} \left(\tau_{\omega}^{-1}\hatphi^{\omega} \right) \,.
\end{align}

\begin{figure}[H]
    \centering
    \usetikzlibrary{arrows.meta,decorations.markings}
    \begin{tikzpicture}[
	scale=1.2, font=\small,
	hankel/.style={red!65!black, thick,
		postaction={decorate, decoration={markings,
				mark=at position 0.5 with {\arrowreversed{Stealth[length=5pt,width=4pt]}}}}},
	kink/.style={red!65!black, thick,
		postaction={decorate, decoration={markings,
				mark=at position 0.5 with {\arrow{Stealth[length=5pt,width=4pt]}}}}},
	laplace/.style={blue!65!black, thick,
		postaction={decorate, decoration={markings,
				mark=at position 0.6 with {\arrow{Stealth[length=6pt,width=5pt]}}}}},
	]
	
	\def\alp{15}
	\def\r{0.15}
	\def\xmax{7.0}
	\pgfmathsetmacro{\XR}{\xmax*cos(\alp)}
	
	\draw[dashed, gray!60, thin] (0,0) -- (\xmax,0) node[right,gray!70]{$\theta$};
	\draw[laplace] (0,0) -- (\XR,{ \xmax*sin(\alp)}) node[right]{$\mathcal{L}_{\theta^-}$};
	\draw[laplace] (0,0) -- (\XR,{-\xmax*sin(\alp)}) node[right]{$\mathcal{L}_{\theta^+}$};
	\fill (0,0) circle (1.5pt) node[below left=1pt]{$0$};
	
	\foreach \w/\lab in {2.0/{\omega_1}, 3.8/{\omega_n}}{
		\pgfmathsetmacro{\ex}{\w + \r*cos(\alp)}
		\pgfmathsetmacro{\eyp}{ \r*sin(\alp)}
		\pgfmathsetmacro{\eym}{-\r*sin(\alp)}
		\pgfmathsetmacro{\offset}{(\XR - \ex)*tan(\alp)}
		\pgfmathsetmacro{\yfarIn} {\eym + \offset}
		\pgfmathsetmacro{\yfarOut}{\eyp + \offset}
		
		\draw[hankel] (\XR, \yfarIn)  -- (\ex, \eym);
		\draw[hankel] (\ex, \eym)
		arc[start angle={-\alp}, end angle={\alp - 360}, radius=\r];
		\draw[hankel] (\ex, \eyp) -- (\XR, \yfarOut);
		
		\fill (\w, 0) circle (1.8pt);
		\node[below=3pt] at (\w, 0) {$\lab$};
		\node[red!65!black, left=4pt] at ({\w + 5*\r - 0.05},0.25) {$\mathcal{H}_{\lab}$};
	}
	
	\node at (3.1, 0) {$\ldots$};
	
	\pgfmathsetmacro{\wk}{5.4}
	\pgfmathsetmacro{\yKinkIn} {(\XR - \wk)*tan(\alp)}
	\pgfmathsetmacro{\yKinkOut}{-(\XR - \wk)*tan(\alp)}
	\draw[kink] (\XR, \yKinkIn)  -- (\wk, 0);
	\draw[kink] (\wk, 0) -- (\XR, \yKinkOut);
	
\end{tikzpicture}
    \caption{Integration contour for evaluating Stokes automorphism}
\end{figure}
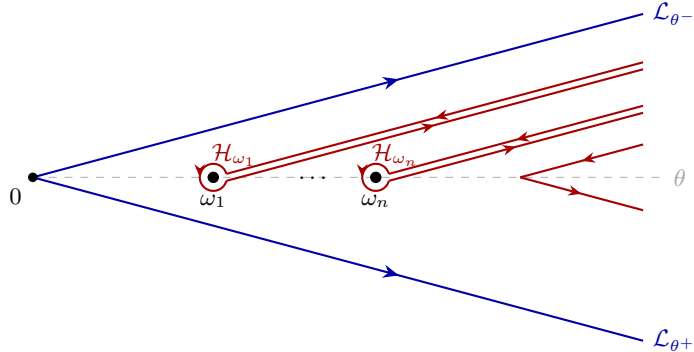
We have the following \cite[Theorem 6.77]{Sauzin}\footnote{Theorem 6.77 in  \cite{Sauzin} includes a uniformity statement for the error term which we omit, as we have not introduced the necessary definitions when introducing $\cN(I)$ above.}
\begin{theorem}
    Consider $m\in \IN^*$ and $\rho \in \IR$ such that $|\omega_m| < \rho < |\omega_{m+1}|$. Let $\hatphi \in \EsingOmegaN{0}{I}$ satisfy $\tau_{\omega_i}^{-1} \left(\Deltaplus_i \hatphi\right) \in \EsingOmegaN{0}{I}$ for $i=1, \ldots, m$. Then
    \begin{align} \label{eq:stokesAutoM}
        \sL^{\theta^+} \hatphi = \sL^{\theta^-}[\Deltaplus_d \hatphi]_m(z) + \cO \left(e^{-\rho \,\re(\e^{\ii\theta} z )} \right)\,,
    \end{align}
    where 
    \begin{align}
        [\Deltaplus_d \cdot]_m := \mathrm{Id} + \sum_{r=1}^m \dotDeltaPlusR\,.
    \end{align}
\end{theorem}
The contributions $\sL^{\theta^-} \dotDeltaPlusR \hatphi$ arise precisely via the Laplace transform of majors of $\dotDeltaPlusR \hatphi$.

Under favorable circumstances, the limit $m \rightarrow \infty$ exists and  the relation \eqref{eq:stokesAutoM} extends to
\begin{align}
    \sL^{\theta^+} = \sL^{\theta^-} \circ \Deltaplus_d \,.
\end{align}
To obtain the Stokes automorphism operator $\stokesAut_d$, which is typically defined as mapping formal power series to formal power series, we consider
\begin{align}
    \sS^{\theta^+} =\sL^{\theta^+} \circ \sB = \sL^{\theta^-} \circ \sB \circ \sB^{-1}\circ \Deltaplus_d \circ \sB= \sL^{\theta^-} \circ \sB \circ \stokesAut_d =\sS^{\theta^{-}} \circ \stokesAut_d \,.
    \label{eq:LaplaceStokes}
\end{align}
In the conventions that we have been using, in which we use the same symbol for operators acting in the Borel plane or the $z$-plane, we thus have
\begin{align} \label{eq:StokesAutoDeltaPlusExp}
    \stokesAut_d = \Deltaplus_d = \exp\left( \sum_{r \in \IN^*} \dotDeltaR \right) \,.
\end{align}
It is also convenient to introduce the discontinuity operator, defined as
\begin{align} \label{eq:discontinuity}
    \disc_{\theta} = \sS^{\theta^-}\left( \stokesAut_{\ray_{\omega_\gamma}} - \id\right) = \sS^{\theta^-} \left(\sum_{r \in \IN^*} \dotDeltaPlusR \right)\,.
\end{align}

\section{Two questions in parametric resurgence} \label{app:parametric}
In the text, we will need to invoke two operations on alien operators in the parametric context.

\subsection{Parametric derivatives and alien operators} \label{app:parametricDerivatives}
In Section \ref{ss:higherAlienPowers}, our reasoning requires commuting derivatives $\partial_t$ with regard to moduli $t$ (we will call such derivatives `parametric') and the alien derivative. The main technical step here is introducing a parametric version of the space of $\Omega$-resurgent functions $\sRhat^{\sr}_\Omega$ whose singularities are both simply ramified and differentiable in $t$. The commutation of $\partial_t$ and all alien operators is then essentially equivalent to imposing that the $t$-derivatives are holomorphic (this follows immediately if $\hatphi(\zeta,t)$ is $C^\infty$ as a function of $\zeta$ {\it and} $t$, simply by differentiating the Cauchy-Riemann relations), and that $t$ differentiation and analytic continuation commute (which uniqueness of analytic continuation should imply barring subtleties regarding the domain of analyticity of $\hatphi$ vs. its parametric derivatives). Working out a detailed definition of this function space will not concern us here.

Let us assume that the conditions ensuring the commutation are met. The commutation is most naturally expressed in terms of the pointed alien operators -- recall that these do not involve a translation of the singularity to the origin. In this case,
\begin{align} \label{eq:commuteParametricAlienDotted}
    \partial_t \, \dot{\sA}^\gamma_\omega = \dot{\sA}^\gamma_\omega \, \partial_t \,.
\end{align}
If we want to work at the level of unpointed alien operators, we need to correct for the fact that upon translation of the singularity to the origin, $\delta^{(k)}_{\omega(t)}$ and $\sing_0 \left(\tfrac{\log \left(\zeta-\omega(t)\right)}{2\pi \ii}\right)$ lose their $t$ dependence:
\begin{align} \label{eq:commuteParametricAlienUndotted}
    \left(-\dot{\omega} \partial_\zeta + \partial_t \right) \,\sA^\gamma_\omega = \sA^\gamma_\omega \,\partial_t  \,,
\end{align}
where the derivative $\partial_\zeta$ is to be understood in the sense of equation \ref{eq:diffIsoSing} (denoted as a partial derivative in the parametric setting here). At the level of formal Laurent series, this maps to the relation
\begin{align} \label{eq:commuteParametricAlienUndottedFormal}
    \left(-\frac{\dot{\omega}}{z} + \partial_t \right) \,\sA^\gamma_\omega = \sA^\gamma_\omega \,\partial_t  \,.
\end{align}
To lift \eqref{eq:commuteParametricAlienDotted} to a relation at the level of formal Laurent series with exponential symbols adjoined, we define the derivative of the exponential symbol as
\begin{align} \label{eq:expSymbol}
    \partial_t \,\e^{-\omega/z} = - \dot{\omega} \,\e^{-\omega/z} \,.
\end{align}
By
\begin{align}
    \e^{-\omega/z} \left(-\frac{\dot{\omega}}{z} + \partial_t \right) = \partial_t \, \e^{-\omega/z} \,,
\end{align}
equation \eqref{eq:commuteParametricAlienUndottedFormal} then directly yields \eqref{eq:commuteParametricAlienDotted}, now applied to formal Laurent series.

Note that ensuring that a power series with $t$-dependent coefficients falls into the class of functions we are considering here is a far harder question.

\subsection{Pointed alien derivatives commute with infinitesimal translations} \label{app:pointedInfinitesimalCommutation}
The alien derivative of $\cZ$ has the form
\begin{equation}
    \dotDelta_\omega[Z(X^I)] = G[\cZ(X^I - c^I g_s)] \,,
\end{equation}
with $G$ a differential operator, cf. equations \eqref{eq:alienZmod} and \eqref{eq:alienZ}.
In Section \ref{ss:higherAlienPowers}, our argument assumes that the relation
\begin{equation}
    \dotDelta_\omega[\cZ(X^I-n c^I g_s)] \overset{?}{=} G[\cZ(X^I - (n+1)c^I g_s)] 
\end{equation}
holds. In this subsection, we will discuss the validity of this assumption.

Let us define an operator
\begin{equation} \label{eq:epsShiftOperator}
    \tau^c \,:\, \tildephi(z;t) \rightarrow \tildephi(z;t + cz ) 
\end{equation}
acting on formal Laurent series $\tildephi(z;t)$ whose coefficients depend on a parameter $t$. As in Appendix \ref{app:parametricDerivatives}, we will not introduce this space carefully: we will require that $\tildephi(z;t) \in \sRtilde^{\sr}_\Omega$ for all $t \in U$ for some $U$ of interest, that $[\partial^n_t,\sB] = 0$ for all $n \in \IN^*$ when acting on $\tildephi(z;t)$, and that a direction in the Borel plane exists along which it is Borel summable for all such $t$. Below, we will denote the Laplace transform in this direction by $\sL$.\footnote{If we want to drop this assumption, we can replace $\sL$ with the inverse Borel transform $\sB^{-1}$ in the arguments below. However, the theory loses much of its interest if Borel-Laplace resummation is not possible.} We will assume that the function space we are considering is stable under the action of $\tau^c$. Finally, our argument will require a regularity property which we introduce below, see equation \eqref{eq:regularity}.

The right-hand side of \eqref{eq:epsShiftOperator} is to be understood as the formal power series
\begin{align}
    \tildephi(z;t+cz) &= \sum_{n=0}^\infty \frac{\tildephi^{(n)}(z;t)}{n!} (cz)^n \,,
\end{align}
where $\tildephi^{(n)}$ indicates the formal power series obtained by termwise differentiation with regard to the parameter $t$. We want to argue that 
\begin{equation} \label{eq:shiftAlien}
    \tau^c \dotDelta_\omega \tildephi = \dotDelta_\omega \tau^c \tildephi \,.
\end{equation}
Let $\hatphi = \sB(\tildephi)$, and $\varphi = \sL \hatphi$. By
\begin{align}
    \sL(1*\hatphi) = z \,\varphi  \,,
\end{align}
where
\begin{align}
    (1* \hatphi)(\zeta) = \int_0^\zeta \hatphi(s)ds =: I_\zeta[\hatphi]\,,
\end{align}
we have
\begin{align}
    \tildephi^{(n)}(z;t) z^n \sim \sL(I^n_\zeta[\hatphi^{(n)}]) \,,
\end{align}
where $\sim$ indicates that both sides have the same asymptotic expansion. Our regularity assumption on $\tildephi$ is that the following relation holds true:
\begin{align} \label{eq:regularity}
    \sum_{n=0}^\infty \frac{\tildephi^{(n)}(z;t)}{n!} (cz)^n \sim \sL \left(\sum_{n=0}^\infty \frac{c^n}{n!} I^n_\zeta \partial_t^n \hatphi \right) \,.
\end{align}
We can then set
\begin{align}
    \hatphi_c(\zeta;t) = \left(\e^{cI_\zeta \partial_t} \hatphi \right)(\zeta;t) \,.
\end{align}
From this expression, we can derive a partial differential equation satisfied by $\hatphi_c(\zeta,t)$: noting that
\begin{align}
    \partial_\zeta I_\zeta^n[\hatphi] =  I_\zeta^{n-1}[\hatphi] \,, 
\end{align}
we have
\begin{align} \label{eq:master}
    \partial_c \partial_\zeta \hatphi_c(\zeta,t) = \partial_t \hatphi_c(\zeta,t) \,, \quad \hatphi_0(\zeta, t) = \hatphi(\zeta,t)\,.
\end{align}
Now suppose that $\hatphi(\zeta;t)$ has a singularity at $\zeta = \omega(t)$ of the form
\begin{equation}
    \hatphi(\zeta;t) = \frac{1}{2\pi \ii} \sum_{k=1}^n \frac{b_k a_k(t)}{u^k}+\frac{\hatphi_\omega(u,t)}{2\pi \ii} \log u + h(u,t) \,, \quad b_k = (-1)^{k-1} (k-1)! \,,
\end{equation}
with $\hatphi_\omega$ and $h(u,t)$ holomorphic in a neighborhood of $0$, where we have introduced $u = \zeta - \omega(t)$. Note that $\hatphi_\omega$ is the name of a function of $u$ and $t$, not a family of functions parametrized by $\omega(t)$. By our assumption of stability under $\tau^c$, $\hatphi^c(u,t) \in \sRhat^{\sr}_\Omega$ for all $t \in U$. We can hence make the ansatz
\begin{equation}
    \hatphi^c(u,t) = \frac{1}{2\pi \ii} \sum_{k=1}^n \frac{b_k a_k^c(t)}{u^k}+\frac{\hatphi_\omega^c(u,t)}{2\pi \ii} \log u + h^c(u,t) 
\end{equation}
to solve the differential equation \eqref{eq:master} with boundary conditions
\begin{equation} \label{eq:bchigher}
    a_k^0(t) = a_k(t) \quad \mathrm{for} \quad k=1,\ldots, n\,, \quad \hatphi^0_\omega(u,t) = \hatphi_\omega(u,t) \,.
\end{equation}
By
\begin{equation}
    \partial_\zeta \mapsto \partial_u \,, \quad 
    \partial_t \mapsto -\omega'(t) \partial_u + \partial_t =: \totalt \,,
\end{equation}
the left and right hand side of equation \eqref{eq:master} become
\begin{align}
    2\pi\ii \, \partial_c \partial_u \hatphi^c(\zeta,t) &=  \sum_{k=2}^{n+1} \frac{b_k \partial_c a^c_k(t)}{u^{k}} + \frac{\partial_c \hatphi^c_\omega(u,t)}{u} + \partial_u \partial_c \hatphi_\omega^c(u,t) \log u + 2\pi \ii \partial_c \partial_u h^c(u,t) \,, \\
    2\pi \ii \totalt \hatphi^c(\zeta,t) &=  - \frac{\omega'(t)b_{n+1}  a^c_{n+1}(t)}{u^{n+1}} +\sum_{k=2}^{n} \frac{b_k \left(-\omega'(t)+ \partial_t \right) a^c_k(t)}{u^{k}} + \frac{\partial_t a^c_1(t) -\omega'(t) \hatphi_\omega^c(u,t)}{u} +  \nonumber\\
    &+\left[\totalt \hatphi_\omega(u,t)\right] \log u + 2\pi \ii \totalt h^c(u,t)\,.
\end{align}
Equating the singular terms yields a system of $n+2$ equations,
\begin{align}
    D_c\, a_n^c(t) &= 0 \label{eq:systemhigherfirst} \,,\\
    D_c \, a_k^c(t) &= \partial_t a_k^c(t) \quad \mathrm{for} \quad k=2, \ldots, n-1 \,, \label{eq:systemnmohigher} \\
    D_c \, \hatphi_\omega^c(0,t) &= \partial_t a_1^c(t) \,, \label{eq:systempenultimatehigher}\\
    \partial_u D_c \, \hatphi_\omega^c(u,t) &= \partial_t \hatphi_\omega^c(u,t) \,,    \label{eq:systemhigherlast}
\end{align}
where we have introduced the operator
\begin{align}
    D_c = \partial_c + \omega'(t) \,.
\end{align}
Upon imposing the boundary conditions \eqref{eq:bchigher}, equation \eqref{eq:systemhigherfirst} is solved by
\begin{equation}
    a^c_n(t) = a_n(t) \e^{-\omega'(t)c} \,.
\end{equation}
We next map equation \eqref{eq:systemhigherlast} to an equation on the Laplace transform $\varphi^c_\omega(z,t)$ of $\hatphi^c_\omega(\zeta,t)$. The left-hand side becomes
\begin{align}
    \int_{\omega(t)}^\infty \e^{-\zeta/z} \partial_u D_c \hatphi^c_\omega(u,t) \, d\zeta &=  \e^{-\omega(t)/z} D_c\int_0^\infty \e^{-\zeta/z} \partial_\zeta \hatphi^c_\omega(\zeta,t) \, d\zeta  \\
    &=  \e^{-\omega(t)/z} D_c\frac{1}{z} \int_0^\infty \e^{-\zeta/z} \left( \hatphi^c_\omega(\zeta,t) - \hatphi^c_\omega(0,t) \right) \, d\zeta  \\
    &= \e^{-\omega(t)/z}  D_c\left( \frac{1}{z} \sL\left(\hatphi^c_\omega(z,t)\right) - \hatphi^c_\omega(0,t) \right)   \,,
\end{align}
while the right hand side gives rise to
\begin{align}
    \int_{\omega(t)}^\infty \e^{-\zeta/z} \partial_t  \hatphi^c_\omega(u,t) \, d\zeta &= \e^{-\omega(t)/z} \partial_t \left[\int_0^\infty \e^{-\zeta/z}  \hatphi^c_\omega(\zeta,t) \, d\zeta \right] \\
    &= \e^{-\omega(t)/z} \partial_t \sL\left(\hatphi^c_\omega(z,t)\right) \,.
\end{align}
By successively invoking equations \eqref{eq:systempenultimatehigher} and \eqref{eq:systemnmohigher} (in reverse order), we obtain
\begin{align} \label{eq:solvingIteratively}
    \left(D_c - z\partial_t \right) \sL\left(\hatphi^c_\omega(z,t)\right) = z\partial_t a_1^c(t)  &\Leftrightarrow  \left(D_c - z\partial_t \right) \left( a_1^c(t) + \sL\left(\hatphi^c_\omega(z,t)\right) \right) = D_c a_1^c(t) \\
    & \Rightarrow \left(D_c - z\partial_t \right) \left(\frac{1}{z} a_2^c(t) + a_1^c(t) + \sL\left(\hatphi^c_\omega(z,t)\right) \right) = \frac{1}{z} D_c a_2^c(t) \nonumber\\
    & \;\;\vdots \nonumber\\
    & \Rightarrow \left(D_c - z\partial_t \right) \left(\sum_{n=1}^\infty \frac{1}{z^{n-1}} a_n^c(t) + \sL\left(\hatphi^c_\omega(z,t)\right) \right) = 0 \,. \nonumber
\end{align}
Noting that
\begin{equation} \label{eq:almost}
    \e^{-\omega(t)/z} \left( D_c - z\partial_t \right) = \left(\partial_c - z \partial_t \right) \e^{-\omega(t)/z} \,,
\end{equation}
we see that the last equation of \eqref{eq:solvingIteratively} is equivalent to
\begin{align}
\left(\partial_c - z\partial_t \right) \e^{-\omega(t)/z}\left(\sum_{n=1}^\infty \frac{1}{z^{n-1}} a_n^c(t) + \sL\left(\hatphi^c_\omega(z,t)\right) \right) = 0
\end{align}
The general form of the solution to this equation is
\begin{equation}
   \e^{-\omega(t)/z}\left(\sum_{n=1}^\infty \frac{1}{z^{n-1}} a_n^c(t) + \sL\left(\hatphi^c_\omega(z,t)\right) \right) = g(z,t + cz) \,.
\end{equation}
Invoking the boundary condition \eqref{eq:bchigher} fixes
\begin{align}
    g(z,t) =  \e^{-\omega(t)/z}\left(\sum_{n=1}^\infty \frac{1}{z^{n-1}} a_n(t) + \sL\left(\hatphi_\omega(z,t)\right) \right) \,.
\end{align}
By performing an asymptotic expansion, we have thus established
\begin{align}
    \dotDelta_{\omega(t)} \tildephi^c(z,t) = \left( \dotDelta_{\omega(t)} \tildephi \right) (z,t+cz) \,.
\end{align}

\bibliographystyle{JHEP}
\bibliography{biblio}
\end{document}